\documentclass[11pt, a4paper]{article}
\usepackage{url}
\usepackage{graphicx}
\usepackage{natbib}
\usepackage{hyperref}
\usepackage[labelfont=bf]{caption}
\usepackage[blocks]{authblk}
\usepackage{setspace}
\usepackage{multirow}
\usepackage{amsmath, amsthm, amssymb, amsfonts}
\usepackage{bm}
\usepackage{paralist}
\usepackage{color}     
\usepackage{changepage}
\usepackage{enumitem}
\usepackage{ulem}
\usepackage[flushleft]{threeparttable}
\usepackage{booktabs,rotating}
\usepackage{longtable}
\usepackage{geometry}

\DeclareMathAlphabet\mathbfcal{OMS}{cmsy}{b}{n}

\newcommand{\mbf}{\mathbf}
\newcommand{\mc}{\mathcal}

\newcommand{\bit}{\begin{itemize}}
\newcommand{\eit}{\end{itemize}}
\newcommand{\ben}{\begin{enumerate}} 
\newcommand{\een}{\end{enumerate}}
\newcommand{\bpm}{\begin{pmatrix}}
\newcommand{\epm}{\end{pmatrix}}
\newcommand{\bbm}{\begin{bmatrix}}
\newcommand{\ebm}{\end{bmatrix}}

\newcommand{\vep}{\varepsilon}

\renewcommand{\l}{\left}
\renewcommand{\r}{\right}

\def\wh{\widehat}
\def\wt{\widetilde}

\newcommand{\E}[0]{\mathsf{E}}
\newcommand{\Var}[0]{\mathsf{Var}}

\newcommand{\diag}[0]{\mathsf{diag}}

\newcommand{\p}{\mathsf{P}}
\newcommand{\R}{\mathbb{R}}
\newcommand{\Z}{\mathbb{Z}}

\newcommand{\iid}{\text{\upshape iid}}
\newcommand{\pca}{\text{\upshape pc}}
\newcommand{\bpca}{\text{\upshape bpc}}
\newcommand{\capp}{\text{\upshape cp}}
\newcommand{\bcapp}{\text{\upshape bcp}}
\newcommand{\shr}{\text{\upshape sh}}
\newcommand{\cs}{\text{\upshape sh}}
\newcommand{\bcs}{\text{\upshape bsh}}
\newcommand{\sca}{\text{\upshape sc}}
\newcommand{\bsca}{\text{\upshape bsc}}
\newcommand{\oracle}{\text{\upshape oracle}}
\newcommand{\poet}{\text{\upshape poet}}
\newcommand{\avg}{\text{\upshape avg}}
\newcommand{\nn}{\nonumber}

\theoremstyle{definition}

\theoremstyle{definition}

\theoremstyle{definition}
\newtheorem{lem}{Lemma}
\theoremstyle{definition}
\newtheorem{prop}{Proposition}
\theoremstyle{definition}
\newtheorem{assum}{Assumption}
\theoremstyle{definition}
\newtheorem{rem}{Remark}
\theoremstyle{definition}

\theoremstyle{definition}
\newtheorem{ex}{Example}

\title{\Large Consistent estimation of high-dimensional factor models when the factor number is over-estimated}
\author{\normalsize Matteo Barigozzi$^1$ \hskip 1cm Haeran Cho$^2$}\vskip -3cm
\date{\small{\today}}

\begin{document}

\maketitle
\begin{abstract} 
A high-dimensional $r$-factor model for an $n$-dimensional vector time series is characterised by the presence of a large eigengap (increasing with $n$) between the $r$-th and the $(r+1)$-th largest eigenvalues of the covariance matrix. Consequently, Principal Component (PC) analysis is the most popular estimation method for factor models and its consistency, when $r$ is correctly estimated, is well-established in the literature. However, popular factor number estimators often suffer from the lack of an obvious eigengap in empirical eigenvalues and tend to over-estimate $r$ due, for example, to the existence of non-pervasive factors affecting only a subset of the series. We show that the errors in the PC estimators resulting from the over-estimation of $r$ are non-negligible, which in turn lead to the violation of the conditions required for factor-based large covariance estimation. To remedy this, we propose new estimators of the factor model based on scaling the entries of the sample eigenvectors. We show both theoretically and numerically that the proposed estimators successfully control for the over-estimation error, and investigate their performance when applied to risk minimisation of a portfolio of financial time series.
\end{abstract}

\footnotetext[1]{Department of Economics, Universit\`a di Bologna, Piazza Scaravilli 2, 40126, Bologna, Italy.\\
Email: \url{matteo.barigozzi@unibo.it}.} 

\footnotetext[2]{School of Mathematics, University of Bristol, Bristol, BS8 1TH, UK.\\
Email: \url{haeran.cho@bristol.ac.uk}.}


\section{Introduction}

Factor modelling is a popular approach to dimension reduction 
in high-dimensional time series analysis.
It has been successfully applied to large panels of time series 
for forecasting macroeconomic variables \citep{stockwatson02JASA}, 
building low-dimensional indicators of the whole economic activity \citep{stockwatson02JBES}
and analysing dynamic brain connectivity using high-dimensional fMRI data \citep{ting2017},
to name a few.

In this paper, we consider one of the most general factor models in the literature,
the approximate dynamic factor model, which 
permits serial dependence in the factors and 
both serial and cross-sectional dependence among the idiosyncratic components. 
More specifically, given an $n$-dimensional vector time series
$\{\mbf x_t = (x_{1t}, \ldots, x_{nt})^\top, \, 1 \le t \le T\}$, 
we investigate the problem estimating the factor model
\begin{align}
\label{eq:model}
x_{it} = \bm\lambda_i^\top\mbf f_t+\vep_{it}, 
\end{align}
where $\bm\lambda_i$ and $\mbf f_t$ are $r$-dimensional vectors of loadings and factors, respectively.
We refer to $\chi_{it} = \bm\lambda_i^\top\mbf f_t$ as the common component and 
$\vep_{it}$ as the idiosyncratic component,
and assume the number of factors, $r$, to be fixed independent of $n$ and~$T$.

The main assumption that guarantees the asymptotic identification under \eqref{eq:model} is 
the existence of a \textit{large} (increasing with $n$) eigengap 
between the $r$ leading eigenvalues of the covariance matrix of $\mbf x_t$
and the remaining ones. Intuitively, since the eigengap is assumed to increase with $n$, 
the more series are pooled together, 
the more the contribution of the factors to the total co-variation in the data 
is likely to emerge over the idiosyncratic components (`blessing of dimensionality'). 
As a consequence, a natural way of estimating \eqref{eq:model} 
is via Principal Component (PC) analysis, 
through which the common components are estimated
as the projection of the data onto the space spanned by 
the leading eigenvectors of the sample covariance matrix, 
i.e., given some estimator $\wh r$ of the factor number $r$, 
the PC estimator of the common component is defined as
\begin{align}
\wh{\chi}_{it}^{\pca}=\sum_{j=1}^{\wh r}\wh{w}_{x,ij}\wh{\mbf w}_{x,j}^\top \mbf x_t ,\label{eq:pca:est}
\end{align}
where $\wh{\mbf w}_{x,j}=(\wh{w}_{x,1j},\ldots, \wh{w}_{x,nj})^\top$ is the normalised eigenvector corresponding to the $j$-th largest eigenvalue of the sample covariance matrix of $\mbf x_t$. The PC estimator \eqref{eq:pca:est} 
allows for consistent estimation of 
the common component of model \eqref{eq:model}, provided that both $n, T \to \infty$  (see \citealp{bai2003}, and \citealp{fan13}).

However, the theoretical properties of PC estimators have always been investigated 
conditional on $\wh r$ being a consistent estimator of $r$,
and the problem of determining $r$ has typically been treated separately. 
Many methods exist for estimating the factor number: \cite{baing02}, \cite{ABC10},
see \cite{onatski10}, \cite{ahn2013}, \cite{yu2018}, \citet{trapani2018}, and \cite{bn2017}, 
to name a few, all of which exploit the postulated existence of the eigengap. 
On the other hand, it is often difficult to identify the large gap
from empirical eigenvalues.
In particular, it is known that the presence of moderate cross-sectional correlations
in the idiosyncratic components shrinks the empirical eigengap 
by introducing some so-called `weak' factor \citep{onatski12},
and we empirically demonstrate that 
commonly adopted factor number estimators 
often over-estimate $r$ in such situations.
Moreover, as noted in \cite{bcf2017}, instabilities in the factor structure
tend to spuriously enlarge the factor space and introduce further difficulties to 
determining the number of factors.
Finally, as shown later in the paper, different estimators frequently return discordant results, 
thus making it ambiguous for the user to choose a single value to rely on. 

\subsection{Our contributions}
\label{sec:contribution}

The question is, what do we do if we have a range of possible candidate estimators of $r$,
or if we believe that none of the estimators is reliable?
One may use the largest number of factors 
returned by available methods, or set it to be even larger, 
with the expectation of avoiding the hazard of under-estimating
the factor-driven variation, which is a problem without any clear solution. 

In this paper, we first show that over-estimation of $r$ can incur non-negligible estimation error
when considering the worst case scenarios for individual common components (see Proposition~\ref{prop:pca}). 
To the best of our knowledge, this problem has not been investigated in the literature before. 
Identifying the theoretical difficulties arising under the time series factor model, 
we propose a novel blockwise estimation technique that enables
rigorous treatment of the PC-based estimators which is another contribution made in this paper.

In order to mitigate the lack of a reliable estimator of $r$, 
we propose a modified PC estimator
which performs as well as the `oracle' estimator constructed {\it with the knowledge of true~$r$}
and, consequently, makes our estimation procedure free from 
the difficult task of estimating $r$ accurately.

More specifically, the factor model \eqref{eq:model} is usually characterised 
by the following eigengap conditions (see e.g. \citealp{fan13}):
\begin{enumerate}[label=(C\arabic*)]
\setlength\itemsep{0em}
\item \label{eq:c1} there exist some fixed $\underline c_j, \bar{c}_j$ such that for $1\le j\le r$,
\begin{align}\nn
0 < \underline c_j <\lim_{n\to\infty}\!\!\inf \frac {\mu_{\chi, j}}{n} \le 
\lim_{n\to\infty}\!\!\sup \frac{\mu_{\chi, j}}{n} < \bar{c}_j < \infty
\end{align}
and $\bar{c}_{j+1} < \underline c_j$ for $j \le r - 1$,
\item  \label{eq:c2} $\mu_{\vep, 1} < C_{\vep} < \infty$ for any $n$,
\end{enumerate}
where $\mu_{\chi, j}$ and $\mu_{\vep, j}$ denote the $j$-th largest eigenvalues of 
the covariance matrices of the common and idiosyncratic components, respectively. 
The linear divergence of eigenvalues in \ref{eq:c1} is a prevailing and natural assumption
in the factor model literature, implying that all series in the panel 
are equally important for the recovery of the factors.
From \ref{eq:c1}, it follows that 
$\mbf w_{\chi, j} = (w_{\chi, 1j}, \ldots, w_{\chi, nj})^\top$,
the normalised eigenvector of the covariance matrix of $\bm\chi_t$ corresponding to $\mu_{\chi, j}$,
has its coordinates asymptotically bounded as 
$\max_{1 \le i \le n} |w_{\chi, ij}| = O(n^{-1/2})$ for all $j \le r$ 
(see~\eqref{eq:bounded:w} below).
Thanks to the eigengap and the Davis-Kahan theorem \citep{yu15},
the coordinates of 
the $r$ leading eigenvectors of the sample covariance matrix of the data, $\wh{\mbf w}_{x, j}, \, j \le r$, are
also bounded asymptotically as
$\max_{1 \le j \le r} \max_{1 \le i \le n} |\wh w_{x, ij}| = O_p(n^{-1/2})$.
On the other hand, precisely due to the lack of this eigengap,
meaningful control of the behaviour of $\wh{w}_{x, ij}, \, j \ge r+1$ is 
not obvious under the dynamic factor model in \eqref{eq:model}. 

Motivated by these observations, we propose to modify $\wh{\mbf w}_{x, j}$ via {\it scaling} as
\begin{align}
\label{eq:w:scale:one}
\wh{\mbf w}^{\sca}_{x, j} = \nu_j^{-1} \; \wh{\mbf w}_{x, j} \quad \text{with} \quad
\nu_j = \max\{1, \delta_n^{-1} \, \max_{1 \le i \le n} |\wh w_{x, ij}|\},
\end{align}
which ensures that the entries of the modified eigenvectors
are bounded by some $\delta_n$ of order $n^{-1/2}$.
By substituting $\wh{\mbf w}^{\sca}_{x, j}$ in place of $\wh{\mbf w}_{x, j}$ in \eqref{eq:pca:est},
we obtain a novel \textit{scaled} PC estimator of the common component. 
While conceptually and computationally simple, 
the scaled PC estimator attains the same asymptotic error bound as
the oracle PC estimator obtained with the true $r$, 
successfully curtailing the error attributed to spurious factors
without requiring the accurate estimation of the factor number
beyond that $\wh r \ge r + 1$.
We also propose a well-motivated choice of the tuning parameter $\delta_n$.

The good performance of the scaled PC estimator when $r$ is over-estimated, 
in contrast to that of the PC estimator, is demonstrated on simulated datasets.
In addition, we investigate the impact of the non-negligible errors in the PC estimator
(or lack thereof in the modified PC estimator) on large covariance matrix estimation
through an application to risk minimisation of a portfolio of financial time series.

\subsection{Relationship to the existing literature}

Recently, \cite{bn2017} adopted the eigenvalue shrinkage for minimum-rank factor analysis
under time series factor models.
Our approach is distinguished from theirs in 
that we aim at avoiding the reliance on the accurate estimation of the factor number itself
in establishing the theoretical consistency of the estimator of common components.
Sharing the aim closer to ours, \cite{fan2019} propose a diversified factor estimator obtained as 
cross-sectional averages of the data with respect to pre-determined weights and
show their robustness to over-estimating the number of factors.

We mention two other approaches to time series factor analysis for which our work can be relevant.
First, assuming that all serial dependence in the data is captured by the factors, 
\citet{lam2011} and \citet{lam2012} proposed an alternative approach to factor model analysis.
Since their method is also based on eigenanalysis of a suitable covariance matrix, 
our methodology can be readily adapted to this case as well. 
Second, \citet{FHLR00} considered a richer factor structure where factors 
are allowed to have lagged effects on the data. 
Estimation of such model is in general based on spectral PC analysis, but
other approaches exist that require standard PC analysis at the initial or final step 
(e.g., \citealp{FHLR05}, \citealp{baing07}, \citealp{FGLR09}, and \citealp{doz2011}),
and our proposed modifications can be easily adopted for this purpose.

Finally, we note that there are some links between the model and the estimators proposed here
and the vast literature on statistical models and methods based on random matrix theory,
see \cite{el2008}, \cite{cai2013}, \cite{donoho2018a} and \cite{donoho2018b},
and also \cite{paul2014} for an overview.

\subsubsection*{Structure of the paper}

The rest of the paper is organised as follows.
We introduce the approximate factor model in Section~\ref{sec:fm},
where we also discuss its estimation via PC, and we
investigate the behaviour of factor number estimators
as well as the impact of over-estimating the factor number on the PC estimator.
In Section~\ref{sec:mod}, we motivate and introduce the modified PC estimator
based on scaling, and study its theoretical properties.
Comparative simulation study of PC-based estimators
is conducted in Section~\ref{sec:sim}, 
and we apply the proposed estimators to financial data analysis in Section~\ref{sec:real}.
All the proofs of the theoretical results and further simulation results are provided in Appendix.

\subsubsection*{Notation}

For a given $m \times n$ matrix $\mbf B$ with $b_{ij}$ denoting its $(i, j)$ element, 
its spectral norm is defined as $\Vert\mbf B\Vert= \sqrt{\mu_1(\mbf B\mbf B^\top)}$,
where $\mu_k(\mbf C)$ denotes the $k$-th largest eigenvalue of $\mbf C$, 
its Frobenius norm as $\Vert\mbf B\Vert_F=\sqrt{\sum_{i=1}^m\sum_{j=1}^n b_{ij}^2}$, 
and also $\Vert\mbf B\Vert_{\max}=\max_{1\le i\le m}\max_{1\le j\le n} |b_{ij}|$.
The sub-exponential norm of a random variable $X$ is defined as
$\Vert X \Vert_{\psi_1}=\inf_k\{k : \E[\exp(|X|/k)]\le 2\}$. 
For a given set $\Pi$, we denote its cardinality by $|\Pi|$.
For any vector $\mbf a = (a_1, \ldots, a_m) \in \R^m$, we denote
$\Vert \mbf a \Vert_0 = \vert \{1 \le i \le m: \, a_i \ne 0 \} \vert$
and $\Vert \mbf a \Vert_\infty = \max_{1 \le i \le m} \vert a_i \vert$.
Also, we use the notations $a \vee b = \max(a, b)$ and $a \wedge b = \min(a, b)$.
The notation $a \asymp b$ indicates that $a$ is of the order of $b$, and 
$a \gg b$ indicates that $a^{-1} b \to 0$.
We denote an $m \times m$-identity matrix by~$\mbf I_m$.

\section{The approximate dynamic factor model}
\label{sec:fm}

\subsection{Model and assumptions}

Recall the factor model in \eqref{eq:model},
where an $n$-dimensional vector time series $\mbf x_t = (x_{1t}, \ldots, x_{nt})^\top$ 
is divided into the common component
$\bm\chi_t = (\chi_{1t}, \ldots, \chi_{nt})^\top = \bm\Lambda \mbf f_t$
driven by the vector of $r$ latent factors $\mbf f_t = (f_{1t}, \ldots, f_{rt})^\top$,
with $\bm\Lambda = [\bm\lambda_1, \ldots, \bm\lambda_n]^\top$ 
as the $n \times r$ matrix of loadings,
and the idiosyncratic component
$\bm\vep_t = (\vep_{1t}, \ldots, \vep_{nt})^\top$. 
Without loss of generality, we assume $\E(f_{jt}) = \E(\vep_{it}) = 0$ for all $i, j, t$. 

We now list and motivate the assumptions imposed on the approximate dynamic factor model \eqref{eq:model} (see e.g., \cite{fan13} and \cite{bcf2017} for similar conditions).

\begin{assum}[Identification]
\label{assum:id} \hfill
\begin{compactenum}
\item[(i)] $\E(\mbf f_t\mbf f_t^\top) = \mbf I_r$ for all $t \ge 1$.
\item[(ii)] There exists a positive definite $r \times r$ matrix $\mbf H$ with distinct eigenvalues and 
such that $n^{-1}\bm\Lambda^\top\bm\Lambda \to \mbf H$ as $n \to \infty$. 
\item[(iii)] There exists $\bar{\lambda} \in (0, \infty)$ such that 
$\Vert\bm\Lambda\Vert_{\max}< \bar{\lambda}$.
\item[(iv)] There exists $C_\vep \in (0, \infty)$ such that, for any $t \ge 1$,
$\sum_{i=1}^n\sum_{i'=1}^n a_ia_{i'}\E(\vep_{it}\vep_{i't}) < C_\vep$
for any sequence of coefficients $\{a_i\}_{i=1}^n$ satisfying $\sum_{i=1}^n a_i^2 = 1$.
\item[(v)] $\E(f_{jt}\vep_{it'}) = 0$ for all $i \le n$, $j \le r$ and $t, t' \le T$.
\end{compactenum}
\end{assum}

We adopt the normalisation given in Assumption~\ref{assum:id}~(i)--(ii) for the purpose of identification; 
in general, factors and loadings are recoverable up to a linear invertible transformation only.
Assumption~\ref{assum:id}~(iii) is a commonly found assumption
in the factor model literature (see Assumption~B in \cite{bai2003}) which,
together with Assumption~\ref{assum:id}~(ii),
requires that factors influence all cross-sections to a similar degree.
Assumption~\ref{assum:id} (iv) allows for mild cross-sectional dependence 
across idiosyncratic components. 
In other words, we are considering an {\it approximate} factor structure, 
as opposed to the classical {\it exact} factor model where 
$\bm\vep_t$ is assumed to be uncorrelated cross-sectionally.  
It is possible to relax Assumption~\ref{assum:id}~(v) and allow for weak dependence 
between the factors and the idiosyncratic components 
(c.f. Assumption~D of \citealp{baing02}).

In order to motivate the assumptions further, we adopt the notations
\begin{align*}
\bm\Gamma_\chi= \bm\Lambda\l(\frac{1}{T}\sum_{t = 1}^T \E(\mbf f_t\mbf f_t^\top)\r)\bm\Lambda^\top = \bm\Lambda\bm\Lambda^\top,
\quad
\bm\Gamma_\vep = \frac{1}{T}\sum_{t = 1}^T \E(\bm\vep_t\bm\vep_t^\top),\;
\text{ and }\;
\bm\Gamma_x = \bm\Gamma_\chi  + \bm\Gamma_\vep.
\end{align*}
If $\mbf f_t$ and $\bm\vep_t$ are covariance stationary, 
then these matrices are the corresponding population covariance matrices.
Also, we denote the eigenvalues (in non-increasing order) of 
$\bm\Gamma_\chi$, $\bm\Gamma_\vep$ and $\bm\Gamma_x$ by 
$\mu_{\chi, j}$, $\mu_{\vep, j}$ and $\mu_{x, j}$, respectively.
Then, Assumption~\ref{assum:id} leads to
\ref{eq:c1}--\ref{eq:c2} in Section~\ref{sec:contribution},
i.e., $\mu_{\chi, j}, \, j \le r$ diverge linearly in $n$ as $n \to \infty$,
whereas $\mu_{\vep, 1}$ is bounded for any $n$.
This condition coincides with Definition~2 in \cite{chamberlain1983} and Assumption~2 in \citet{fan13}, 
and it is also in the same spirit as Assumption~C.4 in \cite{bai2003}
where cross-sectional dependence of idiosyncratic components is assumed to be weak.

Moreover, \ref{eq:c1}--\ref{eq:c2} imply that, due to Weyl's inequality, 
the eigenvalues of $\bm\Gamma_x$, $\mu_{x, j}$, satisfy the following eigengap conditions:
\begin{enumerate}[label=(C\arabic*), start=3]
\setlength\itemsep{0em}
\item \label{eq:c3} The $r$ largest eigenvalues, $\mu_{x, 1}, \ldots, \mu_{x, r}$, 
diverge linearly in $n$ as $n \to \infty$;
\item  \label{eq:c4} the $(r+1)$-th largest eigenvalue, $\mu_{x, r+1}$, stays bounded for any $n$.
\end{enumerate}
From \ref{eq:c1}--\ref{eq:c4} above, it is clear that for consistent estimation of the common components, 
approximate factor models need to be considered in the asymptotic limit where $n \to \infty$, 
i.e., these models enjoy what is sometimes referred to as the blessing of dimensionality. 
In particular, we require:
\begin{assum}
\label{assum:nt}
$n \to \infty$ as $T \to \infty$, with $n = O(T^\kappa)$ for some $\kappa \in (0, \infty)$.
\end{assum}
Under Assumption~\ref{assum:nt}, we operate in a high-dimensional setting
that permits $n \gg T$, unlike in the random matrix theory literature 
where it is typically assumed that $n/T \to y \in (0, \infty)$ \citep{johnstone2001}.
Furthermore we assume:
\begin{assum}[Tail behaviour]
\label{assum:tail} \hfill
\begin{compactenum}
\item[(i)] $\max_{1 \le j \le r} \max_{1 \le t \le T} \Vert f_{jt} \Vert_{\psi_1} < B_f$ for some $B_f \in (0, \infty)$.

\item[(ii)] 
$\max_{1 \le t \le T} \Vert \bm\vep_t \Vert_{\psi_1} < B_\vep$ for some $B_\vep \in (0, \infty)$,
where 
$\Vert \bm\vep_t \Vert_{\psi_1} = \sup_{\mbf v \in \mc \R^n: \, \Vert \mbf v \Vert = 1} 
\Vert \mbf v^\top\bm\vep_t \Vert_{\psi_1}$.
\end{compactenum}
\end{assum}

\begin{assum}[Strong mixing]
\label{assum:dep} 
Denoting the $\sigma$-algebra generated by $\{(\mbf f_t, \bm\vep_t), \, s \le t \le e\}$ 
by $\mc F_s^e$, let
$\alpha(k) = \max_{1 \le t \le T} \sup_{{A \in \mc F_{-\infty}^t, B \in \mc F_{t+k}^\infty}}
|\p(A)\p(B) - \p(A \cap B)|$.
Then, there exist some fixed $c_\alpha, \beta \in (0, \infty)$, 
such that $\alpha(k) \le \exp(-c_\alpha k^\beta)$ for all $k, T \in \Z^+$.
\end{assum}

The sub-exponential tail conditions in Assumption~\ref{assum:tail},
along with Assumption~\ref{assum:dep}, 
allow us to control the deviation of sample covariance estimates from their population counterparts
via Bernstein-type inequality (see Theorem 1 of \citealp{merlevede2011})
under the approximate dynamic factor model. 
We stress that either strict or weak stationarity of $f_{jt}$ and $\vep_{it}$ is not required
in performing the PC-based estimation,
provided that the loadings are time-invariant.

\subsection{Estimation via Principal Component Analysis}
\label{sec:estimation}

The most common way to estimate the approximate factor model \eqref{eq:model} is 
by means of PC analysis, and the asymptotic properties of the PC estimator have been well-established: 
in particular, we refer to \citet{fan13} where a set-up similar to ours is considered. 
 
Recall that the PC estimator of the common component: $\wh\chi_{it}^{\pca} 
= \sum_{j=1}^{\wh r} \wh w_{x, ij} \wh{\mbf w}_{x, j}^\top\mbf x_t$,
where $\wh{\mbf w}_{x, j}$ denote the $j$-th leading normalised eigenvector 
of the sample covariance 
$\wh{\bm\Gamma}_x = T^{-1} \sum_{t = 1}^T \mbf x_t \mbf x_t^\top$,
and $\wh r$ is an estimator of the number of factors $r$.
Theorem~1 of \cite{bcf2017}, 
which is a refinement of Corollary~1 of \citet{fan13}, establishes a uniform bound
on the estimation error over $1 \le i \le n$ and $1 \le t \le T$
of the PC estimator when $r$ is {\it known}, i.e., $\wh r = r$, 
under Gaussianity of the idiosyncratic component.
Here, we generalise the theorem to the case of 
sub-exponential distributions as specified in Assumption~\ref{assum:tail}. 
Its proof can be found in Appendix~\ref{sec:pf:prop:one}.

\begin{prop}
\textit{\label{thm:common}
Under Assumptions~\ref{assum:id}--\ref{assum:dep}, 
the PC estimator $\wh\chi_{it}^{\pca}$ with $\wh r = r$ satisfies
\begin{align}\nn
\max_{1 \le i \le n} \max_{1 \le t \le T}
\vert \wh{\chi}^{\pca}_{it}-\chi_{it}\vert = 
O_p\l\{\Big(\sqrt{\frac{\log(n)}{T}} \vee \frac{1}{\sqrt n}\Big)\log(T)\r\}.
\end{align}
}
\end{prop}

Two key results are required for proving Proposition~\ref{thm:common}.
First, we make use of the eigengap between $\mu_{x, r}$ and $\mu_{x, r+1}$
increasing linearly in $n$ (see \ref{eq:c3}--\ref{eq:c4}),
which ensures that the eigenspace of ${\bm\Gamma}_{\chi}$
is consistently estimated by the $r$ leading eigenvectors of $\wh{\bm\Gamma}_x$. 
More specifically, 
there exists a diagonal $r\times r$ matrix $\mbf S$ with entries $\pm 1$, such that
\begin{align}
\label{eq:davisk}
\Vert \wh{\mbf W}_{x}-  \mbf W_{\chi}\mbf S \Vert \le 
\frac{2^{3/2}\sqrt{r}\Vert\wh{\bm\Gamma}_x-\bm\Gamma_{\chi}\Vert}{\mu_{\chi, r}}
= O_p\l( \sqrt{\frac{\log(n)}{T}} \vee \frac{1}{n}\r), 
\end{align}
where $\wh{\mbf W}_{x} = [\wh{\mbf w}_{x, j}, \, j \le r]$ and
$\mbf W_\chi = [\mbf w_{\chi, j}, \, j \le r]$.
The result in~\eqref{eq:davisk} follows 
from the modified Davis-Kahan theorem of \cite{yu15},
the lower bound of $\underline c_r n$ on $\mu_{\chi, r}$ from~\ref{eq:c1},
and the closeness between $\wh{\bm\Gamma}_x$ and $\bm{\Gamma}_x$ 
under Assumptions~\ref{assum:nt}--\ref{assum:dep} (see Lemma~\ref{lem:cov}~(i) in Section~\ref{sec:prem}).
We can further show that
\begin{align}
\label{eq:evec:consist}
\sqrt n\,\l\Vert\bm\varphi_i^\top(\wh{\mbf W}_x-\mbf W_{\chi}\mbf S)\r\Vert
= O_p\l(\sqrt{\frac{\log(n)}{T}} \vee \frac{1}{\sqrt n}\r),
\end{align}
where $\bm\varphi_i$ an $n$-vector of zeros except for its $i$-th element being one,
see Lemma~\ref{lem:cov}~(iii).

Secondly, denoting the eigendecomposition of the covariance matrix of the common components by
$\bm\Gamma_\chi = \mbf W_\chi\mbf M_\chi\mbf W_\chi^\top$
with $\mbf M_\chi = \diag(\mu_{\chi, 1}, \ldots, \mu_{\chi, r})$,
\ref{eq:c1} leads to
\begin{align}
\label{eq:bounded:w}
\max_{1 \le i \le n} \sqrt{\sum_{j=1}^r w_{\chi,ij}^2}=
\max_{1 \le i \le n} \Vert\bm\varphi_i^\top\mbf W_\chi\Vert \le
\max_{1 \le i \le n} \Vert\bm\varphi_i^\top\bm\Gamma_{\chi}\Vert\, \Vert \mbf W_\chi\Vert \, 
\Vert \mbf M_\chi^{-1}\Vert
= O\l(\frac{1}{\sqrt n}\r),
\end{align}
i.e., asymptotically, each element of $\mbf W_{\chi}$ is $O(n^{-1/2})$. 
This, combined with \eqref{eq:evec:consist},
leads to 
\begin{align}
\label{eq:bounded:ew}
\max_{1 \le i \le n}  \, \max_{1 \le j \le r}  \, |\wh w_{x, ij}| = O_p\l(\frac{1}{\sqrt n}\r).
\end{align}
The bound in~\eqref{eq:bounded:ew} serves as the main motivation 
behind introducing the modified PC estimators in Section~\ref{sec:mod}.

\begin{rem}[Optimality of PC]
\label{rem:pca:cai}
The PC estimator is appealing for the following reasons. 
First, under the assumption of spherical idiosyncratic components, 
$\bm\vep_t \sim_{\iid} \mc N_n(\mbf 0, \sigma^2 \mbf I_n)$ for some $\sigma^2 > 0$, 
the PC estimator of the loadings is asymptotically equivalent to 
their Maximum Likelihood estimator \citep{tipping1999}. 
Second, the sample principal subspace estimator is minimax rate optimal,
see Theorem~5 of \cite{cai2013} which shows that
$\E\Vert \wh{\mbf W}_x\wh{\mbf W}_x^\top - \mbf W_x\mbf W_x^\top \Vert^2_F
\asymp rn/(\mu_{\chi, r} T)$.
This, combined with \ref{eq:c1}, is comparable 
to the convergence rate reported in \eqref{eq:davisk},
although the latter is obtained under the more general approximate dynamic factor model.
Third, when allowing for non-spherical and possibly correlated idiosyncratic components, 
the PC estimator by definition delivers the linear combination of the data with largest variance
in the sense that, for the $j$-th PC, 
$\wh{\Var}(\wh{\mbf w}_{x,j}^\top\mbf x_t)\ge \wh{\Var}(\bm\omega^\top\mbf x_t)$
for any $\bm\omega$ satisfying $\Vert \bm\omega\Vert = 1$ and $\bm\omega^\top\wh{\mbf w}_{x,j^\prime}=0$
for any $j^\prime \le j - 1$,
where $\wh{\Var}(\cdot)$ denotes the sample variance operator.
\end{rem}

\subsection{(Over-)estimation of the number of factors}
\label{sec:over}

In practice, the true number of factors $r$ is unknown and 
its estimation has been one of the most researched problems 
in the factor model literature (see the references in the Introduction).
Based on the conditions~\ref{eq:c3}--\ref{eq:c4}, 
a prevailing approach is to identify a `large' gap
between the successive estimated eigenvalues 
$\wh{\mu}_{x, j}, \, 1 \le j \le r_{\max}$
of the sample covariance matrix $\wh{\bm\Gamma}_x$, 
where $r_{\max}$ denotes the maximum allowable number of factors 
often required as an input parameter to the estimation procedure. 
Here we focus on two of the most popular methods.

The information criterion-based method proposed by \cite{baing02} estimates $r$ as
\begin{align}
\label{eq:r:bn}
\wh r = \arg\!\!\!\!\!\!\min_{1 \le q \le r_{\max}} \text{IC}(q), \mbox{ where } 
\text{IC}(q) = \log\l(\frac 1n \sum_{j = q+1}^n \wh{\mu}_{x, j}\r) + q \cdot g(n, T),
\end{align}
with a penalty function $g(n, T)$ satisfying $g(n, T) \to 0$ and 
$\{(n \wedge T) \cdot g(n, T)\} \to \infty$ as $n, T \to \infty$.
The eigenvalue ratio-based estimator by \cite{ahn2013}, returns
\begin{align}
\label{eq:r:ah}
\wh r = \arg\!\!\!\!\!\!\max_{1 \le q \le r_{\max}} \text{GR}(q), \mbox{ where } 
\text{GR}(q) = \frac{\log(1+\wh\mu_{x, q}^*)}{\log(1+\wh\mu_{x, q+1}^*)} \mbox{ with }
\wh\mu_{x, q}^* = \frac{\wh\mu_{x, q}}{\sum_{j = q+1}^n \wh\mu_{x, j}}.
\end{align}
Implicitly, the information criterion in \eqref{eq:r:bn} performs thresholding on 
the scaled sample eigenvalues $\wh\mu_{x, q}^*$ with respect to $g(n, T)$,
and selects an index $q$ among those that correspond to $\wh\mu_{x, q}^*$ surviving the thresholding.
On the other hand, the eigenvalue ratio approach in \eqref{eq:r:ah}
considers the ratio of the successive scaled eigenvalues 
without taking into account the size of the eigenvalues.
This difference frequently leads to distinct estimators from the different approaches,
not to mention that, as shown in \citet{ABC10}, the various choices of $g(n, T)$ 
often result in different factor number estimators.
Another parameter whose choice may affect the estimation result for 
both of the estimators \eqref{eq:r:bn}--\eqref{eq:r:ah} is $r_{\max}$.
Moreover, while \ref{eq:c3}--\ref{eq:c4} are asymptotic conditions,
the lack of an obvious eigengap in the empirical eigenvalues
poses a challenge in the estimation of $r$. 
Consequently, the estimated number of factors is highly variable as the following quantities vary:
the dimensions $n$ and $T$,
the degree of cross-sectional correlations in the idiosyncratic components, 
and the signal-to-noise ratio 
represented by the ratio between $\Var(\chi_{it})$ and $\Var(\vep_{it})$,
see e.g., the numerical studies in
\cite{ahn2013} and \cite{trapani2018}.

For an illustration, we conduct a comparative simulation study
by applying the two estimators \eqref{eq:r:bn} 
(with $g(n, T) = (n+T)\log(n \wedge T)/(nT)$, i.e., $IC_2$ of \cite{baing02}) and \eqref{eq:r:ah}
with a generous but reasonable choice $r_{\max} = [\sqrt{n \wedge T}]$,
to datasets simulated under Model~1 as described in Section~\ref{sec:sim}.
The results are reported in Figure~\ref{fig:sim:r}.
It is apparent that the estimator \eqref{eq:r:bn} fails to return the true number of factors $r = 5$
in the presence of moderate degree of cross-sectional correlations in $\bm\vep_t$,
especially when $n$ is small.
While~\eqref{eq:r:ah} performs considerably better for this particular data generating process,
we provide in Appendix~\ref{sec:est:r} the scenarios 
where this method also fails to return the correct number of factors.
We note that the performance of the estimators does not improve with increasing sample size $T$. 
In almost all cases considered, the factor number is over-estimated, i.e., $\wh r \ge r$,
with \eqref{eq:r:ah} occasionally delivering $\wh r < r$.

\begin{figure}[htb]
\centering
\includegraphics[width=.6\linewidth]{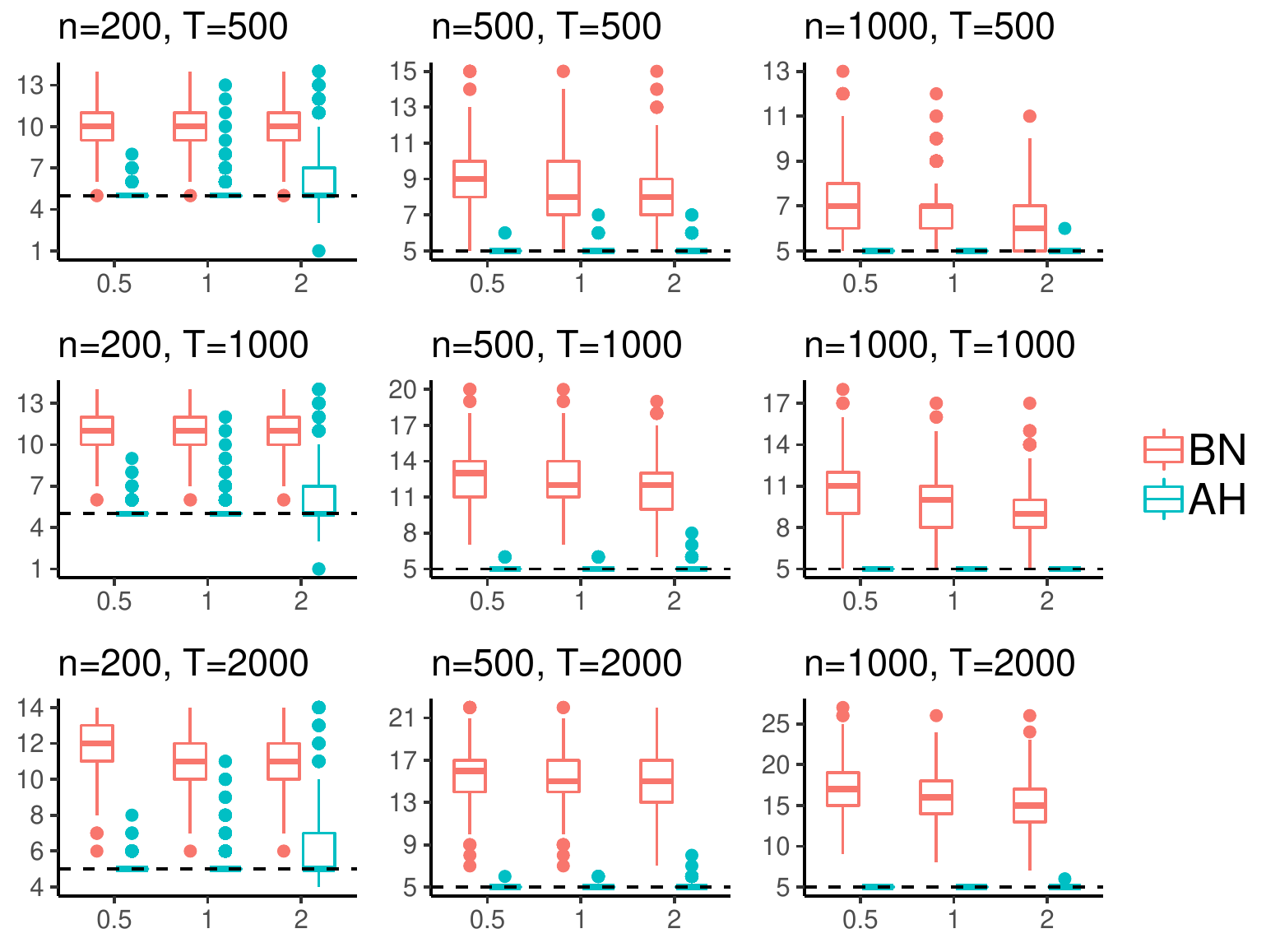}
\caption{\small Box plots of $\wh r$ returned by \eqref{eq:r:bn} (BN) and 
\eqref{eq:r:ah} (AH) over $1000$ realisations generated under Model 1
with $T \in \{500, 1000, 2000\}$ (top to bottom), 
$n \in \{200, 500, 1000\}$ (left to right) and
$\phi \in \{0.5, 1, 2\}$ (left to right within each plot, controls the noise-to-signal ratio);
horizontal broken lines indicate the true factor number $r = 5$.}
\label{fig:sim:r}
\end{figure}

Obviously, when $\wh r < r$, the PC estimator \eqref{eq:pca:est} or indeed, 
any estimator of the common component
does not capture the contribution from one or more factors, 
which inevitably incurs a non-negligible error and no remedy to this problem exists. 
To circumvent this problem, the user may be tempted to increase $\wh r$ 
based on the reasoning that 
the contribution of spurious factors beyond the $r$-th one is negligible
and thus such a strategy would be risk averse. 
However, this reasoning is incorrect as we show in the formal theoretical treatment 
of the impact of over-estimation of $r$ on factor analysis in the next section.

\subsection{PC estimator when $r$ is over-estimated}
\label{sec:blockwise}

While \citet{onatski15} shows in his Proposition~1 that 
the errors due to the over-estimation of $r$ is negligible
once aggregated over cross-sections and time,
a formal analysis of the impact of the over-estimated factor number on 
the PC estimators of {\it individual} common components
has not yet been conducted to the best of our knowledge. 

Recalling the PC estimator \eqref{eq:pca:est},
we have the following decomposition of the estimation error when $\wh r > r$, 
\begin{align}\label{eq:pca_error}
\wh{\chi}^{\pca}_{it}-\chi_{it}=\l(\sum_{j=1}^r \wh{w}_{x,ij}\wh{\mbf w}_{x,j}^\top \mbf x_t-\chi_{it}\r)
+ \sum_{j=r+1}^{\wh r} \wh{w}_{x,ij}\wh{\mbf w}_{x,j}^\top \mbf x_t.
\end{align}
The rate of convergence for the error in the oracle PC estimator
(first term in the RHS of \eqref{eq:pca_error})
is given in Proposition~\ref{thm:common}. 
Our interest lies in the theoretical treatment of the second term 
representing the over-estimation error. 
This faces two main challenges.

\bit[leftmargin=.7cm]
\item[(a)] The large eigengap between $\mu_{\chi, r}$ and $\mu_{\chi, r+1} = 0$ 
(see \ref{eq:c1}) and Davis-Kahan theorem play a key role in controlling the distance
between the empirical principal subspace spanned by the $r$ leading eigenvectors 
of $\wh{\bm\Gamma}_x$ and those of $\bm\Gamma_\chi$, as reported in \eqref{eq:davisk}. 
On the other hand, due to the lack of eigengap between the successive $\mu_{x, j}, \, j \ge r+1$ (see \ref{eq:c4}) or any other structural assumptions,
the behaviour of $\wh{\mbf w}_{x, j}$ for $j \ge r+1$ cannot be controlled 
in a meaningful way. 

\item[(b)] The eigenvectors $\wh{\mbf w}_{x, j}$, $1\le j \le \wh r$,
are obtained from the full sample covariance matrix
and thus are dependent on $\mbf x_t, \, 1 \le t \le T$,
which, together with the issue noted in (a), makes it difficult to analyse the stochastic properties of 
$\wh{\mbf w}_{x,j}^\top \mbf x_t$ for $j \ge r + 1$.
\eit


With these difficulties, we derive the following uniform but uninformative 
upper bound on the over-estimation error:
\begin{align}
\label{eq:pca:nonbound}
\max_{1 \le i \le n} \max_{1 \le t \le T}
\l\vert \sum_{j=r+1}^{\wh r} \wh{w}_{x,ij}\wh{\mbf w}_{x,j}^\top \mbf x_t \r\vert
\le \sum_{j = r+1}^{\wh r} \max_{1 \le i \le n}|\wh{w}_{x, ij}| \Vert \wh{\mbf w}_{x, j} \Vert \cdot
\max_{1 \le t \le T} \Vert \mbf x_t \Vert
= O_p(\sqrt n \log(T)).
\end{align}

In the next section, we propose modifications of the PC estimator
which directly address the issue raised in (a)
but first, we introduce a novel `blockwise' estimation technique which,
under the time series factor model~\eqref{eq:model},
allows for bypassing the issue raised in (b) 
and hence enables a rigorous theoretical analysis of the PC estimator
when $\wh r \ge r + 1$.
For this, we split the data into blocks of size $b_T$, say
$\{\mbf x_t, \, t \in I_\ell\}$ 
for $I_\ell := \{(\ell - 1)b_T + 1, \ldots, \min(\ell b_T, T)\}$, 
$\ell = 1, \ldots, L_T := \lceil T/b_T \rceil$.
Also, denote by 
$\bar{I}_\ell := \{1, \ldots, T\} \setminus \bigcup_{m \in \{\ell, \ell \pm 1\}} I_\ell$,
i.e., the set of indices that do not belong to $I_\ell$ or its adjacent blocks,
and by $\wh{\mbf w}_{x, j}^{(\ell)}$
the $j$-th leading eigenvector of 
$\wh{\bm\Gamma}_x^{(\ell)} = \vert \bar{I}_\ell \vert^{-1} 
\sum_{t \in \bar{I}_\ell} \mbf x_t\mbf x_t^\top$,
i.e., the sample covariance matrix constructed by {\it omitting}
the $\ell$-th and its adjacent blocks.
Then, we obtain the blockwise PC estimator of $\chi_{it}$ as
\begin{align}
\label{eq:ss:pca:est}
\wh\chi^{\bpca}_{it} = 
\sum_{j = 1}^{\wh r} \wh w^{(\ell)}_{x, ij}(\wh{\mbf w}^{(\ell)}_{x, j})^\top\mbf x_t
\quad \text{for} \quad t \in I_\ell, \quad 1 \le \ell \le  L_T.
\end{align}
In other words, the common components are estimated in a blockwise manner
as projections of $\mbf x_t$ onto the principal subspace of the subsample
obtained from omitting the current block as well as its immediate neighbours.
We select the block size $b_T$ to balance
between avoiding the asymptotic loss in efficiency
by having $|\bar{I}_\ell| \ge T - 3b_T$ as large as possible,
and ensuring that
the dependence between $\wh{\mbf w}^{(\ell)}_{x, j}$ and $\mbf x_t, \, t \in I_\ell$
is sufficiently weak under the strong mixing condition in Assumption~\ref{assum:dep}~(ii),
hence permitting the rigorous theoretical treatment of 
$(\wh{\mbf w}^{(\ell)}_{x, j})^\top\mbf x_t$ for $j \ge r + 1$.


\begin{prop}
\label{prop:pca} 
\textit{Let Assumptions~\ref{assum:id}--\ref{assum:dep} hold
and assume $r + 1 \le \wh r \le \bar r$ for some fixed $\bar r$.
Additionally, assume that $\mbf f_t$ and $\bm\vep_t$ are weakly stationary.
Suppose 
\begin{align}\label{eq:prop:pca}
\max_{1 \le i \le n}\, \max_{r + 1 \le j \le \wh r} |\wh w^{(\ell)}_{x, ij}| = O_p(n^{-\alpha/2}),
\end{align}
for some $1\le \ell \le L_T$ and $\alpha \in [0, 1]$,
and let $b_T = \log^{1/\beta + \delta} T$ for $\beta$ in Assumption~\ref{assum:dep} 
and some fixed $\delta > 0$. Then,
\begin{align}
\label{eq:prop:pca:res}
\max_{1 \le i \le n} \max_{1 \le t \le T}
|\wh\chi^{\bpca}_{it} - \chi_{it}| = 
O_p\l[n^{(1 - \alpha)/2} \l(\sqrt{\frac{\log(n)}{T}} 
\vee \frac{1}{\sqrt n}\r)\log(T) \r].
\end{align}
}
\end{prop}

The proof of Proposition~\ref{prop:pca} is provided in Appendix~\ref{sec:pf:prop:pca}. 
Condition \eqref{eq:prop:pca} is very general and its motivation is as follows.
Writing
$x_{it} = \sum_{j = 1}^{n \wedge T} \wh w_{x, ij} \wh{\mbf w}_{x, j}^\top\mbf x_t$,
we have
\begin{align}
\label{eq:var}
\frac{1}{T}\sum_{t = 1}^T x_{it}^2 = \sum_{j=1}^{n \wedge T} \wh{w}_{x, ij}^2 \wh\mu_{x, j} < \infty \quad
\text{a.s. for all } 1 \le i \le n,
\end{align}
which implies that $\max_{1 \le i \le n} |\wh{w}_{x, ij}| = O(\wh\mu_{x, j}^{-1/2})$.
In addition, the rate of convergence of the sample covariance matrix
$n^{-1}\Vert \wh{\bm\Gamma}_x - \bm\Gamma_x \Vert = O_p(\sqrt{\log(n)/T})$
(see Lemma~\ref{lem:cov}) and \ref{eq:c4}
yields $\wh\mu_{x, j} = O_p(n\sqrt{\log(n)/T}) = o_p(n)$ for $j \ge r + 1$.
These arguments hold for blockwise estimators as well, and
indicate that there may be (spuriously) large coordinates 
in the empirical eigenvectors $\wh{\mbf w}_{x, j}^{(\ell)}, \, j \ge r + 1$
that fall in the regime of $\alpha < 1$. 
In other words, \eqref{eq:prop:pca} is merely a consequence of 
the boundedness of $\mu_{x, j}, \, j \ge r + 1$
without any further structural assumptions on the model \eqref{eq:model}.

It has been shown that for a random matrix $\mbf M \in \R^{n \times T}$ with independent entries,
the eigenvectors of $T^{-1} \mbf M^\top\mbf M$
are `delocalised' in probability with the bound $1/\sqrt{n}$ up to a logarithmic factor
(see Theorem~B.3 of \cite{vu2015} and a survey given in \cite{o2016}).
In view of this, when $\mbf x_t \sim_{\iid} (\mbf 0, \bm\Gamma_x)$ and 
follows an {\it exact} factor model with $\bm\Gamma_\vep = \mbf I_n$,
the condition~\eqref{eq:prop:pca} is met with $\alpha = 1$ up to a logarithmic factor, and
the consistency of the PC estimator derived in Theorem~\ref{thm:common} carries over even with $\wh r > r$.
However, under the approximate time series factor model adopted in this paper, 
there is no such theoretical guarantee to the best of our knowledge.
In Section~\ref{sec:sim}, we verify that, 
under a variety of data generating models, 
the non-leading empirical eigenvectors
indeed exhibit `sparsity' with few very large coordinates,
thus corresponding to the regime $\alpha \simeq 0$.

The following Examples~\ref{ex:rev:one}--\ref{ex:rev:two}
provide the lower bounds complementing 
upper bounds in~\eqref{eq:prop:pca} and~\eqref{eq:prop:pca:res}
for a particular example where $\bm\Gamma_\vep$ follows a sparse spiked covariance model.
Together, Proposition~\ref{prop:pca} and Examples~\ref{ex:rev:one}--\ref{ex:rev:two}
are indicative of the potential pitfalls stemming from the over-estimation of $r$
for the PC estimator of the common component.

\begin{ex}[Lower bound on $\max_{1 \le i \le n} \max_{r + 1 \le j \le \wh r} \vert \wh{w}_{x, ij} \vert$]
\label{ex:rev:one}
We assume that $\bm\Gamma_\vep = \Delta_n \mbf v \mbf v^\top + \sigma^2 \mbf I_n$
with $\Vert \mbf v \Vert_0 \asymp n^\alpha$ for some $\alpha \in [0, 1)$ and $\mbf v^\top\mbf v=1$.
Also, we suppose $n \asymp T^\kappa$ for some $\kappa > 0$ (see Assumption~\ref{assum:nt}).
This leads to
\begin{align}
\label{ex:rev:one:model}
\bm\Gamma_x = \mbf W_\chi \mbf M_\chi \mbf W_\chi^\top + \Delta_n \mbf v \mbf v^\top + \sigma^2 \mbf I_n,
\end{align}
where $\mbf W_\chi$ is the $n\times r$ matrix of normalised eigenvectors and 
$\mbf M_\chi$ the $r \times r$ diagonal matrix of eigenvalues of $\bm\Gamma_\chi$.
Further, we assume that $\Delta_n \asymp n^\nu$ for some 
$\max(0, 1 - 1/(2\kappa) + \alpha/2) < \nu < 1$, and 
let $\mbf v^\top \mbf w_{\chi, j} = 0$ for all $j = 1, \ldots, r$. 
In this model, the idiosyncratic component has a one-factor structure with a weakly pervasive factor
where its `strength' $\Delta_n$ increases with $\alpha$.
We may interpret this as the weak factor being prevalent in all the elements 
belonging to a group defined by the support of $\mbf v$.
In time series setting, such a structure has also been considered by 
\cite{demol2008}, \cite{lam2011} and \cite{onatski12}, among others. 

When $\nu > 0$, the model~\eqref{ex:rev:one:model} does not fulfil Assumption~\ref{assum:id}~(iv).
However, even when $\nu \in (0, 1)$, the oracle PC estimator obtained with $\wh r = r$ 
can be shown to be consistent by adapting the proof of Proposition~\ref{thm:common} :
From $\Vert \wh{\bm\Gamma}_x - \bm\Gamma_\chi \Vert \le 
\Vert \wh{\bm\Gamma}_x - \bm\Gamma_x \Vert + \Vert \bm\Gamma_\vep \Vert$, we yield
\begin{align}
\Vert \wh{\mbf W}_x - \mbf W_\chi \mbf S \Vert & = O_p\l(\sqrt{\frac{\log(n)}{T}} \vee \frac{1}{n^{1 - \nu}}\r),
\quad \text{and} \nn \\
\max_{1 \le i \le n} \max_{1 \le t \le T} \vert \wh{\chi}^{\pca}_{it} - \chi_{it} \vert
& = O_p\l\{\Big(\sqrt{\frac{\log(n)}{T}} \vee \frac{1}{n^{(1 - \nu)/2}}\Big) \log(T) \r\}.
\label{eq:ex:rev:one}
\end{align}

Under model~\eqref{ex:rev:one:model}, for large enough $n$, we have
\begin{align*}
\mbf w_{x, j} = \l\{\begin{array}{ll}
\mbf w_{\chi, j} & \text{for } 1 \le j \le r, \\
\mbf v & \text{for } j = r + 1,
\end{array}\r.
\quad \text{with} \quad
\mu_{x, j} = \l\{\begin{array}{ll}
\mu_{\chi, j} + \sigma^2 & \text{for } 1 \le j \le r, \\
\Delta_n + \sigma^2 & \text{for } j = r + 1, \\
\sigma^2 & \text{for } r + 2 \le j \le n.
\end{array}\r.
\end{align*}

As in~\eqref{eq:davisk}, we apply Corollary~1 of \cite{yu15} and yield
\begin{align}
\label{eq:ex:rev:two}
\Vert \wh{\mbf w}_{x, r + 1} - s \mbf v \Vert \le \frac{2^{3/2} \Vert \wh{\bm\Gamma}_x - \bm\Gamma_x \Vert}
{\min(\mu_{x, r} - \mu_{x, r + 1}, \mu_{x, r + 1} - \mu_{x, r + 2})}
= O_p\l(n^{1-\nu}\sqrt{\frac{\log(n)}{T}}\r) = o_p(n^{-\alpha/2})
\end{align}
for some $s \in \{-1, 1\}$,
i.e., $\wh{\mbf w}_{x, r + 1}$ achieves consistency in estimating $\mbf v$
albeit at a slower convergence rate than that reported in~\eqref{eq:davisk}.
Also, the sparsity of $\mbf v$ leads to $n^{-\alpha/2} (\max_{1 \le i \le n} \vert v_i \vert)^{-1} = O(1)$,
and thus from \eqref{eq:ex:rev:two}, for some fixed $C_0 > 0$,
\begin{align}
\max_{1 \le i \le n} \vert \wh{w}_{x, i, r + 1} \vert =
\max_{1 \le i \le n} \vert v_i \vert + O_p\l(n^{1 - \nu} \sqrt{\frac{\log(n)}{T}}\r)
\ge C_0 n^{-\alpha/2}.
\label{eq:ex:rev:three}
\end{align}
\end{ex}

\begin{ex}[Lower bound on the estimation error in~\eqref{eq:prop:pca:res}]
\label{ex:rev:two}
Continuing with the model~\eqref{ex:rev:one:model} imposed on $\bm\Gamma_x$,
we further assume that $\mbf x_t \sim_{\iid} \mc N_n(\mbf 0, \bm\Gamma_x)$
and $n \asymp T$ for simplicity (i.e., $\kappa = 1$)
such that $\nu \in ((1 + \alpha)/2, 1)$.
Under independence, we simplify the blockwise estimator as
\begin{align*}
\wh\chi^{\bpca}_{it} = \sum_{j = 1}^{\wh r} \wh w^{(\ell)}_{x, ij} (\wh{\mbf w}^{(\ell)}_{x, j})^\top\mbf x_t
\quad \text{for} \quad t \in I_\ell, \quad \ell = 0, 1,
\end{align*}
with $I_0 = \{2u, \, 1 \le u \le \lfloor T/2 \rfloor\} = \bar{I}_1$ and
$I_1 = \{ 2u + 1, \, 0 \le u \le \lfloor T/2 \rfloor\} = \bar{I}_0$.
Suppose that $\wh r = r + 1$. Then there exist fixed $C_k > 0, \, 1 \le k \le 4$ such that
\begin{align}
& \max_{1 \le i \le n} \max_{1 \le t \le T} \vert \wh{\chi}^{\bpca}_{it} - \chi_{it} \vert
\ge \max_{\ell = 0, 1} \max_{t \in I_\ell} \l\vert \wh{w}^{(\ell)}_{x, 1, r + 1} \r\vert \; 
\l\vert (\wh{\mbf w}^{(\ell)}_{x, r + 1})^\top \mbf x_t \r\vert 
\nn \\
& \qquad \qquad \qquad \qquad \qquad \qquad  - \max_{1 \le i \le n} \max_{\ell = 0, 1} \max_{t \in I_\ell}
\l\vert \sum_{j = 1}^r \wh{w}^{(\ell)}_{x,ij}(\wh{\mbf w}^{(\ell)}_{x, j})^\top \mbf x_t - \chi_{it} \r\vert
\nn \\
& \ge C_1 n^{-\alpha/2} \max_{\ell = 0, 1} \max_{t \in I_\ell} 
\l\vert (\wh{\mbf w}^{(\ell)}_{x, r + 1})^\top \mbf x_t \r\vert
- C_2 n^{-(1 -\nu)/2} \sqrt{\log(T)}
\nn \\
& \ge  C_3 n^{-\alpha/2} \cdot \sqrt{\Delta_n \log(T)} - C_2 n^{\nu/2 - 1/2} \sqrt{\log(T)}
\ge C_4 n^{(\nu-\alpha)/2} \sqrt{\log(T)} \nn
\end{align}
where all the inequalities except for the first are understood as holding with probability tending to one.
The second inequality follows from~\eqref{eq:ex:rev:one}, \eqref{eq:ex:rev:three} and that $n \asymp T$,
with the rate $\sqrt{\log(T)}$ due to the stronger Gaussian assumption we impose here
in place of the sub-exponential tail in Assumption~\ref{assum:tail}~(ii).
The penultimate inequality holds
by Theorem~3.4 of \cite{hartigan2014} since for each $\ell = 0, 1$, we have
$(\wh{\mbf w}^{(\ell)}_{x, r + 1})^\top \mbf x_t \sim_{\iid} \mc N(0, \wt\sigma^2)$ for $t \in I_\ell$ with
\begin{align*}
\wt\sigma^2 \ge \Delta_n \l\{(\wh{\mbf w}^{(\ell)}_{x, r + 1})^\top \mbf v\r\}^2 + \sigma^2
\ge \Delta_n \l\{1 + o_p(n^{-\alpha/2}) \r\} + \sigma^2
\end{align*}
by applying Corollary~1 of \cite{yu15} as in~\eqref{eq:ex:rev:two}.
For comparison, we derive the upper bound on the estimation error in this setting
as in Proposition~\ref{prop:pca}.
From~\eqref{eq:ex:rev:three}, we have
$\max_{\ell = 0, 1} \max_{1 \le i \le n} \vert \wh{w}^{(\ell)}_{x, i, r + 1} \vert \asymp n^{-\alpha/2}$
and by adopting the arguments analogous to those used in the proof of Proposition~\ref{prop:pca}, 
it is readily seen that 
\begin{align*}
\max_{1 \le i \le n} \max_{1 \le t \le T} \vert \wh{\chi}^{\bpca}_{it} - \chi_{it} \vert
= O_p(n^{(\nu - \alpha)/2}\sqrt{\log(T)}).
\end{align*}
\end{ex}

Examples~\ref{ex:rev:one}--\ref{ex:rev:two} demonstrate that 
in the presence of weak factors,
the PC estimator can incur non-negligible error increasing with $n$
due to the `localised' behaviour of $\wh{\mbf w}_{x, j}, \, j \ge r + 1$
when the factor number is over-estimated.
In practice, such situations can emerge when the idiosyncratic component exhibits a group structure 
that induces the presence of weak factors. 
\cite{chudik2011} discuss the plausibility of {\it semi-weak} and {\it semi-strong} factors
corresponding to $\Vert \bm\Gamma_\vep \Vert \asymp n^\nu$ with $\nu \in (0, 1)$ in real datasets.
In Section~\ref{sec:real}, the daily returns of the stocks comprising the Standard \& Poor's 100 index 
are analysed where Figure~\ref{fig:sp100:idio} shows a clear group structure
in the idiosyncratic component which is in line with the model~\eqref{ex:rev:one:model}. 
We also refer to \cite{BH16} where the network structure in the idiosyncratic component 
of the similar dataset has been analysed in detail.

\begin{rem}[Large covariance matrix estimation]
\label{rem:poet}
\cite{fan2011} and \cite{fan13} investigate the problem of large covariance matrix estimation
with an estimator comprised of a factor-driven covariance matrix of the common component
and a thresholded idiosyncratic covariance matrix under the assumption of sparsity on $\bm\Gamma_\vep$
(see~\eqref{eq:poet});
in the former, the factors are assumed to be observable and 
the latter extends the estimator to the case of unobservable factors.
For the consistency of the thresholded idiosyncratic covariance matrix,
Assumption~2.2 of \cite{fan2011} requires $\wh\vep_{it}$, an estimator of $\vep_{it}$, to satisfy
\begin{align}
\label{eq:fan:cond}
\max_{1 \le i \le n} \frac{1}{T} \sum_{t = 1}^T |\wh\vep_{it} - \vep_{it}|^2 = o_p(1) \quad \text{and} \quad
\max_{1 \le i \le n} \max_{1 \le t \le T} |\wh\vep_{it} - \vep_{it}| = O_p(1),
\end{align}
and Lemma~C.11 of \cite{fan13} verifies the conditions
for the PC estimator combined with the estimator of $r$ proposed in \cite{baing02}.
However, Proposition~\ref{thm:common} indicates that 
the second condition in \eqref{eq:fan:cond} may be violated when $\wh r  > r$.
Moreover, our numerical studies in Section~\ref{sec:sim} demonstrate
that neither of the conditions in \eqref{eq:fan:cond} are fulfilled by the PC estimator 
when the idiosyncratic components are moderately correlated,
which in turn implies that the covariance matrix estimator of \citet{fan13}
will suffer from relying on the accurate estimation of $r$.
We further explore this point by applying our methodology to 
estimating the covariance of a panel of financial time series in Section~\ref{sec:real}.
\end{rem}


\section{Modification of the PC estimator}
\label{sec:mod}

\subsection{Scaled PC estimator}
\label{sec:scale}

Recall that due to the presence of an eigengap \ref{eq:c3}--\ref{eq:c4}
and the consistency of the $r$ leading eigenvectors of $\wh{\bm\Gamma}_x$
(see \eqref{eq:davisk}),
we obtain the uniform bound of $O_p(n^{-1/2})$ on $|\wh{w}_{x, ij}|, \, j \le r$,
see \eqref{eq:bounded:ew}.
In other words, with large probability, there exists some fixed $c_w > 0$
such that $|\wh w_{x, ij}|, \, j \le r$ is bounded by $c_w n^{-1/2}$
uniformly in $1 \le i \le n$ and $1 \le j \le r$.
Motivated by these observations, we propose the scaled PC estimator
\begin{align}
& \wh\chi^{\sca}_{it} = \sum_{j=1}^{\wh r} \wh{w}^{\sca}_{x, ij} (\wh{\mbf w}^{\sca}_{x, j})^\top \mbf x_t,
\quad \text{where}
\label{eq:sca:est}
\\
& \wh{\mbf w}^{\sca}_{x, j} = \nu_j^{-1} \; \wh{\mbf w}_{x, j} \quad \text{with} \quad
\nu_j = \max\l\{1, \frac{\sqrt{n}}{c_w} \, \max_{1 \le i \le n} |\wh w_{x, ij}|\r\}.
\label{eq:w:scale}
\end{align}

We can choose $c_w$ such that with large probability,
the proposed scaling does not alter the contribution from 
$\wh{\mbf w}_{x, j}, \, j \le r$ to $\wh\chi^{\sca}_{it}$ by yielding $\nu_j = 1$ for $j \le r$,
even though it is applied \textit{without} knowing $r$.
On the contrary, for $\wh{\mbf w}_{x, j}, \, j \ge r+1$,
any large contribution from the spurious factors is scaled down by the factor of $\nu_j$.

\begin{rem}[Choice of $c_w$]
\label{rem:cw}
In our numerical analysis, we observe that the performance of
the scaled PC estimator does not vary much with respect to reasonably chosen $c_w$, 
see Figure~\ref{fig:cv:err} (the details of the experiment is deferred to Appendix~\ref{sec:cw}).
Unlike e.g., the methods based on singular value thresholding,
our scaled PC estimator does not `kill' any factors including the spurious ones,
and thus avoids the hazard of under-estimating the contribution of the factors completely 
provided that $\wh r \ge r$. 
We find the choice of $c_w = 1.1 \times \sqrt n \; \max_{1 \le i \le n} |\wh w_{x, i1}|$ works reasonably
well and adopt it throughout the numerical studies,
which is shown to work well for a range of models in Section~\ref{sec:sim}.
\end{rem}

\begin{figure}[htb]
\centering
\includegraphics[width=1\textwidth]{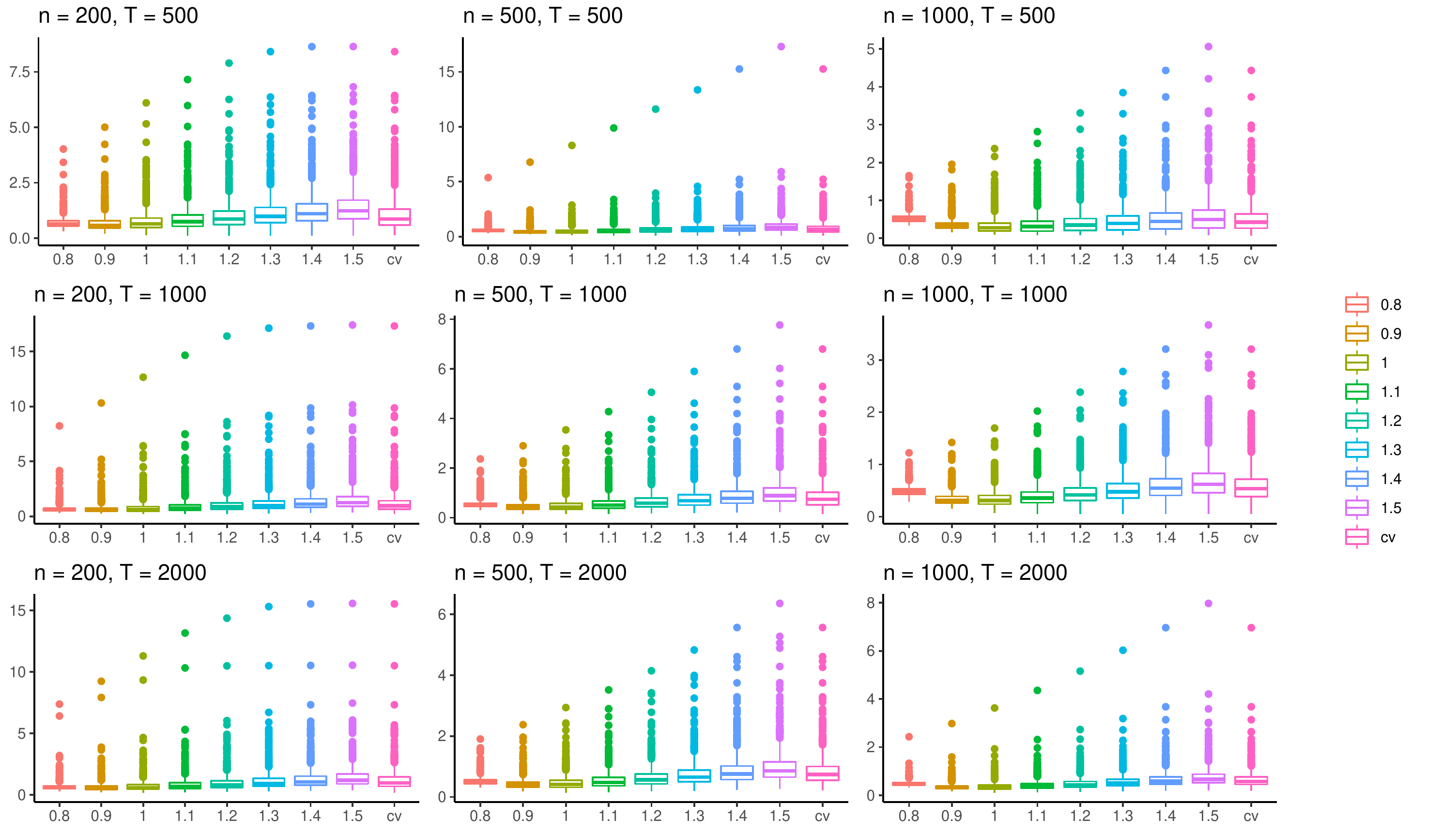}
\caption{\small Box plots of the estimation errors of $\wh{\chi}^{\bsca}_{it}(c_w)$
averaged over $1000$ realisations generated under Model~1 of Section~\ref{sec:sim:model} with $\phi = 1$,
$n \in \{200, 500, 1000\}$ (left to right) and $T \in \{500, 1000, 2000\}$ (top to bottom),
for $c_w \in \{0.8, \ldots, 1.5\} \times \sqrt n \; \max_{1 \le i \le n} |\wh w_{x, i1}|$ 
and the cross-validated choice (`CV')
(left to right within each plot).}
\label{fig:cv:err}
\end{figure}

Scaling preserves the orthogonality among $\wh{\mbf w}^{\sca}_{x, j}, \, j \le \wh r$,
which facilitates the theoretical treatment of the scaled PC estimator.
Following the same reasoning as in Section~\ref{sec:blockwise},
we continue the discussion on the theoretical properties of the scaled PC estimator
by considering its blockwise counterpart, which, recalling the notations from Section~\ref{sec:blockwise}, is given by
\begin{align}
\label{eq:ss:sca:est}
\wh\chi^{\bsca}_{it} =
\sum_{j = 1}^{\wh r} \wh w^{\sca, (\ell)}_{x, ij}(\wh{\mbf w}^{\sca, (\ell)}_{x, j})^\top\mbf x_t \quad \text{for} \quad t \in I_\ell, \quad 1 \le \ell \le L_T,
\end{align}
where
$\wh{\mbf w}^{\sca, (\ell)}_{x, j}$ is defined analogously as in \eqref{eq:w:scale}
with $\wh{\mbf w}^{(\ell)}_{x, j}$ in place of $\wh{\mbf w}_{x, j}$.

\begin{prop}
\label{prop:scaled}
\textit{Let Assumptions~\ref{assum:id}--\ref{assum:dep} hold
and suppose $r + 1 \le \wh r \le \bar r$ for some fixed~$\bar r$.
Additionally, assume that $\mbf f_t$ and $\bm\vep_t$ are weakly stationary.
Then, there exists a fixed constant $c_w$
satisfying $c_w \ge \sqrt{n} \times \max_{1 \le \ell \le L_T} 
\max_{1 \le i \le n} \max_{1 \le j \le r} \vert \wh{w}_{x, ij} \vert$
such that
\begin{align}
\max_{1 \le i \le n} \; \max_{1 \le t \le T} \vert \wh\chi_{it}^{\bsca} - \chi_{it} \vert 
= O_p\l\{\l(\sqrt{\frac{\log(n)}{T}} \vee \frac{1}{\sqrt n}\r)\log(T) \r\}.
\label{eq:scaling:est:err}
\end{align}
}
\end{prop}
The proof is provided in Appendix~\ref{sec:pf:prop:scaled}.
Compared to Propositions~\ref{thm:common} and~\ref{prop:pca},
Proposition~\ref{prop:scaled} establishes that under the same conditions,
the scaled PC estimator attains the same rate of convergence 
as the oracle PC estimator obtained with the true number of factors,
{\it without} requiring such knowledge and
{\it regardless} of the behaviour of $\wh{\mbf w}_{x, j}, \, j \ge r+1$.

\subsection{Relationship to capped PC estimator}
\label{sec:cap}

Similarly motivated by the uniform boundedness of $\vert \wh w_{x, ij} \vert$ for $j \le r$ 
(see \eqref{eq:bounded:ew}),
\cite{bcf2017} proposed the capped PC estimator of $\chi_{it}$:
\begin{align}
\label{eq:cap:est}
\wh{\chi}^{\capp}_{it} = \sum_{j = 1}^{\wh r} 
\wh{w}^{\capp}_{x, ij} (\wh{\mbf w}^{\capp}_{x, j})^\top \mbf x_t,
\end{align}
where each element of $\wh{\mbf w}^{\capp}_{x, j}$ is obtained by capping $\wh w_{x, ij}$ as
\begin{align}
\label{eq:w:cap}
\wh{w}^{\capp}_{x, ij} = \wh{w}_{x, ij} \, \mathbb I\l(|\wh{w}_{x, ij}| \le \frac{c_w}{\sqrt n}\r)
+ \mbox{sign}(\wh{w}_{x, ij}) \cdot \frac{c_w}{\sqrt n} \, \mathbb I\l(|\wh{w}_{x, ij}| > \frac{c_w}{\sqrt n}\r)
\end{align}
for some fixed $c_w > 0$.
Capping can be viewed as the projection of each $\wh{\mbf w}_{x, j}$ onto the
$\ell_\infty$-sphere of radius $c_w n^{-1/2}$.
As with scaling, asymptotically, capping does not alter the contribution from the leading $r$ 
eigenvectors of $\wh{\bm\Gamma}_x$,
while it truncates any large contribution from spurious factors when $\wh r \ge r+1$,
all {\it without} the knowledge of the true $r$. 
We generalise Theorem~2 of \cite{bcf2017} 
whereby lifting the assumption of Gaussianity imposed on $\bm\vep_t$ in the latter;
the proof can be found in Appendix~\ref{sec:pf:capping}.

\begin{prop}
\label{thm:two}
{\it Let Assumptions~\ref{assum:id}--\ref{assum:dep} hold and suppose 
$r + 1 \le \wh r \le \bar r$ for some fixed $\bar r$.
Then, there exists a fixed constant $c_w$
satisfying $c_w \ge \sqrt{n} \times \max_{1 \le i \le n} \max_{1 \le j \le r} \vert \wh{w}_{x, ij} \vert$
such that
$\max_{1 \le i \le n} \; \max_{1 \le t \le T}  \vert \wh\chi^{\capp}_{it} - \chi_{it} \vert = O_p(\log T)$.
}
\end{prop}

Refinement of the upper bound given in Proposition~\ref{thm:two} is a difficult task
as reasoned in (a)--(b) of Section~\ref{sec:blockwise},
even when considering its blockwise version, $\wh{\chi}^{\bcapp}_{it}$,
due to the lack of orthogonality of the capped eigenvectors.
Nevertheless, Proposition~\ref{thm:two} shows that
the capped estimator $\wh\chi^{\capp}_{it}$ improves upon
the worst case performance of the PC estimator reported in \eqref{eq:pca:nonbound}. 

Unlike the capped PC estimator,
the scaled PC estimator shrinks down the eigenvectors after modification
from $\Vert \wh{\mbf w}_{x, j} \Vert^2 = 1$ 
to $\Vert \wh{\mbf w}^{\sca}_{x, j} \Vert^2  = \nu_j^{-1}$,
which further curtails the spurious contribution from $\wh{\mbf w}_{x, j}, \, j \ge r + 1$
as demonstrated in the following example.

\begin{ex}
For simplicity, let us ignore the stochastic nature of $\wh{\mbf w}_{x, j}$ and
suppose that $\wh{\mbf w}_{x, j^\prime}$ for some $j^\prime \ge r+1$ is approximately sparse.
That is, there exists $\mc S \subset \{1, \ldots, n\}$ with $|\mc S| = O(1)$ and a fixed $c_0 > 0$ 
such that $|\wh w_{x, ij^\prime}| \ge c_0, \, i \in \mc S$, while 
$\max_{i \notin \mc S} |\wh w_{x, ij^\prime}| = O(n^{-1/2})$.
Then, we have $\Vert \wh{\mbf w}^{\sca}_{x, j^\prime} \Vert^2 \le c_w(c_0\sqrt n)^{-1}$,
which shrinks the overall contribution of the $j^\prime$-th estimated factor to 
$\wh{\bm\chi}^{\sca}_t$ by the factor of $\sqrt n$, in comparison with that to the PC estimator.
In the same scenario, however, capping does not always lead to $\Vert \wh{\mbf w}^{\capp}_{x, j^\prime} \Vert = o(1)$.
Consider e.g., $\wh{\mbf w}_{x, j^\prime} = (1/\sqrt 2, 1/\sqrt{2(n-1)}, \ldots, 1/\sqrt{2(n-1)})^\top$
and $c_w/\sqrt{n} \ge 1/\sqrt{2(n-1)}$, in which case $\Vert \wh{\mbf w}^{\capp}_{x, j^\prime} \Vert \ge 1/\sqrt{2}$.
\end{ex}

\subsection{Relationship to eigenvalue shrinkage}

Recalling that $\max_{1 \le i \le n} |\wh w_{x, ij}| = O(\wh\mu_{x, j}^{-1/2})$ 
(see the discussion following \eqref{eq:var}),
we may re-write the scaling factor $\nu_j$ using the 
choice $c_w = 1.1 \times \sqrt{n} \max_{1 \le i \le n} |\wh w_{x, i1}|$ 
as suggested in Remark \ref{rem:cw}:
\begin{align*}
\nu_j &= \max\l\{1, \frac{\max_{1 \le i \le n} |\wh w_{x, ij}|}{1.1 \times \max_{1 \le i \le n} |\wh w_{x, i1}|}\r\}
= \max\l\{1, \sqrt{\frac{C_j\wh\mu_{x, 1}}{\wh\mu_{x, j}}}\r\}, \quad \text{such that}
\\
\wh{\chi}^{\sca}_{it} &= \sum_{j = 1}^{\wh r}
\min\l\{1, \sqrt{\frac{\wh\mu_{x, j}}{C_j\wh\mu_{x, 1}}}\r\} \wh{w}_{x, ij} \wh{\mbf w}_{x, j}^\top\mbf x_t
\end{align*}
with some fixed $C_j > 0$. In other words, for some choice of $c_w$,
the scaled PC estimator admits a representation 
as a PC estimator combined with the eigenvalue-based shrinkage.
Ideal choices for $C_j$ are $C_ j \gg \wh\mu_{x, j}/\wh\mu_{x, 1}$ for $j \ge r+1$, 
and $C_j \le \wh\mu_{x, j}/\wh\mu_{x, 1}$ for $j \le r$ which, however, are infeasible 
since they require the knowledge of $r$. 
We consider a simpler but feasible choice of $C_j = 1$ for all $j$, 
and define the modified PC estimator based on eigenvalue shrinkage:
\begin{align}
\label{eq:cs:est}
\wh\chi^{\cs}_{it} = \sum_{j = 1}^{\wh r} \sqrt{\frac{\wh\mu_{x, j}}{\wh\mu_{x, 1}}} \cdot
\wh{w}_{x, ij} (\wh{\mbf w}_{x, j})^\top \mbf x_t.
\end{align}
Its blockwise version $\wh\chi^{\bcs}_{it}$ is defined analogously
with $\wh{\mbf w}_{x, j}^{(\ell)}$ and the corresponding eigenvalues $\wh\mu_{x, j}^{(\ell)}$
replacing $\wh{\mbf w}_{x, j}$ and $\wh\mu_{x, j}$, respectively.  
This estimator is expected to keep under control the over-estimation error, 
since $\wh\mu_{x, j}/\wh\mu_{x, 1} = o_p(1)$ for $j \ge r+1$
while being asymptotically bounded away from zero for $j \le r$. 
Hence, $\wh\chi^{\cs}_{it}$ preserves the contribution of the leading PCs although with a possible bias. 
From the simulation studies in Section~\ref{sec:sim},
we observe that any bias incurred by over-shrinkage is well compensated by 
its effectiveness in shrinking down the spuriously large over-estimation error.

\begin{rem}[Eigenvalue shrinkage]
\label{rem:shr}
The good performance of the shrinkage estimator in \eqref{eq:cs:est} may be explained by its link to the literature on eigenvalue shrinkage-based estimators.
\cite{donoho2018a} and \cite{donoho2018b} investigate
the optimal eigenvalue shrinkage for spiked covariance matrix estimation
when $\mbf x_t \sim_{\iid} \mc N_n(\mbf 0, \bm\Gamma_x)$
with $\mu_{x, 1} \ge \ldots \ge \mu_{x, r} > 1$ and $\mu_{x, j} = 1, \, j \ge r + 1$.
It has been shown that 
for any loss function considered therein,
the optimal eigenvalue shrinkage function $\eta$ yields $\eta(\wh\mu_{x, j}) < \wh\mu_{x, j}$.
Heuristically, shrinkage of eigenvalues 
not only accounts for the upward shift of empirical eigenvalues,
but also the inconsistency in empirical eigenvectors \citep{donoho2018a}.
\end{rem}

\section{Simulation studies}
\label{sec:sim}

\subsection{Set-up}
\label{sec:sim:model}

We consider the following data generating model
which allows for serial correlations in $f_{jt}$ and 
both serial and cross-sectional correlations in $\vep_{it}$.
\begin{align}
& x_{it} = r^{-1/2} \sum_{j=1}^r \lambda_{ij}f_{jt} + \sqrt{\phi} \vep_{it}, \, 1 \le i \le n; \, 1 \le t \le T,
\quad \mbox{where} \label{eq:sim:model}
\\
& f_{jt} = \rho_{f, j} f_{j, t-1} + u_{jt}, \quad  \vep_{it} = \rho_{\vep, i} \vep_{i, t-1} + v_{it}, \nn
\end{align}
with factor loadings $\lambda_{ij} \sim_{\iid} \mc N(0, 1)$, 
factor innovations $u_{jt} \sim_{\iid} \mc N(0, 1/(1 - \rho_{f, j}^2))$, 
and the autoregressive parameters as
$\rho_{f, j} = \rho_f - 0.05(j - 1)$ with $\rho_f = 0.5$.
For the idiosyncratic innovations $v_{it}$, we consider the following two models. \medskip

\textbf{Model~1}. With $H = 10$, $e_{it} \sim_{\iid} \mc N(0, 1 - \rho_{\vep, i}^2)$
and $\beta_i \sim_{\iid} \text{Unif}\{-0.15, 0.15\}$,
we generate
$v_{it} = (1 + 2\beta_i^2 H)^{-1/2}(e_{it} + \beta_i \sum_{l = i - H, l \ne i}^{i + H} e_{lt})$,
and set $\rho_{\vep, i} \sim_{\iid} \text{Unif}\{0.2, -0.2\}$.
This model has been taken from \cite{baing02} except that 
we select the parameters $\rho_{\vep, i}$, $\beta_i$ and $H$ of smaller magnitude such that
the problem of identifying $r$ is in fact {\it easier} here than in the original paper.
\medskip

\textbf{Model~2 \citep{cai2015}}. The vector $\mbf v_t=(v_{1t},\ldots,v_{nt})^\top$ is such that 
$\mbf v_t = \bm\Gamma^{1/2}_v \mbf e_t$, where $\bm\Gamma_v = \mbf V\bm\Delta\mbf V^\top + \mbf I_n$,
and $\mbf e_t \sim_{\iid} \mc N_n(0, (1-\rho_{\vep, i}^2) \mbf I_n)$.
The diagonal matrix $\bm\Delta$ has $r$ non-zero eigenvalues taking equidistant values from $20$ to $10$,
and $\mbf V$ is chosen as the $r$ leading left singular vectors of 
a matrix $\mbf M \in \R^{n \times r}$,
whose first $[\varrho n]$ rows are drawn independently from $\mc N(0, 1)$
and the rest are set to zero. 
By construction, this models adds $r$ additional `weak' factors stemming 
from the large (although bounded for all $n$) eigenvalues of $\bm\Gamma_v$.
\medskip

Model~1 provides a benchmark as it is popularly adopted in the factor model literature, 
while Model~2 mimics the case of weak factors considered 
in Examples~\ref{ex:rev:one}--\ref{ex:rev:two}.

We control the `noise-to-signal' ratio 
with $\phi \in \{0.5, 1, 2\}$
such that larger values of $\phi$ correspond to the low signal-to-noise ratio.
Throughout, we set $r = 5$, and consider $T \in \{500, 1000, 2000\}$ and $n \in \{200, 500, 1000\}$,
and $\varrho \in \{0.2, 0.5, 0.9\}$ for Model~2.

We explore the in-sample estimation accuracy of 
the PC estimator $\wh\chi^{\pca}_{it}$ in \eqref{eq:pca:est},
the scaled estimator $\wh\chi^{\sca}_{it}$ in \eqref{eq:sca:est},
the capped estimator $\wh\chi^{\capp}_{it}$ in \eqref{eq:cap:est}
and the shrinkage estimator
$\wh\chi^{\cs}_{it}$ in \eqref{eq:cs:est},
with and without blockwise estimation for which we set $b_T = [\log^2 T]$.
For estimating $r$, we consider the two estimators
\eqref{eq:r:bn} (`BN') and \eqref{eq:r:ah} (`AH'),
setting $r_{\max} = [\sqrt{n \wedge T}]$.
For comparison, we also investigate the performance of the oracle estimator
$\wh\chi^{\oracle}_{it}$ defined as 
the PC estimator~\eqref{eq:pca:est} obtained with the true $r$. 
For each setting, $1000$ realisations have been generated.
We provide the results obtained under Model~1 in the main text, 
and additional simulation results under Model~2 are available in Appendix~\ref{sec:sim:add}.
Based on Figure~\ref{fig:sim:r}, 
we report the results when the factor number estimator~\eqref{eq:r:bn} is used,
in order to contrast the behaviour of the proposed modified PC estimators
to that of the PC estimator in terms of their `insensitivity' to the over-estimation of $r$.

\subsection{Results}
\label{sec:sim:res}

\begin{figure}[htb]
\centering
\includegraphics[width=.8\linewidth]{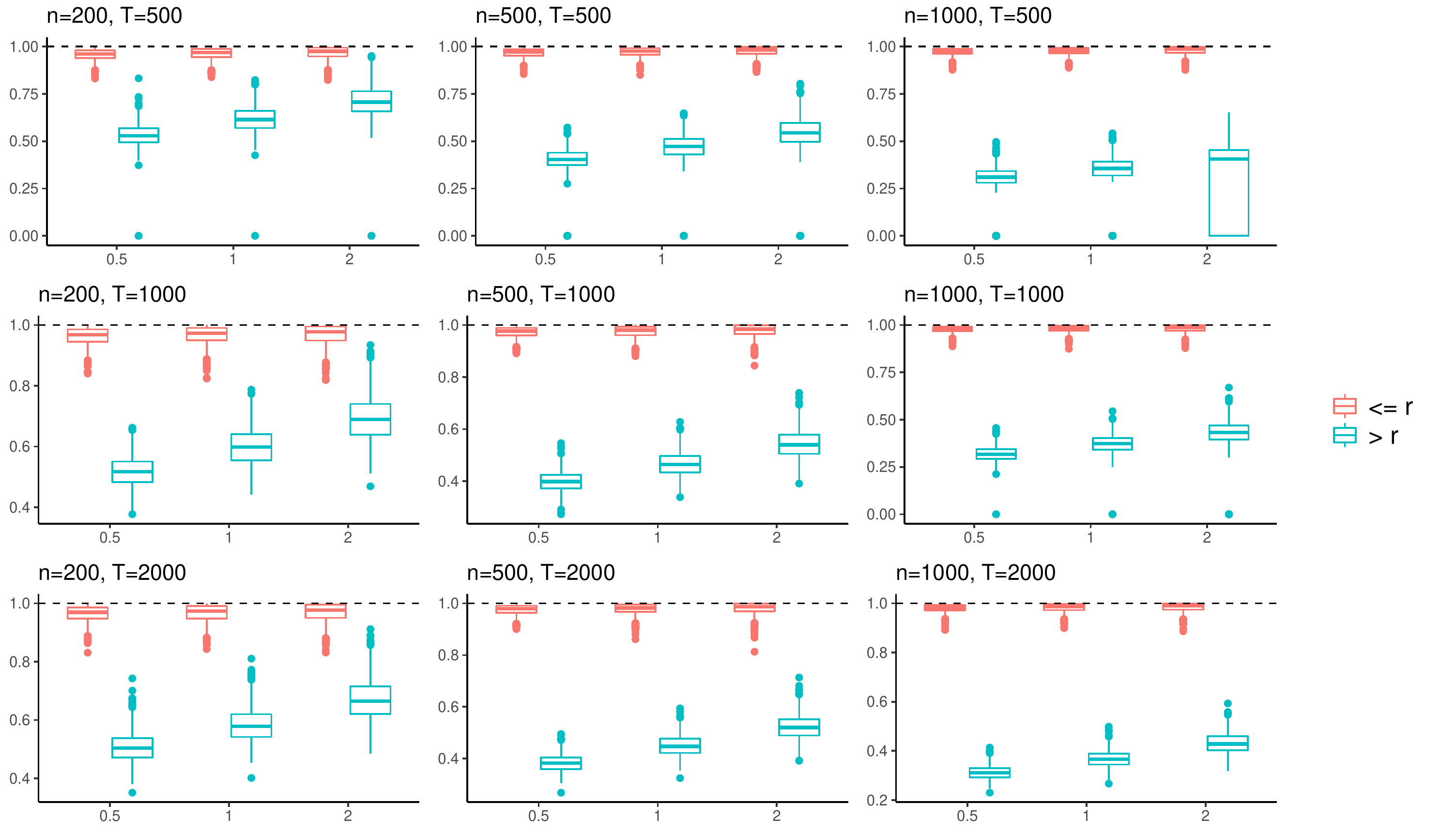}
\caption{\small Box plots of
$\Vert\wh{\mbf w}^{\sca}_{x, j}\Vert$ averaged for $2 \le j \le r$ (`$\le r$')
against that averaged for $r + 1 \le j \le \wh r$ (`$> r$')
averaged $1000$ realisations generated under Model~1
with $n \in \{200, 500, 1000\}$ (left to right), $T \in \{500, 1000, 2000\}$ (top to bottom)
and $\phi \in \{0.5, 1, 2\}$  (left to right within each plot).}
\label{fig:sim:v}
\end{figure}

First, we investigate the amount of scaling and capping applied to $\wh{\mbf w}_{x, j}$
when $j \le r$ and $j \ge r + 1$,
in order to verify whether the asymptotic argument in \eqref{eq:bounded:ew} is valid 
for finite $n$ and $T$.
Figure~\ref{fig:sim:v} plots
the norm of the scaled eigenvectors $\Vert \wh{\mbf w}^{\sca}_{x, j} \Vert$ in~\eqref{eq:w:scale}
averaged over $2 \le j \le r$ and $r + 1 \le j \le \wh{r}$, respectively,
for varying $T$, $n$ and $\varrho$.
The results confirm that across different scenarios,
scaling does not alter the contribution
from the leading $r$ eigenvectors of $\wh{\bm\Gamma}_x$,
while curtailing that from $\wh{\mbf w}_{x, j}, \, j \ge r + 1$
by yielding $\Vert \wh{\mbf w}^{\sca}_{x, j} \Vert \ll \Vert \wh{\mbf w}_{x, j} \Vert = 1, \, j \ge r + 1$
especially for large $n$.
This in turn indicates that there are a few spuriously large coordinates in $\wh{\mbf w}_{x, j}, \, j \ge r + 1$, 
corresponding to the regime $\alpha \simeq 0$ in Proposition~\ref{prop:pca}. 
As shown below, this leads to the undesirable behaviour of the PC estimator
while affecting the modified PC estimators to a much lesser degree.

\begin{figure}[htb]
\centering
\includegraphics[width=.6\textwidth]{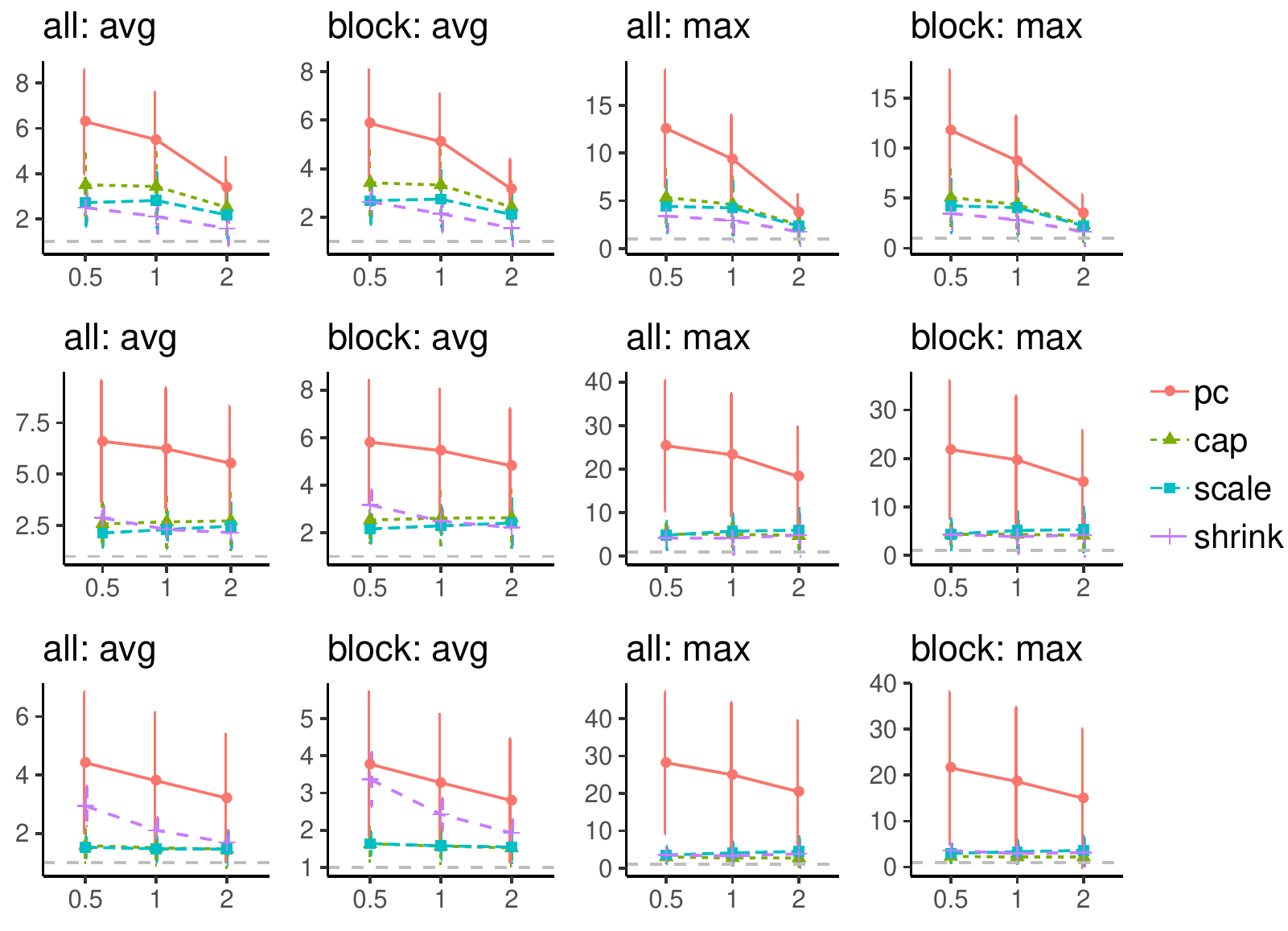}
\caption{\small $\text{err}_{\avg}(\wh\chi^{\circ}_{it})$
and $\text{err}_{\max}(\wh\chi^{\circ}_{it})$
of $\wh{\chi}^{\pca}_{it}$, $\wh{\chi}^{\capp}_{it}$, $\wh{\chi}^{\sca}_{it}$
and $\wh{\chi}^{\cs}_{it}$ estimated using the entire sample (`all'),
and their blockwise counterparts (`block'),
averaged over $1000$ realisations generated under Model~1 with $T = 500$,
$n \in \{200, 500, 1000\}$ (top to bottom)
and $\phi \in \{0.5, 1, 2\}$ (left to right within each plot).
The vertical errors bars represent the standard deviations.}
\label{fig:sim:err:T500}
\end{figure}

\begin{figure}[htb]
\centering
\includegraphics[width=.6\textwidth]{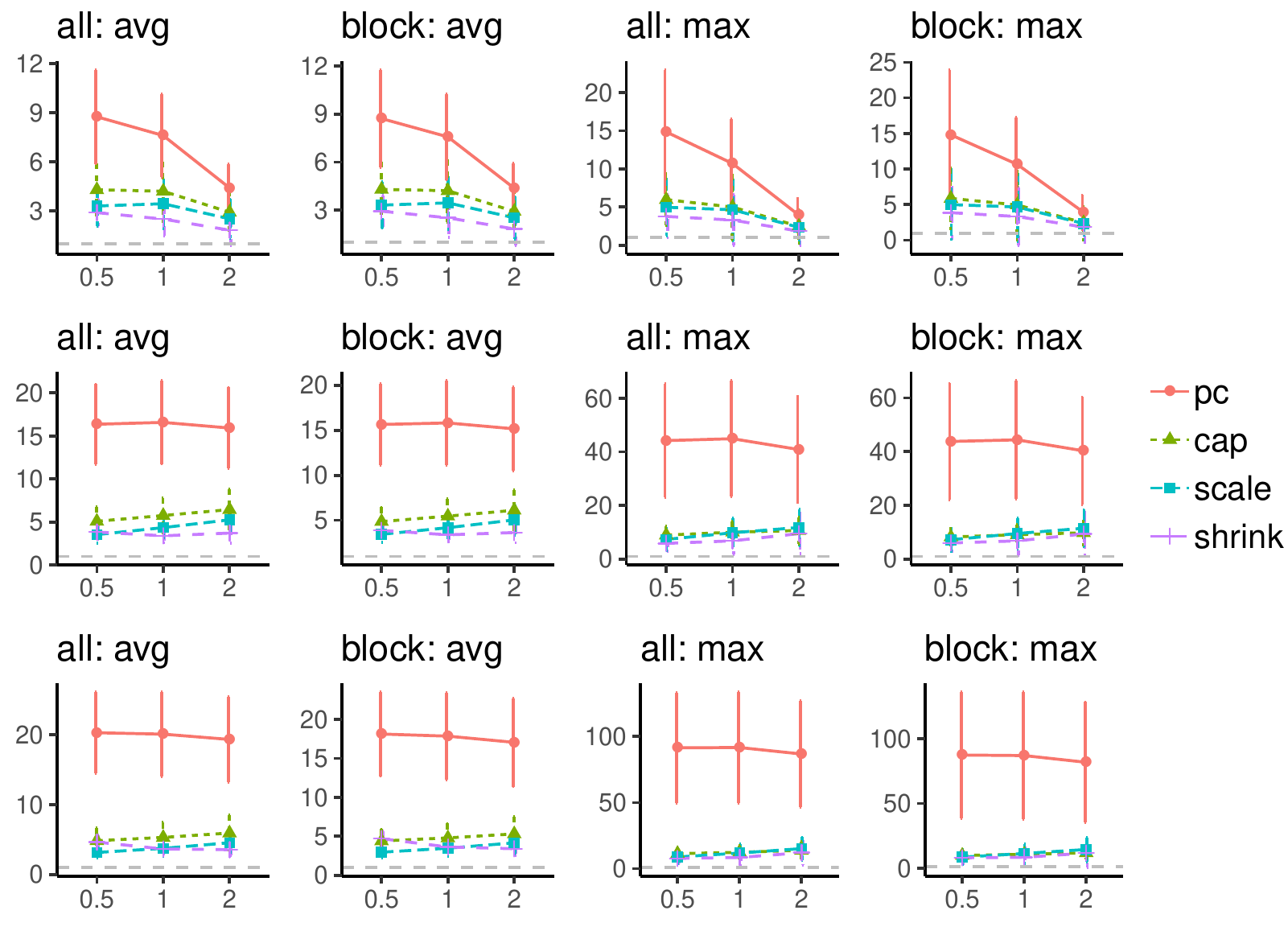}
\caption{\small $\text{err}_{\avg}(\wh\chi^{\circ}_{it})$
and $\text{err}_{\max}(\wh\chi^{\circ}_{it})$
under Model~1 with $T = 2000$.}
\label{fig:sim:err:T2000}
\end{figure}

Next, we evaluate the accuracy of an estimator $\wh\chi^\circ_{it}$ of $\chi_{it}$
relative to that of the oracle estimator, using the following error measures
\begin{align*}
\text{err}_{\avg} = \frac{n^{-1}\sum_{i=1}^n\sum_{t=1}^T(\wh\chi^\circ_{it} - \chi_{it})^2}
{\wh{\E}\{n^{-1}\sum_{i=1}^n\sum_{t=1}^T(\wh\chi^{\oracle}_{it} - \chi_{it})^2\}},
\quad
\text{err}_{\max} = \frac{\max_{1 \le i \le n}\sum_{t=1}^T(\wh\chi^\circ_{it} - \chi_{it})^2}
{\wh{\E}\{\max_{1 \le i \le n}\sum_{t=1}^T(\wh\chi^{\oracle}_{it} - \chi_{it})^2\}},
\end{align*}
where $\wh{\E}$ denotes the average over all Monte Carlo repetitions,
and $\circ$ denotes the use of PC, capped, scaled or shrinkage estimator and their blockwise counterparts.
We note that $\text{err}_{\max}$ is 
specifically to capture the possible deterioration in the estimators for individual $i$
due to the over-estimation of $r$,
and both measures are closely related to the conditions in \eqref{eq:fan:cond}.
Figures \ref{fig:sim:err:T500}--\ref{fig:sim:err:T2000}
show the average and standard deviation
of $\text{err}_{\avg}$ and $\text{err}_{\max}$
over $1000$ Monte Carlo realisations.

Overall, blockwise estimators do not lose efficiency compared to their whole sample counterparts
or, even perform slightly better in terms of the relative efficiency compared to 
the oracle PC estimator.
It is evident that PC estimator exhibits the worst performance in almost all cases,
in terms of both the average and variability of the two different error measures.
Indeed, $\text{err}_{\max}^{\pca}$ indicates
that the PC estimator with an over-estimated factor number
can be worse by hundredfold than the oracle PC estimator for some coordinates.

Capping and scaling lead to considerable improvement
with respect to both error measures, with and without blockwise estimation,
and marginally the scaled PC estimator tends to return smaller estimation errors. We note that $\wh\chi_{it}^{\cs}$ yields the smallest estimation error in many scenarios.
Exceptions occur when $n$ is relatively larger than $T$:
the PC-based estimator of the factor space is expected to be
highly accurate in this setting due to the blessing of dimensionality,
and the bias introduced by eigenvalue shrinkage 
tends to deteriorate the performance of $\wh\chi_{it}^{\cs}$ 
(see Figure~\ref{fig:sim:err:T500s2} in Appendix~\ref{sec:sim:add}).

As the signal-to-noise ratio decreases, the gap between 
the performance of $\wh\chi^{\pca}_{it}$ and our modified estimators gets closer,
as the consistent estimation of $\chi_{it}$ itself becomes more challenging,
i.e., the error due to the over-estimation of $r$ in \eqref{eq:pca_error}
becomes dominated by the first term.
Increasing $n$ also tends to close this gap
as the performance of the estimator of $r$ improves.
This, however, has the opposite effect on $\text{err}_{\max}(\wh\chi_{it}^{\pca})$
since the maximum is taken over the $n$ cross-sections.
In general, $\text{err}_{\avg}$ and $\text{err}_{\max}$ evaluated at modified PC estimators
exhibit much less fluctuations as $n$ and $T$ vary.
Noting the close relationship between $\text{err}_{\avg}(\wh\chi_{it}^{\pca})$ 
and $\text{err}_{\max}(\wh\chi_{it}^{\pca})$
and the conditions in \eqref{eq:fan:cond}, 
the results reported here indicate that the popularly adopted covariance matrix estimator based on 
factor analysis will suffer from the over-estimation of $r$.

\section{Real data analysis}
\label{sec:real}

We consider risk minimisation for a portfolio consisting of 
the log returns of the daily closing values of the stocks comprising the Standard \& Poor's 100
(S\&P100) index between July 2006 and September 2013 
(denoted by $\{x_{it}, \, i \le n, \, t \le T\}$ with $n = 90$ and $T = 1814$)
following the exercise conducted in \cite{lam2016}.
The dataset is available from Yahoo Finance.

As evidence of structural changes has been observed in a similar financial data dataset \citep{bcf2017},
we choose to adopt a rolling window of size $\wt T = 253$ (number of trading days each year)
and evaluate the performance of a portfolio on a monthly basis 
(with $21$ as the number of trading days each month).
At the beginning of each month, different methods are adopted to 
estimate the covariance matrix of stock returns using one year of past returns.
Each portfolio has weights given by
\begin{align*}
\wh{\bm\omega}^{\circ}_k  = 
\arg\!\!\!\!\!\!\!\!\!\!\min_{\bm\omega \in \R^n: \, \bm\omega^\top\mbf 1_n = 1} \bm\omega^\top \wh{\bm\Sigma}^\circ_k \bm\omega
= \frac{\wh{\bm\Sigma}^\circ_k \mbf 1_n}{\mbf 1_n^\top \wh{\bm\Sigma}^\circ_k \mbf 1_n}
\quad \text{for } k = 1, \ldots, M = \lceil (T - \wt T)/21 \rceil,
\end{align*}
where $\wh{\bm\Sigma}^\circ_k$ denotes some covariance matrix estimator
based on the $k$-th rolling window $R_k = [21(k - 1) + 1, 21(k - 1) + \wt T]$,
and $\mbf 1_n$ denotes a vector consisting of $n$ ones.
At the end of each month, 
we compute the total excess return, the out-of-sample variance and the mean Sharpe ratio, given by
\begin{align*}
& \wh\tau(\wh{\bm\Sigma}^\circ) = \sum_{k = 1}^M \sum_{t = 21(k - 1) + 1}^{\min(21k, T - \wt T)} (\wh{\bm\omega}^\circ_k)^\top \mbf x_{\wt T + t},
\\
& \wh\sigma^2(\wh{\bm\Sigma}^\circ) 
= \frac{1}{T - \wt T} \sum_{k = 1}^M \sum_{t = 21(k - 1) + 1}^{\min(21k, T - \wt T)} 
\{(\wh{\bm\omega}^\circ_k)^\top \mbf x_{\wt T + t} - \wh\mu(\wh{\bm\Sigma}^\circ_k)\}^2 \quad \text{and}
\\
& \text{SR}(\wh{\bm\Sigma}^\circ) = \frac{1}{M} \sum_{k = 1}^M \frac{\wh\tau(\wh{\bm\Sigma}^\circ_k)
}{\wh\sigma(\wh{\bm\Sigma}^\circ_k)},
\end{align*}
where $\wh{\tau}(\wh{\bm\Sigma}^\circ_k)$, $\wh{\mu}(\wh{\bm\Sigma}^\circ_k)$ 
and $\wh\sigma^2(\wh{\bm\Sigma}^\circ_k)$ 
denote the total and mean excess return and the out-of-sample variance
calculated for each portfolio from $\wh{\bm\Sigma}^\circ_k$.

For covariance matrix estimation, we consider the following two approaches
that separately estimate the factor-driven and idiosyncratic contributions.
\begin{description}
\item[Exact factor modelling (EFM).]
We force the covariance matrix of the idiosyncratic component to be diagonal and obtain
\begin{align}
\label{eq:efm:pca}
\wh{\bm\Sigma}^{\pca}_k = \wt{T}^{-1} \sum_{t \in R_k} 
(\wh{\bm\chi}^{\pca}_t)^\top \wh{\bm\chi}^{\pca}_t
+ \diag\Big( \wt{T}^{-1} \sum_{t \in R_k} 
(\wh{\bm\vep}^{\pca}_t)^\top \wh{\bm\vep}^{\pca}_t \Big),
\end{align}
where the operator $\diag(\mbf A)$ returns a diagonal matrix with 
the diagonal elements of $\mbf A$ in its diagonal,
$\wh{\bm\chi}^{\pca}_t, \, t \in R_k$ is obtained from 
$\wh{\bm\Gamma}_{x, k} = \wt{T}^{-1} \sum_{t \in R_k} \mbf x_t\mbf x_t^\top$
and $\wh{\bm\vep}^{\pca}_t = \mbf x_t  - \wh{\bm\chi}^{\pca}_t$.
Similarly, we yield
$\wh{\bm\Sigma}^{\capp}_k$, $\wh{\bm\Sigma}^{\sca}_k$ and $\wh{\bm\Sigma}^{\cs}_k$
with $\wh{\bm\chi}^{\capp}_t$, $\wh{\bm\chi}^{\sca}_t$ and $\wh{\bm\chi}^{\cs}_t$
replacing $\wh{\bm\chi}^{\pca}_t$ in \eqref{eq:efm:pca}, respectively.
Also, we consider their blockwise versions
$\wh{\bm\Sigma}^{\bpca}_k$, $\wh{\bm\Sigma}^{\bcapp}_k$, $\wh{\bm\Sigma}^{\bsca}_k$ 
and $\wh{\bm\Sigma}^{\bcs}_k$
with the size of blocks $b = \lceil \log^2\wt{T} \rceil$.

\item[POET \citep{fan13}.] We adopt the POET estimator
\begin{align}
\label{eq:poet}
\wh{\bm\Sigma}^{\poet}_k = \wt{T}^{-1} \sum_{t \in R_k} 
(\wh{\bm\chi}^{\pca}_t)^\top \wh{\bm\chi}^{\pca}_t
+ \mc T\Big( \wt{T}^{-1} \sum_{t \in R_k} 
(\wh{\bm\vep}^{\pca}_t)^\top \wh{\bm\vep}^{\pca}_t\Big),
\end{align}
where $\mc T(\cdot)$ performs an element-wise hard-thresholding 
(except for its diagonals) 
with an adaptively chosen threshold recommended by \cite{fan13}
including a constant selected via cross-validation.
\end{description}

The results obtained for varying $\wh r$ are reported in Table~\ref{table:risk};
note that $\wh r = 6$ is selected by the information criterion of \cite{baing02} applied to the whole data.
We note that $\wh r = 2$ may already be over-estimating the number of factors,
in that the idiosyncratic component estimated with $\wh r = 1$ exhibits a prominent group structure
and this may be falsely detected as a factor via PCA, see Figure~\ref{fig:sp100:idio}.
As demonstrated in Example~\ref{ex:rev:one}, 
possibly highly structured nature of the idiosyncratic component
may lead to some elements of $\wh{\mbf w}_{x, j}, \, j \ge 2$ being (spuriously) large, 
and poor performance of the corresponding PC estimator 
(see also Proposition~\ref{prop:pca} and Example~\ref{ex:rev:two}).
The EFM-based method combined with the PC estimator with $\wh r = 2$
performs the best among all the methods.
With $\wh r = 6$, the POET yields 
large negative total excess returns and large out-of-sample variance,
which confirms our observation in Remark~\ref{rem:poet}
that the covariance (precision) matrix estimation based on factor modelling
is susceptible to the errors arising from factor number estimation.
Overall, the methods based on EFM perform better than the POET according to all measures considered.
Among the modified PC estimators, the one that
applies the largest amount of shrinkage ($\wh{\bm\chi}^{\cs}_t$)
achieves the most consistent performance with regards to the choice of $\wh r$.
Interestingly, the blockwise approach yields marginally better performance
than the corresponding whole sample counterparts.

\begin{figure}[htb]
\centering
\includegraphics[width=1\textwidth]{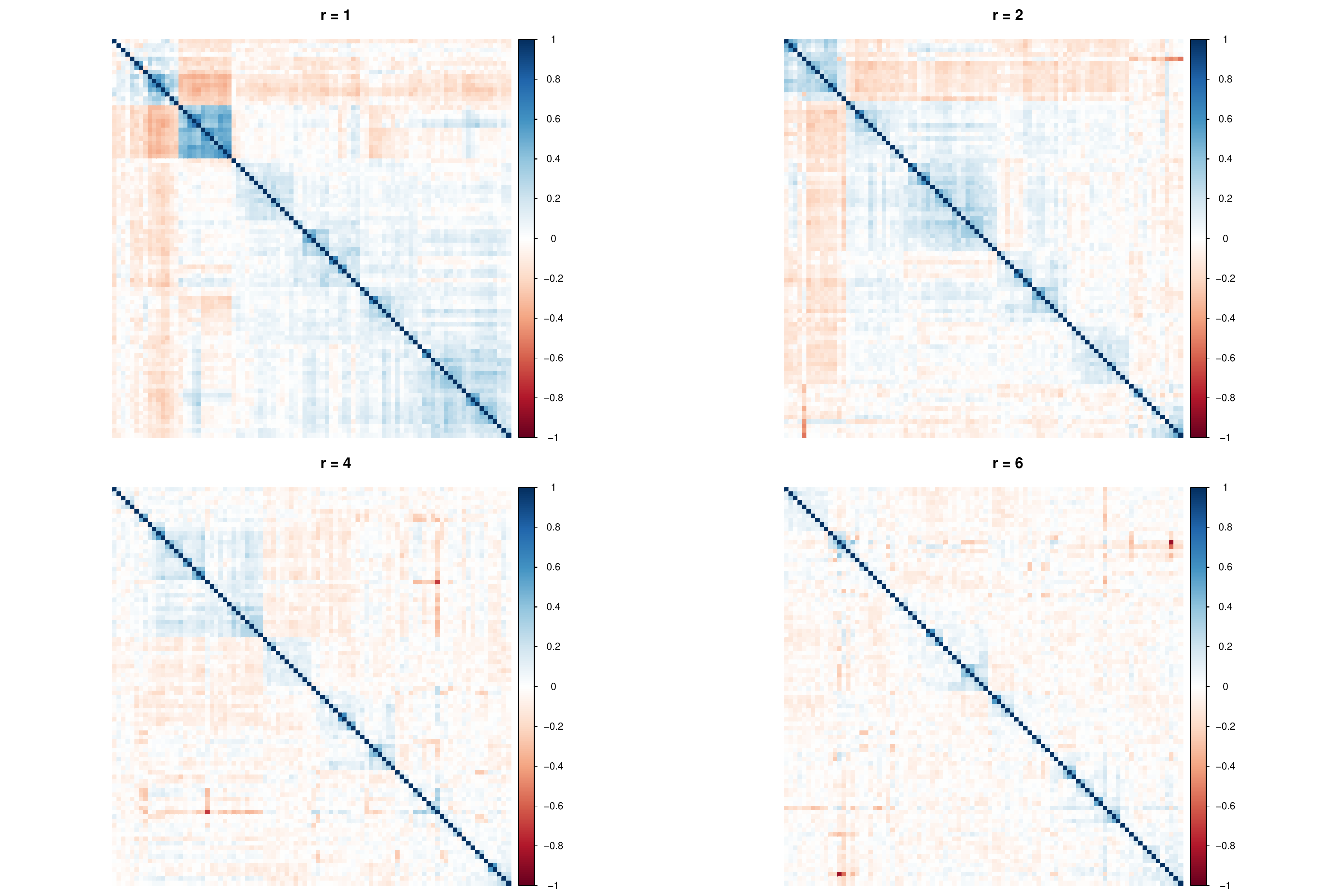}
\caption{\small Heatmap of the correlation matrix of the idiosyncratic component estimated 
with $\wh r \in \{1, 2, 4, 6\}$ (from left to right, top to bottom).
The variables are ordered via hierarchical clustering with the complete linkage method.}
\label{fig:sp100:idio}
\end{figure}

\begin{table}[htb]
\caption{Performance of portfolios constructed with different covariance estimators.}
\label{table:risk}
\resizebox{\columnwidth}{!}
{\small
\begin{tabular}{cc|ccc||ccc||ccc}
\hline
\hline
& & \multicolumn{3}{c||}{$\wh r = 2$} & \multicolumn{3}{c||}{$\wh r = 4$} & \multicolumn{3}{c}{$\wh r = 6$} \\
method & & $\wh\tau$ & $\wh\sigma^2$ & SR& $\wh\tau$ & $\wh\sigma^2$ & SR & $\wh\tau$ & $\wh\sigma^2$ & SR \\
\hline
EFM & $\wh{\bm\chi}^{\pca}_t$ & 25.032 & 0.917 & 0.932 &-28.321 &  0.936  & 0.003 & -22.515 &  0.818 &  0.460 \\ 
 & $\wh{\bm\chi}^{\capp}_t$ & 22.513 & 0.922 & 0.865 &-31.045 &  0.883  &-0.028  &-21.691 &  0.753 &  0.439 \\ 
 & $\wh{\bm\chi}^{\sca}_t$ & 20.164 & 0.900 & 0.888 &-24.468   &0.832  & 0.133& -27.461 &  0.742 &  0.301 \\ 
 & $\wh{\bm\chi}^{\cs}_t$ & 14.072 & 0.890 & 0.809 &4.079 &0.861& 0.652 & 4.618 & 0.835 & 0.703  \\ 
 & $\wh{\bm\chi}^{\bpca}_t$ & 21.288 & 0.835 & 0.907 &-17.115&   0.813  & 0.331& -14.934 &  0.750 &  0.468 \\ 
 & $\wh{\bm\chi}^{\bcapp}_t$ & 20.807 & 0.842 & 0.888 &-19.331&   0.818 &  0.270& -14.538 &  0.756 &  0.478 \\ 
 & $\wh{\bm\chi}^{\bsca}_t$ & 20.626 & 0.833  & 0.886 &-16.133  & 0.804&   0.367& -13.939  & 0.760 &  0.469 \\ 
 & $\wh{\bm\chi}^{\bcs}_t$ & 18.866 & 0.828 & 0.859 &10.638  &0.825 & 0.766& 10.605 & 0.822 & 0.777 \\ 
\hline
POET & &  -76.446 & 6.061&-0.315  &-299.643  & 69.678 &  -0.331 & -355.691 &  85.428  &  0.254 \\
\hline
\end{tabular}}
\end{table}

\section{Conclusions}
\label{sec:conc}

Factor number estimation is a challenging task due to the lack of a clear gap in empirical eigenvalues,
and various estimators tend to over-estimate $r$
in the presence of moderate cross-sectional correlations in the idiosyncratic component.
In this paper, we make the first attempt at establishing the non-negligibility of
the error due to the over-estimation of $r$ in the widely adopted PC estimator.
In doing so, we propose a novel blockwise estimation technique,
which enables rigorous treatment of this over-estimation error under 
a time series factor model.
Also, we propose the modified PC estimators which are easily implemented
while performing as well as the oracle PC estimator with known $r$,
and verify this via extensive simulation studies.
In practice, we recommend the use of $\wh\chi_{it}^{\shr}$
unless $n$ is much greater than $T$ (an unlikely setting for e.g., economic and financial data),
which shows very good practical performance both on simulated and real-life datasets.

\bibliographystyle{asa}
\bibliography{fbib}

\clearpage

\appendix

\section{Proofs of the main results}
\label{sec:appA}
\subsection{Preliminaries}
\label{sec:prem}

The following lemmas hold under Assumptions~\ref{assum:id}--\ref{assum:dep}.
Their proof can be found in Appendix~\ref{sec:pf:prem}.

\begin{lem}
\label{lem:four}
$\max_{1 \le i \le n} \max_{1 \le t \le T} |\chi_{it}| = O_p(\log(T))$, 
$\max_{1 \le i \le n} \max_{1 \le t \le T} |\vep_{it}| = O_p(\log(T))$,
and $\max_{1 \le i \le n} \max_{1 \le t \le T} |x_{it}| = O_p(\log(T))$.
\end{lem}

\begin{lem}
\label{lem:block:cov} Let $b_T$ satisfy $b_T \to \infty$ and $T^{-1} b_T \to 0$,
and $L_T = \lceil T/b_T \rceil$.
\begin{compactenum}
\item[(i)] $\max_{1 \le \ell \le L_T} 
n^{-1}\Vert \wh{\bm\Gamma}^{(\ell)}_x - \bm\Gamma_{\chi}\Vert =
O_p\Big(\sqrt{\frac{\log(n)}{T}} \vee \frac{1}{n}\Big)$.
\item[(ii)] $\max_{1 \le \ell \le L_T} \max_{1 \le i \le n} n^{-1/2}\Vert \bm\varphi_i^\top(\wh{\bm\Gamma}_x^{(\ell)} - \bm\Gamma_{\chi}) \Vert =
 O_p\l(\sqrt{\frac{\log(n)}{T}} \vee \frac{1}{\sqrt n}\r)$.
\end{compactenum}
Also, there exists an orthonormal matrix $\mbf S_\ell \in \R^{r \times r}$
which, for $\wh{\mbf W}_x^{(\ell)} = [\wh{\mbf w}_{x, j}^{(\ell)}, \, j \le r]$, satisfies
\begin{compactenum}
\item[(iii)]
$\max_{1 \le \ell \le L_T} 
\Vert \wh{\mbf W}^{(\ell)}_x - \mbf W_{\chi}\mbf S_\ell \Vert
= O_p\l(\sqrt{\frac{\log(n)}{T}} \vee \frac{1}{n}\r)$;
\item[(iv)]
$\max_{1 \le \ell \le L_T} \max_{1 \le i \le n}
\sqrt n\,\Vert\bm\varphi_i^\top(\wh{\mbf W}^{(\ell)}_x - \mbf W_{\chi}\mbf S_\ell)\Vert
= O_p\l(\sqrt{\frac{\log(n)}{T}} \vee \frac{1}{\sqrt n}\r)$.
\end{compactenum}
\end{lem}

\begin{lem} \hfill
\label{lem:cov} 
\begin{compactenum}
\item[(i)] $n^{-1}\Vert \wh{\bm\Gamma}_x - \bm\Gamma_{\chi}\Vert =
O_p\Big(\sqrt{\frac{\log(n)}{T}} \vee \frac{1}{n}\Big)$.
\item [(ii)] $\max_{1 \le i \le n} n^{-1/2}\Vert \bm\varphi_i^\top(\wh{\bm\Gamma}_x - \bm\Gamma_{\chi}) \Vert =
O_p\l(\sqrt{\frac{\log(n)}{T}} \vee \frac{1}{\sqrt n}\r)$.
\item [(iii)] $\max_{1 \le i \le n} \sqrt n\,\Vert\bm\varphi_i^\top(\wh{\mbf W}_x-\mbf W_{\chi}\mbf S)\Vert
= O_p\l(\sqrt{\frac{\log(n)}{T}} \vee \frac{1}{\sqrt n}\r)$
for some orthonormal $r\times r$ matrix $\mbf S$.
\end{compactenum}
\end{lem}

\begin{lem}
\label{lem:block:vep} 
{\it Let $\ell(t)$ denote the index of the block for which $t \in I_{\ell(t)}$,
and $b_T = \log^{1/\beta + \delta} T$ for some $\delta > 0$.
Then, 
$\max_{1 \le t \le T} \vert (\wh{\mbf w}_{x, j}^{\ell(t)})^\top \bm\vep_t \vert = O_p(\log(T)).
$}
\end{lem}

\subsection{Proof of Proposition~\ref{prop:pca}}
\label{sec:pf:prop:pca}

Recall the definition of $\ell(t)$ in Lemma~\ref{lem:block:vep}.
Note that
\begin{align*}
& \max_{1 \le i \le n} \max_{1 \le t \le T} |\wh\chi_{it}^{\bpca} - \chi_{it}|
\le 
\max_{1 \le i \le n} \max_{1 \le t \le T} \Big\vert \sum_{j=1}^r \wh w^{\ell(t)}_{x, ij} 
(\wh{\mbf w}^{\ell(t)}_{x, j})^\top \mbf x_t - \chi_{it}\Big\vert
\\
& + 
\max_{1 \le i \le n} \max_{1 \le t \le T} \Big\vert \sum_{j=r+1}^{\wh r} \wh w^{\ell(t)}_{x, ij} 
(\wh{\mbf w}^{\ell(t)}_{x, j})^\top \bm\chi_t \Big\vert
+ 
\max_{1 \le i \le n} \max_{1 \le t \le T} \Big\vert \sum_{j=r+1}^{\wh r} \wh w^{\ell(t)}_{x, ij} 
(\wh{\mbf w}^{\ell(t)}_{x, j})^\top \bm\vep_t \Big\vert
\\
& =: I + II + III. 
\end{align*}

From Lemmas~\ref{lem:four} and \ref{lem:block:cov}~(iii)--(iv), 
using the analogous arguments as those adopted 
in the proof of Proposition~\ref{thm:common} in Appendix~\ref{sec:pf:prop:one}.
\begin{align}
I \le&
\max_{1 \le i \le n}\max_{1 \le t \le T} 
\vert \bm\varphi_i^{\top}\wh{\mbf W}^{\ell(t)}_x(\wh {\mbf W}_x^{\ell(t)})^\top\mbf x_t -
\bm\varphi_i^{\top}\mbf W_\chi\mbf S_{\ell(t)}(\wh{\mbf W}_x^{\ell(t)})^\top\mbf x_t\vert
\nn \\ 
& + \max_{1 \le i \le n}\max_{1 \le t \le T} 
\vert \bm\varphi_i^{\top}\mbf W_\chi\mbf S_{\ell(t)}(\wh {\mbf W}_x^{\ell(t)})^\top\mbf x_t -
\bm\varphi_i^{\top}\mbf W_\chi\mbf W_\chi^\top\mbf x_t\vert 
+ \max_{1 \le i \le n}\max_{1 \le t \le T} 
\vert \bm\varphi_i^{\top}\mbf W_\chi\mbf W_\chi^\top\bm\vep_t\vert
\nn \\
=&  O_p\l\{\l(\sqrt{\frac{\log(n)}{T}} \vee \frac{1}{\sqrt n}\r)\log(T)\r\}.
\label{eq:block:oracle:err}
\end{align}

Let $\wh{\mbf W}^{\ell(t)}_{x, (r+1):k}=[\wh{\mbf w}^{\ell(t)}_{x, j}, \, r + 1 \le j \le k]$.
Under Assumption~\ref{assum:id}~(i), it follows that
$\mbf W_\chi^\top \bm\Lambda\bm\Lambda^\top \mbf W_\chi = \mbf M_\chi$
and hence $\mbf W_\chi$ may be regarded as the left singular vectors of $\bm\Lambda$.
Then, from the orthogonality of eigenvectors, \ref{eq:c1} and Lemma~\ref{lem:block:cov}~(iii), 
\begin{align}
&\max_{1 \le \ell \le L_T} \Vert (\wh{\mbf W}^{(\ell)}_{x, (r+1):k})^\top \bm\Lambda \Vert 
= \max_{1 \le \ell \le L_T} \Vert (\wh{\mbf W}^{(\ell)}_{x, (r+1):k})^\top \mbf W_\chi \mbf M_\chi^{1/2} \Vert
\nn 
\\
&\le \max_{1 \le \ell \le L_T} \Vert (\wh{\mbf W}^{(\ell)}_{x, (r+1):k})^\top (\mbf W_\chi\mbf S_\ell - \wh{\mbf W}_x^{(\ell)}) \Vert \; \Vert \mbf M_\chi^{1/2} \Vert
= O_p\l(\sqrt{\frac{n\log(n)}{T}} \vee \frac{1}{\sqrt n}\r)
\label{eq:lem:lamb}
\end{align}
for any fixed $k \ge r+1$.
Together with the condition \eqref{eq:prop:pca} and Lemma~\ref{lem:four}, 
it yields
\begin{align*}
II &=& 
O_p\l\{n^{-\alpha/2} \cdot \Big(\sqrt{\frac{n\log(n)}{T}} \vee \frac{1}{\sqrt n}\Big) \cdot \log(T) \r\}
= O_p\l\{\Big(\sqrt{\frac{n^{1 - \alpha}\log(n)}{T}} \vee \sqrt{\frac{1}{n^{(1 + \alpha)}}}\Big)\log(T)\r\}.
\end{align*}
Finally, from Lemma~\ref{lem:block:vep},
$III = O_p(n^{-\alpha/2}\log(T))$,
and the conclusion follows.

\subsection{Proof of Proposition~\ref{prop:scaled}}
\label{sec:pf:prop:scaled}

Recall the definition of $\ell(t)$ in Lemma~\ref{lem:block:vep}.
Note that
\begin{align*}
& \max_{1 \le i \le n} \; \max_{1 \le t \le T} \vert \wh\chi_{it}^{\bsca} - \chi_{it} \vert
\le
\max_{1 \le i \le n} \; \max_{1 \le t \le T} 
\Big\vert \sum_{j=1}^r \wh w^{\sca, \ell(t)}_{x, ij} (\wh{\mbf w}^{\sca, \ell(t)}_{x, j})^\top 
\mbf x_t - \chi_{it} \Big\vert
\\
& + \max_{1 \le i \le n} \; \max_{1 \le t \le T} \Big\vert  
\sum_{j=r+1}^{\wh r} \wh w^{\sca, \ell(t)}_{x, ij} (\wh{\mbf w}^{\sca, \ell(t)}_{x, j})^\top \bm\chi_t
\Big\vert
+\max_{1 \le i \le n} \; \max_{1 \le t \le T} \Big\vert
\sum_{j=r+1}^{\wh r} \wh w^{\sca, \ell(t)}_{x, ij} (\wh{\mbf w}^{\sca,\ell(t)}_{x, j})^\top
\bm\vep_t \Big\vert
\\
&=: I + II + III.
\end{align*}
Since scaling does not alter the $r$ leading eigenvectors with probability tending to one,
thanks to the arguments leading to \eqref{eq:bounded:ew} and Lemma~\ref{lem:block:cov}, 
we derive that
$I = O_p\{(\sqrt{\log(n)/T} \vee 1/\sqrt n)\log(T)\}$ as in \eqref{eq:block:oracle:err}.
Next, due to the orthogonality of $\wh{\mbf w}^{\sca, \ell(t)}_{x, j}, \, j \le \wh r$,
\begin{align*}
II &= \max_{1 \le i \le n} \; \max_{1 \le t \le T} 
\Big\vert \sum_{j = r+1}^{\wh r} \wh w^{\sca, \ell(t)}_{x, ij} (\wh{\mbf w}^{\sca, \ell(t)}_{x, j})^\top 
\{\bm\chi_t - \wh{\mbf W}^{\ell(t)}_{x, 1:r}(\wh{\mbf W}^{\ell(t)}_{x, 1:r})^\top\mbf x_t\} \Big\vert
\\
&=  O_p\l\{\Big(\sqrt{\frac{\log(n)}{T}} \vee \frac{1}{\sqrt n}\Big)\log(T)\r\}
\end{align*}
from the bound on $I$ and the uniform boundedness of $|\wh w^{\sca, \ell(t)}_{x, ij}|$.
Finally, Lemma~\ref{lem:block:vep} and the definition of $|\wh w^{\sca, \ell(t)}_{x, ij}|$ yield 
$III = O_p(\log(T)/\sqrt n)$, which concludes the proof.

\clearpage

\numberwithin{equation}{section}
\renewcommand*{\thelem}{B.\arabic{lem}}
\renewcommand{\thefigure}{C.\arabic{figure}}

\setcounter{figure}{0}

\section{Further proofs}

\subsection{Proof of the results in Section~\ref{sec:prem}}
\label{sec:pf:prem}

\begin{proof}[Proof of Lemma~\ref{lem:four}]
Assumptions~\ref{assum:id}~(iii) and \ref{assum:tail} and Proposition~2.7.1 in \citet{vershynin2018} yield
\begin{align*}
\p\Big(\max_{1 \le i \le n} \max_{1 \le t \le T} |\chi_{it}| > C\log(T)\Big)
\le 
\p\Big(r\bar{\lambda}\max_{1 \le j \le r}\max_{1 \le t \le T} |f_{jt}| > C\log(T)\Big)
\le 2rT \exp(-C\log(T)/B_f) \to 0
\end{align*}
for some fixed $C > B_f$.
Similarly, from Assumptions~\ref{assum:nt}~and~\ref{assum:tail}, there exists some fixed 
$C^\prime > B_\vep(\kappa + 1)$
such that
\begin{align*}
\p\Big(\max_{1 \le i \le n} \max_{1 \le t \le T} |\vep_{it}| > C^\prime\log(T)\Big) \le 
2nT \exp(-C^\prime\log(T)/B_\vep) \to 0.
\end{align*}
From the above, the third statement follows.
\end{proof}

\begin{proof}[Proof of Lemma~\ref{lem:block:cov}]
For some $\ell \le L_T$,
\begin{align*}
& \max_{1 \le i, i' \le n} \l\vert \frac{1}{|\bar{I}_\ell|}\sum_{t \in \bar{I}_\ell} x_{it}x_{i't} - 
\E\Big(\frac{1}{|\bar{I}_\ell|}\sum_{t \in \bar{I}_\ell} x_{it}x_{i't}\Big)\r\vert
\le
\max_{1 \le j, j' \le r} r^2\bar{\lambda}^2 \l\vert \frac{1}{|\bar{I}_\ell|}\sum_{t \in \bar{I}_\ell} f_{jt}f_{j't} - 
\E\Big(\frac{1}{|\bar{I}_\ell|}\sum_{t \in \bar{I}_\ell} f_{jt}f_{j't}\Big)\r\vert
\\
& +
\max_{1 \le i, i' \le n} \l\vert \frac{1}{|\bar{I}_\ell|}\sum_{t \in \bar{I}_\ell} \vep_{it}\vep_{i't} - 
\E\Big(\frac{1}{|\bar{I}_\ell|}\sum_{t \in \bar{I}_\ell} \vep_{it}\vep_{i't}\Big)\r\vert
+
2\max_{\substack{1 \le j \le r \\ 1 \le i \le n}} r\bar{\lambda}
\l\vert \frac{1}{|\bar{I}_\ell|}\sum_{t \in \bar{I}_\ell} f_{jt}\vep_{it}\r\vert
=: I + II + III.
\end{align*}
Under Assumptions~\ref{assum:nt}~and~\ref{assum:tail}~(i), Lemma~A.2 of \cite{fan2011}
indicates that there exist some fixed $B > 0$ such that
$\p(|\vep_{it}\vep_{i't}| > u) \le \exp[1 - (u/B)^{1/2}]$ for any $u > 0$.
Then by Theorem~1 of \cite{merlevede2011}, Bonferroni correction
and that $|\bar{I}_\ell| \ge T(1 - 3T^{-1}b_T)$, we yield
\begin{align*}
\p\l(II \ge C\sqrt{\frac{\log(n)}{T}}\r) &\le 
n^2T\l\{ T\exp\l[-\frac{(C^2T\log(n))^{\gamma/2}}{C_1}\r]
+ \exp\l[-\frac{C^2T\log(n)}{C_2(1 + C_3 T)}\r] \r.
\\
& \l. + \exp\l[-\frac{C^2\log(n)}{C_4}
\exp\Big(\frac{(C^2T\log(n))^{\gamma(1 - \gamma)/2}}{C_5\log^{\gamma}T}\Big)\r]
\r\} = o\l(\frac{1}{T}\r)
\end{align*}
for sufficiently large but fixed $C > 0$,
where $\gamma = (\beta^{-1} + 2)^{-1} \in (0, 1)$ is defined with $\beta$ in Assumption~\ref{assum:dep}, 
and  $C_k, \, 1 \le k \le 5$ are fixed constants.
We can similarly show that $I, III = O_p(\sqrt{\log(n)/T})$ and, moreover,
the bounds for each $\ell$ hold with probability tending to one at the rate $o(T^{-1})$,
and therefore the uniform bound over $\ell = 1, \ldots, L_T$ follows. 
Together with Assumption~\ref{assum:id}~(iv), we have
\begin{align*}
\max_{1 \le \ell \le L_T} \frac 1 n\Vert \wh{\bm\Gamma}_x^{(\ell)} - \bm\Gamma_{\chi}\Vert \le
\max_{1 \le \ell \le L_T} \frac 1 n\Vert \wh{\bm\Gamma}_x^{(\ell)} - \E(\wh{\bm\Gamma}_{x}^{(\ell)}) \Vert_F  +
\max_{1 \le \ell \le L_T} \frac 1 n\Vert \E(\wh{\bm\Gamma}_\vep^{(\ell)})\Vert 
= O_p\l(\sqrt{\frac{\log(n)}{T}} \vee \frac{1}{n}\r),
\end{align*}
where $\wh{\bm\Gamma}_\vep^{(\ell)}$ is defined analogously as
$\wh{\bm\Gamma}_x^{(\ell)}$.
For (ii), we now deal with an $n$-dimensional vector such that
\begin{align*}
\frac{1}{\sqrt n} \Vert \bm\varphi_i^\top(\wh{\bm\Gamma}^{(\ell)}_x - \bm\Gamma_\chi) \Vert
\le 
\frac{1}{\sqrt n} \Vert \bm\varphi_i^\top(\wh{\bm\Gamma}^{(\ell)}_x - \bm\Gamma_x) \Vert
+ \frac{1}{\sqrt n} \Vert \bm\Gamma_\vep \Vert
= O_p\l(\sqrt{\frac{\log(n)}{T}} \vee \frac{1}{\sqrt n} \r).
\end{align*}
Part (iii) is proved using part (i) and Davis-Kahan theorem as in \eqref{eq:davisk}.
For (iv), first note that
as a consequence of (i) and Weyl's inequality, 
$\wh\mu_{x, j}^{(\ell)}$, the $j$-th largest eigenvalue of $\wh{\bm\Gamma}_x^{(\ell)}$, satisfy
\begin{align*}
\frac 1 n|\wh\mu_{x, j}^{(\ell)}-\mu_{\chi, j}|\leq \frac 1 n
\Vert \wh{\bm\Gamma}_x^{(\ell)}-\bm\Gamma_{\chi}\Vert 
= O_p\l(\sqrt{\frac{\log(n)}{T}} \vee \frac{1}{n}\r), \quad j=1, \ldots, r.
\end{align*}
Then it follows that $\wh{\mu}_{x, r}^{(\ell)}/n \ge \underline c_r + O_p(\sqrt{\frac{\log n}{T}} \vee \frac{1}{n})$,
and therefore
\begin{align}
\label{eq:M:one}
\Big\Vert\Big(\frac{\mbf M_{\chi}}{n}\Big)^{-1}\Big\Vert = \frac{n}{\mu_{\chi, r}} = O(1),
\qquad\Big\Vert\Big(\frac{\wh{\mbf M}_{x}^{(\ell)}}{n}\Big)^{-1}\Big\Vert = \frac{n}{\wh{\mu}_{x, r}} = O_p(1),
\end{align}
for $\wh{\mbf M}_{x}^{(\ell)} = \diag(\wh{\mu}_{x, 1}^{(\ell)}, \ldots, \wh{\mu}_{x, r}^{(\ell)})$.
Besides,
\begin{align}
& \Big\Vert\Big(\frac{\wh{\mbf M}_{x}^{(\ell)}}{n}\Big)^{-1 } - 
\Big(\frac{\mbf M_{\chi}}{n}\Big)^{-1}\Big\Vert 
\le \sqrt{\sum_{j=1}^r\Big(\frac{n}{\wh{\mu}_{x, j}^{(\ell)}}-\frac{n}{\mu_{\chi, j}}\Big)^2} 
\leq \sum_{j=1}^r n\Big|\frac{\wh{\mu}_{x, j}^{(\ell)}-\mu_{\chi, j}}{\wh{\mu}_{x, j}^{(\ell)}\mu_{\chi, j}}\Big| 
\nn \\
\le& \frac{r\max_{1 \le j \le r}|\wh{\mu}_{x, j}^{(\ell)}-\mu_{\chi, j}|}
{n\underline c_r^2+O_p\Big(n\sqrt{\frac{\log(n)}{T}} \vee 1\Big)}
= O_p\l(\sqrt{\frac{\log(n)}{T}} \vee \frac{1}{n}\r). \label{eq:M:two}
\end{align}
Using (ii)--(iii), \eqref{eq:M:one}--\eqref{eq:M:two}, and that
$\Vert\mbf W_{\chi}\mbf S\Vert=1$ and $\Vert\bm\varphi_i^{\top}\bm\Gamma_{\chi}\Vert=O(\sqrt n)$, 
we yield
\begin{align}
&\sqrt n\Vert\bm\varphi_i^\top(\wh{\mbf W}_{x}^{(\ell)}- \mbf W_{\chi}\mbf S)\Vert=
\frac 1 {\sqrt n}\Big\Vert
 {\bm\varphi_i^\top} \Big\{\wh{\bm\Gamma}_{x}^{(\ell)}\wh{\mbf W}_{x}^{(\ell)}\Big(\frac{\wh{\mbf M}_{x}^{(\ell)}}{n}\Big)^{-1}
-\bm\Gamma_{\chi}\mbf W_{\chi}\mbf S\Big(\frac{\mbf M_{\chi}}{n}\Big)^{-1}\Big\}
\Big\Vert\nn\\
\leq& \frac 1 {\sqrt n}\Vert
\bm\varphi_i^\top\big(\wh{\bm\Gamma}_{x}^{(\ell)}-\bm\Gamma_{\chi}\big)\Vert\,
\Big\Vert\Big(\frac{\mbf M_{\chi}}{n}\Big)^{-1}\Big\Vert+
\frac 1 {\sqrt n}\Vert
\bm\varphi_i^\top\bm\Gamma_{\chi}\Vert\,\Big\Vert\Big(\frac{\wh{\mbf M}_{x}^{(\ell)}}{n}\Big)^{-1}-\Big(\frac{\mbf M_{\chi}}{n}\Big)^{-1}\Big\Vert\nn\\ 
+ &\frac 1 {\sqrt n}\Vert
\bm\varphi_i^\top\bm\Gamma_{\chi}\Vert\,
\Big\Vert\Big(\frac{\mbf M_{\chi}}{n}\Big)^{-1}\Big\Vert\,\Vert
\wh{\mbf W}_{x}^{(\ell)}-\mbf W_{\chi}\mbf S\big\Vert + o_p\l( \sqrt{\frac{\log(n)}{T}} \vee \frac 1 {\sqrt n}\r)
= O_p\l( \sqrt{\frac{\log(n)}{T}} \vee \frac{1}{\sqrt n}\r).\nn
\end{align} 
\end{proof}

\begin{proof}[Proof of Lemma~\ref{lem:cov}]
The proof follows the same arguments as that of Lemma~\ref{lem:block:cov}~(i)--(ii) and (iv).
\end{proof}

\begin{proof}[Proof of Lemma~\ref{lem:block:vep}]
Let $\mc F^{(\ell)} = \sigma\{(\mbf f_t, \bm\vep_t), \, t \in \bar{I}_\ell\}$,
the $\sigma$-algebra generated by $(\mbf f_t, \bm\vep_t), \, t \in \bar{I}_\ell$.
Note that 
\begin{align*}
& \max_{1 \le t \le T} \E\vert (\wh{\mbf w}_{x, j}^{\ell(t)})^\top \bm\vep_t \vert^2
= \max_{1 \le t \le T} \E\l[ (\wh{\mbf w}_{x, j}^{\ell(t)})^\top \bm\vep_t \bm\vep_t^\top \wh{\mbf w}_{x, j}^{\ell(t)} \r]
\le \E\l[ \max_{1 \le t \le T} (\wh{\mbf w}_{x, j}^{\ell(t)})^\top \E\big(\bm\vep_t \bm\vep_t^\top\big) 
\wh{\mbf w}_{x, j}^{\ell(t)} \r]
\\
& + \max_{1 \le t \le T} \E\l\vert (\wh{\mbf w}_{x, j}^{\ell(t)})^\top 
\Big[ \E\big(\bm\vep_t \bm\vep_t^\top | \mc F^{\ell(t)}\big) - 
\E\big(\bm\vep_t \bm\vep_t^\top\big) \Big]
\wh{\mbf w}_{x, j}^{\ell(t)} \r\vert
:= I + II.
\end{align*}
From Assumption~\ref{assum:id}~(iv) and the normalisation of $\wh{\mbf w}_{x, j}^{\ell(t)}$, 
we have $I < C_\vep$.
Next, 
\begin{align*}
\vert II \vert & \le 
\max_{1 \le t \le T} \Big\Vert  \E\big(\bm\vep_t \bm\vep_t^\top | \mc F^{\ell(t)}\big) 
- \E\big(\bm\vep_t \bm\vep_t^\top\big) \Big\Vert
\le
n \max_{1 \le i, i' \le n} \max_{1 \le t \le T} \l\vert \E(\vep_{it}\vep_{i't} | \mc F^{\ell(t)}) 
- \E(\vep_{it}\vep_{i't}) \r\vert
\\
& \le  6 n \exp\l[ - \frac{c_\alpha}{2}\Big(\log^{1/\beta + \delta} T \Big)^\beta \r] \cdot
\l[\max_{1 \le i \le n}\max_{1 \le t \le T} \E(\vep_{it}^4) \r]^{1/2} \to 0
\end{align*}
as $n, T \to \infty$ under Assumption~\ref{assum:nt},
where the second inequality follows from Theorem 14.2 of \cite{davidson1994},
Assumptions~\ref{assum:tail} (ii) and \ref{assum:dep}
and that $\min_{t \in I_\ell} \min_{u \in \bar{I}_\ell} |t - u| \ge b_T = \log^{1/\beta + \delta} T$.
In other words, $\Var((\wh{\mbf w}^{\ell(t)}_{x, j})^\top\bm\vep_t)$
is bounded uniformly in $t$ for large $T$.
This, together with Assumption~\ref{assum:tail}~(ii) 
and Proposition 2.7.1 of \citet{vershynin2018},
completes the proof.
\end{proof}

\subsection{Proof of Proposition~\ref{thm:common}}
\label{sec:pf:prop:one}

Note that 
$\wh{\chi}_{it}^{\pca} = \bm\varphi_i^{\top}\wh{\mbf W}_x\wh {\mbf W}_x^\top\mbf x_t$ and
$\chi_{it} 
= \bm\varphi_i^{\top} \mbf W_\chi \mbf W_\chi^\top\bm \chi_t$.
Hence,
\begin{align}
\max_{1 \le i \le n}\max_{1 \le t \le T} \vert \wh{\chi}_{it}^{\pca} - \chi_{it} \vert 
\le &
\max_{1 \le i \le n}\max_{1 \le t \le T} 
\vert \bm\varphi_i^{\top}\wh{\mbf W}_x\wh {\mbf W}_x^\top\mbf x_t - 
\bm\varphi_i^{\top}\mbf W_\chi\mbf W_\chi^\top\mbf x_t\vert 
\nn \\
&
+ \max_{1 \le i \le n}\max_{1 \le t \le T} 
\vert \bm\varphi_i^{\top}\mbf W_\chi\mbf W_\chi^\top\bm\vep_t\vert
= I + II.
\label{eq:chi1}
\end{align}
For $I$, we have 
\begin{align}
I \le &
\max_{1 \le i \le n}\max_{1 \le t \le T} 
\vert \bm\varphi_i^{\top}\wh{\mbf W}_x\wh {\mbf W}_x^\top\mbf x_t -
\bm\varphi_i^{\top}\mbf W_\chi\mbf S\wh{\mbf W}_x^\top\mbf x_t\vert
+ 
\max_{1 \le i \le n}\max_{1 \le t \le T} 
\vert \bm\varphi_i^{\top}\mbf W_\chi\mbf S\wh {\mbf W}_x^\top\mbf x_t -
\bm\varphi_i^{\top}\mbf W_\chi\mbf W_\chi^\top\mbf x_t\vert 
\nn \\
\le & \max_{1 \le i \le n} \Vert\bm\varphi_i^\top(\wh{\mbf W}_x - \mbf W_\chi\mbf S)\Vert\,
\Vert \wh{\mbf W}_x\Vert \, \max_{1 \le t \le T} \Vert \mbf x_t\Vert
+
\max_{1 \le i \le n} \Vert\bm\varphi_i^\top\mbf W_\chi\Vert \,
\Vert\wh{\mbf W}_x\mbf S-\mbf W_\chi \Vert \, \max_{1 \le t \le T} \Vert \mbf x_t\Vert 
\nn \\
= &  O_p\l\{\l(\sqrt{\frac{\log(n)}{T}} \vee \frac{1}{\sqrt n}\r)\log(T)\r\},
\label{eq:chi3}
\end{align}
from \eqref{eq:davisk}--\eqref{eq:bounded:w}, Lemma~\ref{lem:cov}~(iii) and~\ref{lem:four}. As for $II$, due to normalisation of the eigenvectors, we invoke Assumption \ref{assum:id} (iv):
\begin{align}
\nn
\E(\Vert{\mbf W}_{\chi}^{\top}\bm\vep_t\Vert^2) = \sum_{j=1}^r \E\{ (\mbf w_{\chi, j}^\top \bm\vep_t)^2 \} 
= \sum_{j=1}^r\sum_{i, i'=1}^n w_{\chi, ij} w_{\chi,i'j}\E(  \vep_{it}\vep_{i't} ) < rC_\vep.
\end{align}
Then, by Proposition~2.7.1 of \cite{vershynin2018} and Bonferroni correction,
$\max_{1 \le t \le T} \Vert{\mbf W}_{\chi}^{\top}\bm\vep_t\Vert = O_p(\log(T))$,
and using \eqref{eq:bounded:w} yields
\begin{align}
\max_{1 \le i \le n} \max_{1 \le t \le T} \vert \bm\varphi_i^{\top}\mbf W_\chi\mbf W_\chi^\top\bm\vep_t\vert
\le 
\max_{1 \le i \le n} \Vert \bm\varphi_i^{\top}\mbf W_\chi\Vert \,
\max_{1 \le t \le T}\Vert{\mbf W}_{\chi}^{\top}\bm\vep_t\Vert = O_p\l(\frac{\log(T)}{\sqrt n}\r).
\label{eq:chi5}
\end{align}
Substituting \eqref{eq:chi3} and \eqref{eq:chi5} into \eqref{eq:chi1} completes the proof.
\hfill $\Box$

\subsection{Proof of Proposition~\ref{thm:two}}
\label{sec:pf:capping}

Note that
\begin{align*}
\max_{1 \le i \le n} \; \max_{1 \le t \le T} \vert \wh\chi_{it}^{\capp} - \chi_{it} \vert
\le &
\max_{1 \le i \le n} \; \max_{1 \le t \le T} 
\Big\vert \sum_{j=1}^r \wh w^{\capp}_{x, ij} (\wh{\mbf w}^{\capp}_{x, j})^\top 
\mbf x_t - \chi_{it} \Big\vert
\\
& + \max_{1 \le i \le n} \; \max_{1 \le t \le T} \Big\vert  
\sum_{j=r+1}^{\wh r} \wh w^{\capp}_{x, ij} (\wh{\mbf w}^{\capp}_{x, j})^\top \mbf x_t\Big\vert
=: I + II.
\end{align*}
Since capping does not alter the $r$ leading eigenvectors with probability converging to one
thanks to Equation \eqref{eq:bounded:ew},
we derive that
$I = O_p\{(\sqrt{\log(n)/T} \vee 1/\sqrt n)\log(T)\}$ as in the proof of Proposition \ref{thm:common} 
(see Appendix~\ref{sec:pf:prop:one}).
Next, 
\begin{align*}
II \le \max_{1 \le t \le T} 
\sum_{j = r+1}^{\wh r} \max_{1 \le i \le n} |\wh w^{\capp}_{x, ij}| \big\Vert (\wh{\mbf w}^{\capp}_{x, j})^\top \mbf x_t \big\Vert
=  O_p\l(\frac{1}{\sqrt n} \cdot \sqrt{n}\log(T)\r) = O_p(\log(T))
\end{align*}
from the uniform boundedness of $|\wh w^{\capp}_{x, ij}|$ and Lemma \ref{lem:four}.
\hfill $\Box$

\clearpage

\section{Additional simulation results}
\label{sec:sim:add}

\subsection{Estimation of $r$}
\label{sec:est:r}

Figures \ref{fig:sim:r:s2}--\ref{fig:sim:r:s9} plot the estimated number of factors
returned by \eqref{eq:r:bn} (`BN') and \eqref{eq:r:ah} (`AH')
when applied to $1000$ realisations generated under Model 2 
with varying $\varrho$, $T$, $n$ and $\phi$.

Under Model 2, the presence of $5$ additional weak factors 
leads the information criterion in \eqref{eq:r:bn} to 
frequently return $\wh r = 10$ or even larger.
The eigenvalue ratio estimator \eqref{eq:r:ah} 
also exhibits such tendency for small $n$ and large $T$,
and as the support of the leading eigenvectors of $\bm\Gamma_v$
under Model~2 becomes denser (increasing $\varrho$).

\begin{figure}[htbp]
\centering
\includegraphics[width=1\textwidth]{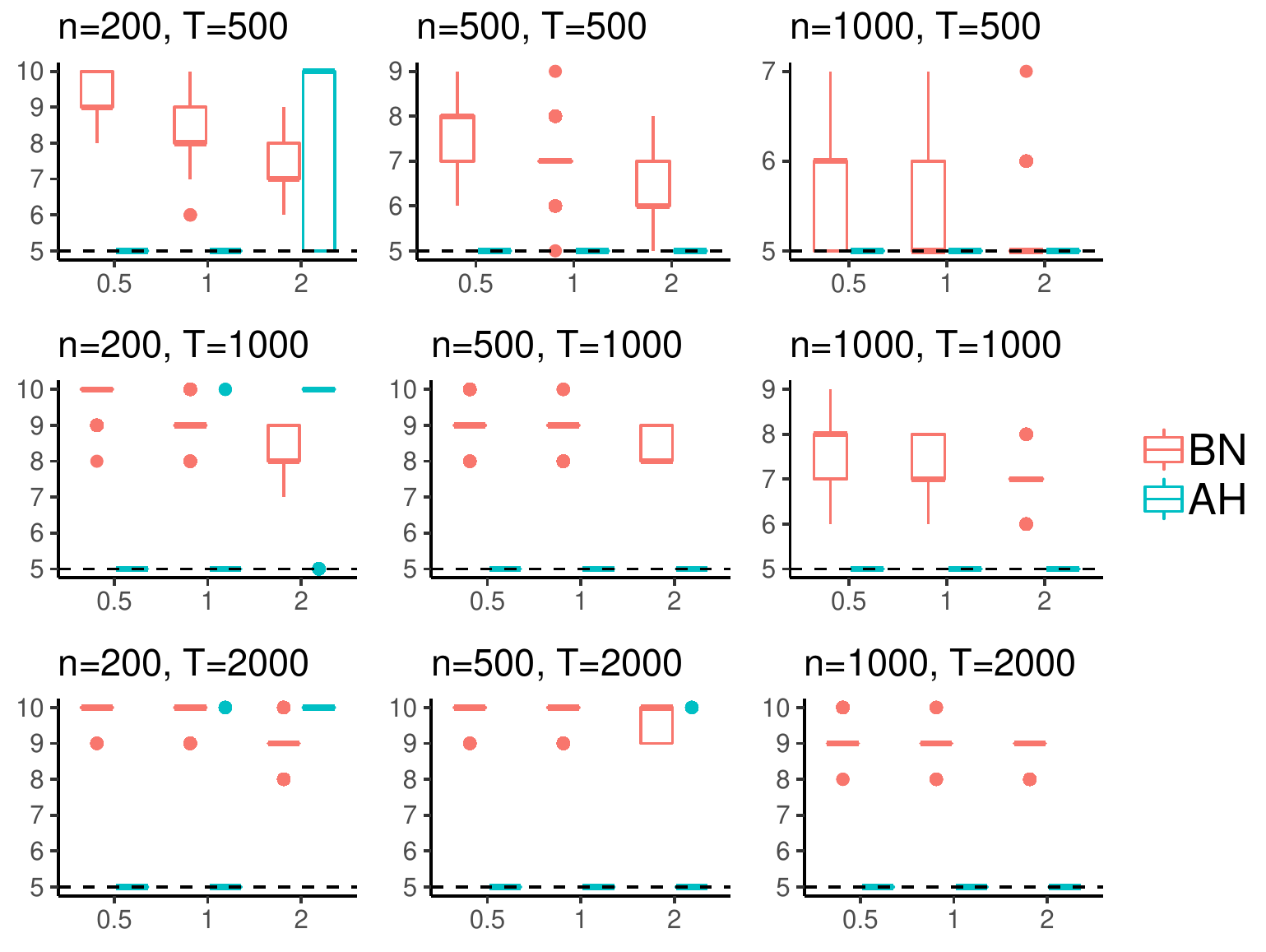}
\caption{Box plots of $\wh r$ returned by factor number estimators proposed by
\cite{baing02} (BN) and \cite{ahn2013} (AH) over $1000$ realisations 
generated under Model 2 with $\varrho = 0.2$,
$T \in \{500, 1000, 2000\}$ (top to bottom),
$n \in \{200, 500, 1000\}$ (left to right)
and $\phi \in \{0.5, 1, 2\}$ (left to right within each plot, controls the noise-to-signal ratio);
horizontal broken lines indicate the true factor number $r = 5$.}
\label{fig:sim:r:s2}
\end{figure}

\begin{figure}[htbp]
\centering
\includegraphics[width=1\textwidth]{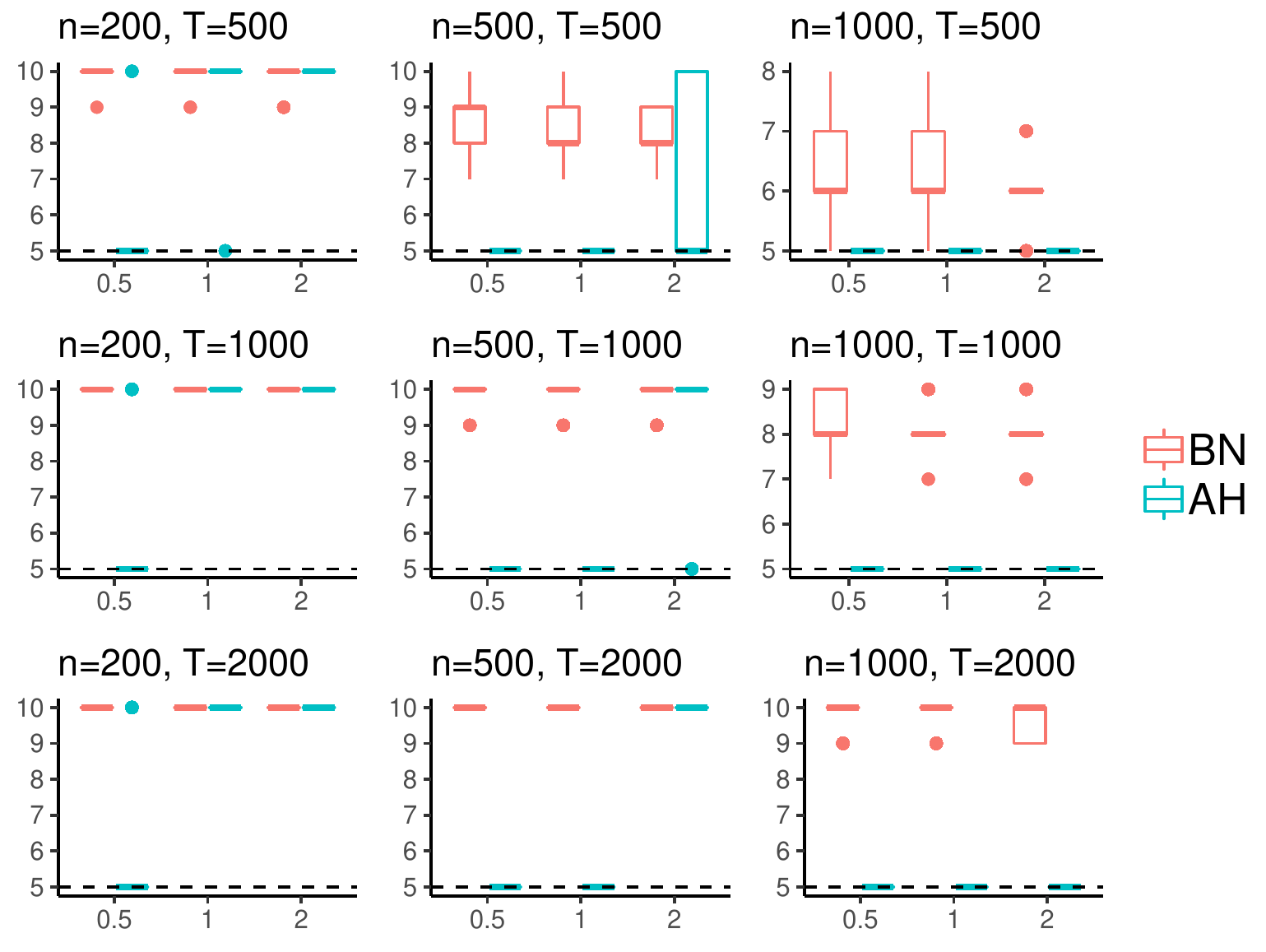}
\caption{Box plots of $\wh r$ returned by factor number estimators proposed by
\cite{baing02} (BN) and \cite{ahn2013} (AH) over $1000$ realisations 
generated under Model 2 with $\varrho = 0.5$,
$T \in \{500, 1000, 2000\}$ (top to bottom),
$n \in \{200, 500, 1000\}$ (left to right)
and $\phi \in \{0.5, 1, 2\}$ (left to right within each plot, controls the noise-to-signal ratio);
horizontal broken lines indicate the true factor number $r = 5$.}
\label{fig:sim:r:s5}
\end{figure}

\begin{figure}[htbp]
\centering
\includegraphics[width=1\textwidth]{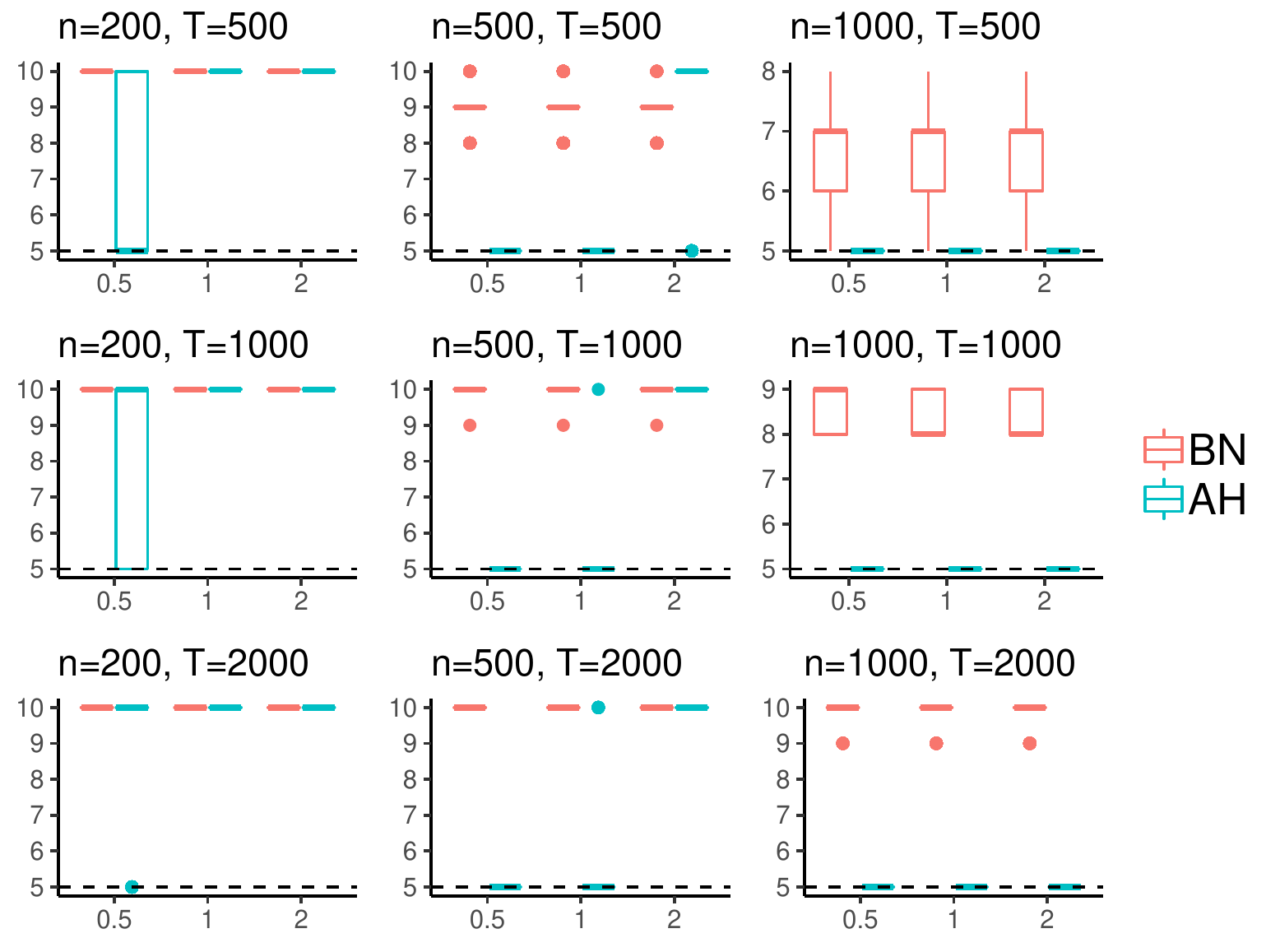}
\caption{Box plots of $\wh r$ returned by factor number estimators proposed by
\cite{baing02} (BN) and \cite{ahn2013} (AH) over $1000$ realisations 
generated under Model 2 with $\varrho = 0.9$,
$T \in \{500, 1000, 2000\}$ (top to bottom),
$n \in \{200, 500, 1000\}$ (left to right)
and $\phi \in \{0.5, 1, 2\}$ (left to right within each plot, controls the noise-to-signal ratio);
horizontal broken lines indicate the true factor number $r = 5$.}
\label{fig:sim:r:s9}
\end{figure}

\clearpage

\subsection{Empirical eigenvectors after scaling}

Figures \ref{fig:sim:cs:s2}--\ref{fig:sim:cs:s9} plot the behaviour of 
empirical eigenvectors after scaling,
for $\wh{\mbf w}_{x, j}, \, 2 \le j \le r$ and $\wh{\mbf w}_{x, j}, \, r + 1 \le j \le \wh r$ separately,
over $1000$ realisations generated under Model~2 with varying $\varrho$, $T$, $n$ and $\phi$.

\begin{figure}[htbp]
\includegraphics[width=1\textwidth]{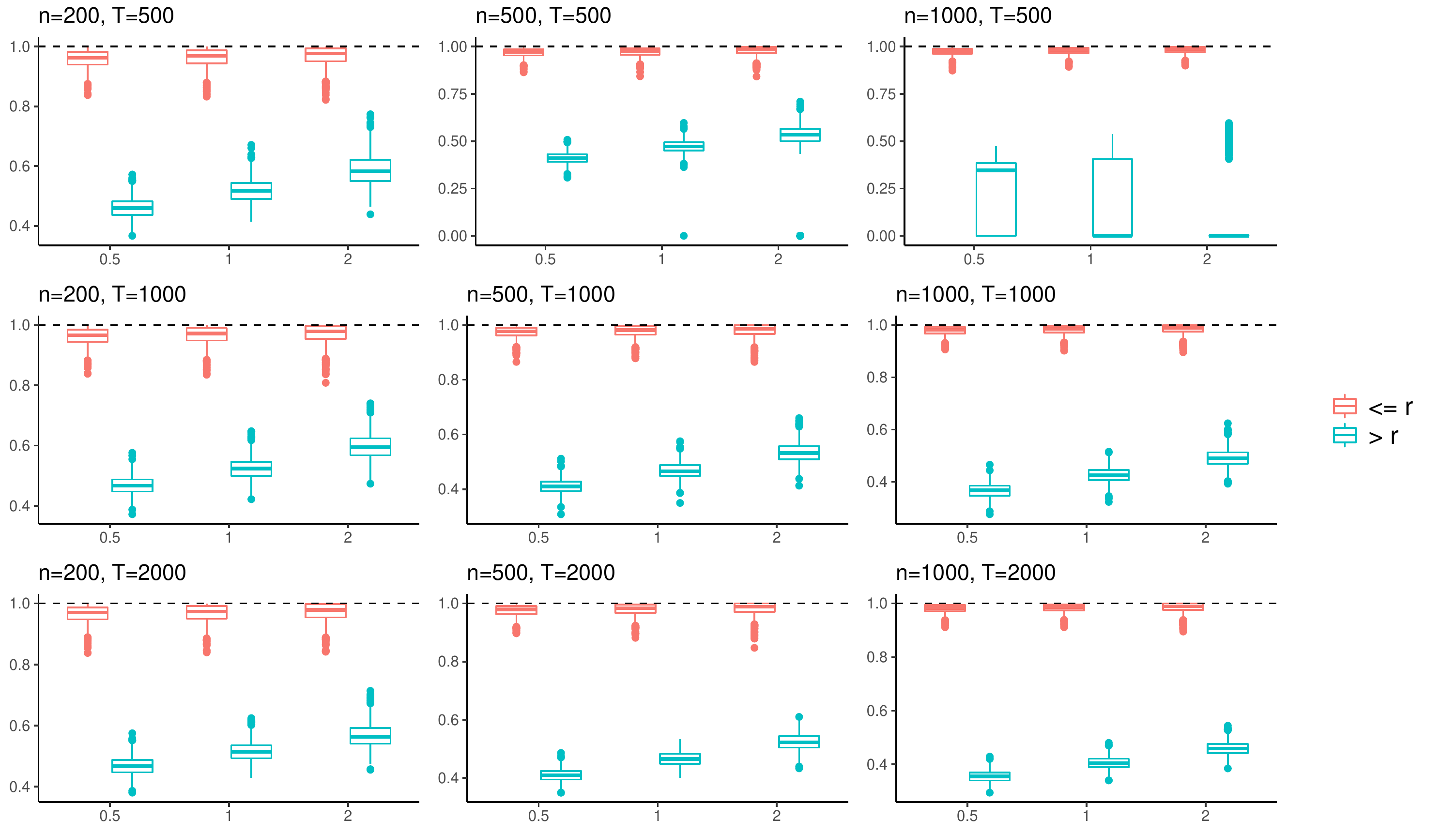}
\caption{Box plots of
$\Vert\wh{\mbf w}^{\sca}_{x, j}\Vert$ averaged for $2 \le j \le r$ (`$\le r$')
against that averaged for $r + 1 \le j \le \wh r$ (`$> r$')
averaged $1000$ realisations 
generated under Model 2 with $\varrho = 0.2$,
with $n \in \{200, 500, 1000\}$ (left to right), $T \in \{500, 1000, 2000\}$ (top to bottom)
and $\phi \in \{0.5, 1, 2\}$  (left to right within each plot).}
\label{fig:sim:cs:s2}
\end{figure}

\begin{figure}[htbp]
\includegraphics[width=1\textwidth]{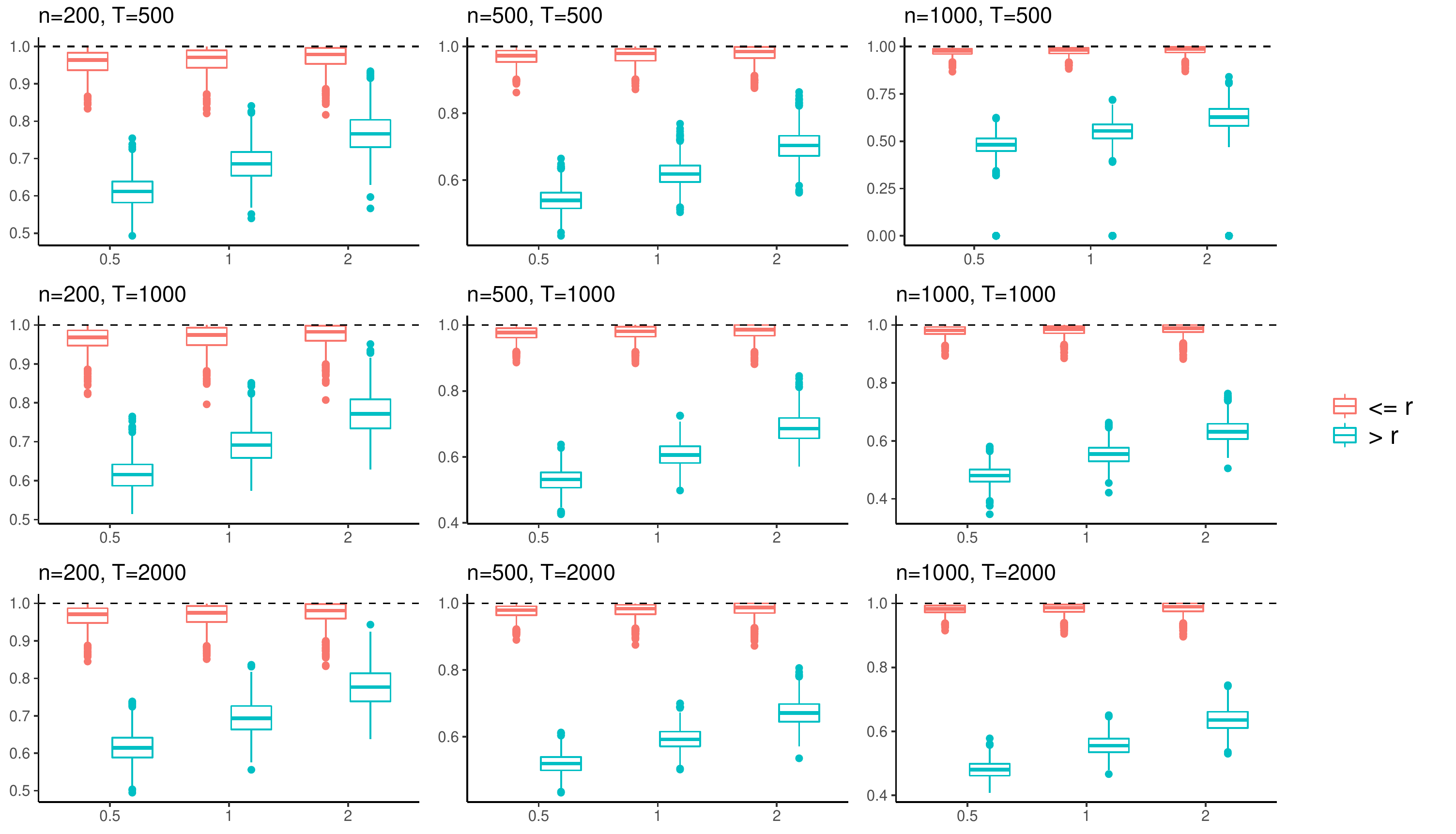}
\caption{Box plots of
$\Vert\wh{\mbf w}^{\sca}_{x, j}\Vert$ averaged for $2 \le j \le r$ (`$\le r$')
against that averaged for $r + 1 \le j \le \wh r$ (`$> r$')
averaged $1000$ realisations 
generated under Model 2 with $\varrho = 0.5$,
with $n \in \{200, 500, 1000\}$ (left to right), $T \in \{500, 1000, 2000\}$ (top to bottom)
and $\phi \in \{0.5, 1, 2\}$  (left to right within each plot).}
\label{fig:sim:cs:s5}
\end{figure}

\begin{figure}[htbp]
\includegraphics[width=1\textwidth]{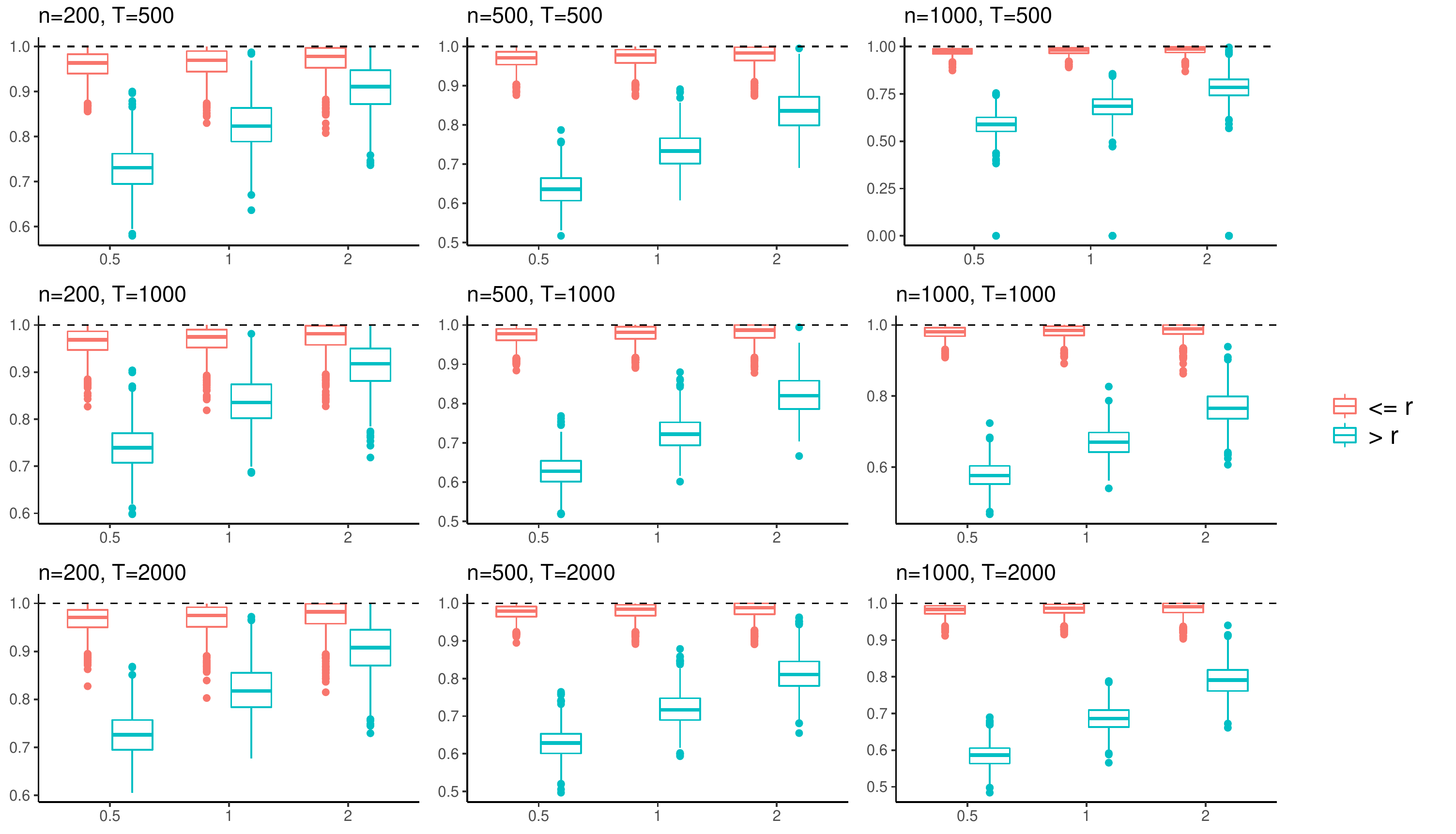}
\caption{Box plots of
$\Vert\wh{\mbf w}^{\sca}_{x, j}\Vert$ averaged for $2 \le j \le r$ (`$\le r$')
against that averaged for $r + 1 \le j \le \wh r$ (`$> r$')
averaged $1000$ realisations 
generated under Model 2 with $\varrho = 0.9$,
with $n \in \{200, 500, 1000\}$ (left to right), $T \in \{500, 1000, 2000\}$ (top to bottom)
and $\phi \in \{0.5, 1, 2\}$  (left to right within each plot).}
\label{fig:sim:cs:s9}
\end{figure}

\clearpage

\subsection{Estimation of $\chi_{it}$}

Figures \ref{fig:sim:err:T1000}--\ref{fig:sim:err:T2000s9}
plot the relative error measures $\text{err}_{\avg}$ and $\text{err}_{\max}$
(to that of the oracle PCA estimator obtained with the true $r$)
evaluated at various PC-based estimators of $\chi_{it}$
when applied to $1000$ realisations generated under Models 1--2
with varying $\varrho$, $T$, $n$ and $\phi$.
Tables \ref{table:sim:err:500}--\ref{table:sim:err:2000}
further report the summary of the errors.

\begin{figure}[htb]
\centering
\includegraphics[width=1\textwidth]{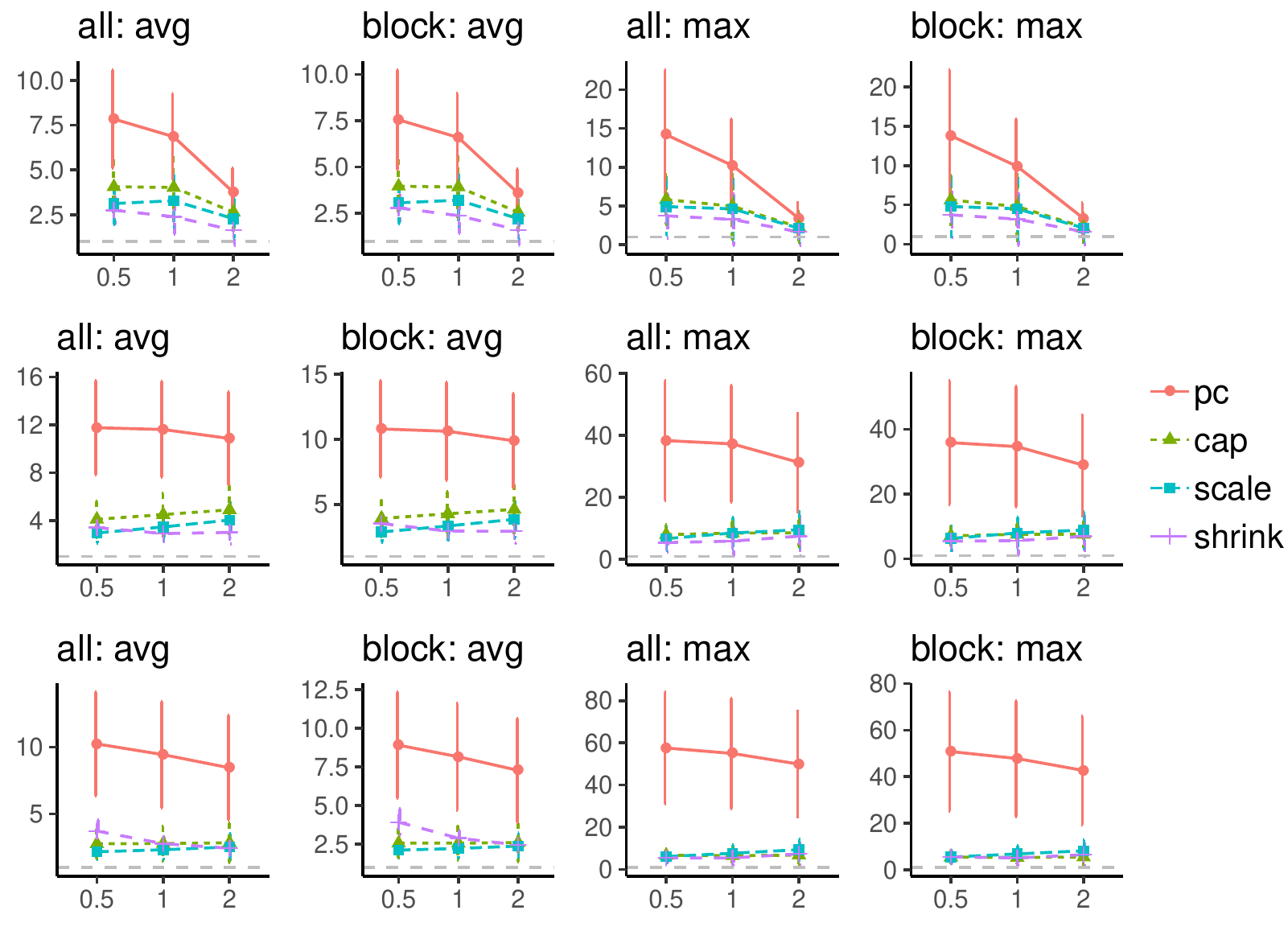}
\caption{$\text{err}_{\avg}(\wh\chi^{\circ}_{it})$
and $\text{err}_{\max}(\wh\chi^{\circ}_{it})$
for $\wh{\chi}^{\pca}_{it}$, $\wh{\chi}^{\capp}_{it}$, $\wh{\chi}^{\sca}_{it}$
and $\wh{\chi}^{\cs}_{it}$ estimated using the entire sample (`all'),
and their blockwise counterparts 
$\wh{\chi}^{\bpca}_{it}$, $\wh{\chi}^{\bcapp}_{it}$, $\wh{\chi}^{\bsca}_{it}$
and $\wh{\chi}^{\bcs}_{it}$ (`block'),
averaged over $1000$ realisations generated under Model 1 with $T = 1000$,
$n \in \{200, 500, 1000\}$ (top to bottom)
and $\phi \in \{0.5, 1, 2\}$ (left to right within each plot).
The vertical errors bars represent the standard deviations.}
\label{fig:sim:err:T1000}
\end{figure}

\begin{figure}[htb]
\centering
\includegraphics[width=1\textwidth]{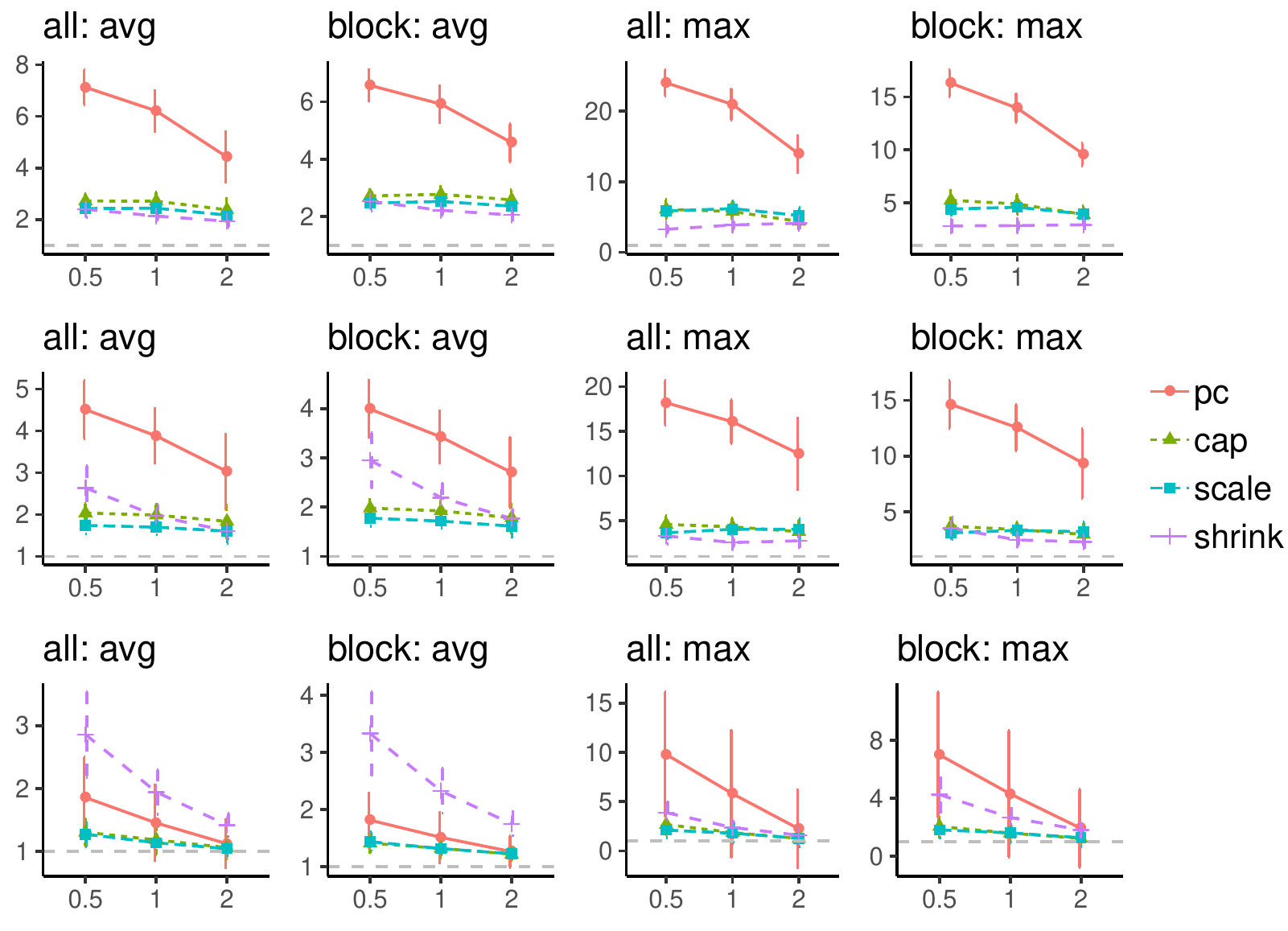}
\caption{$\text{err}_{\avg}(\wh\chi^{\circ}_{it})$
and $\text{err}_{\max}(\wh\chi^{\circ}_{it})$
averaged over $1000$ realisations generated under Model 2 with $\varrho = 0.2$, $T = 500$,
$n \in \{200, 500, 1000\}$ (top to bottom)
and $\phi \in \{0.5, 1, 2\}$ (left to right within each plot).}
\label{fig:sim:err:T500s2}
\end{figure}

\begin{figure}[htb]
\centering
\includegraphics[width=1\textwidth]{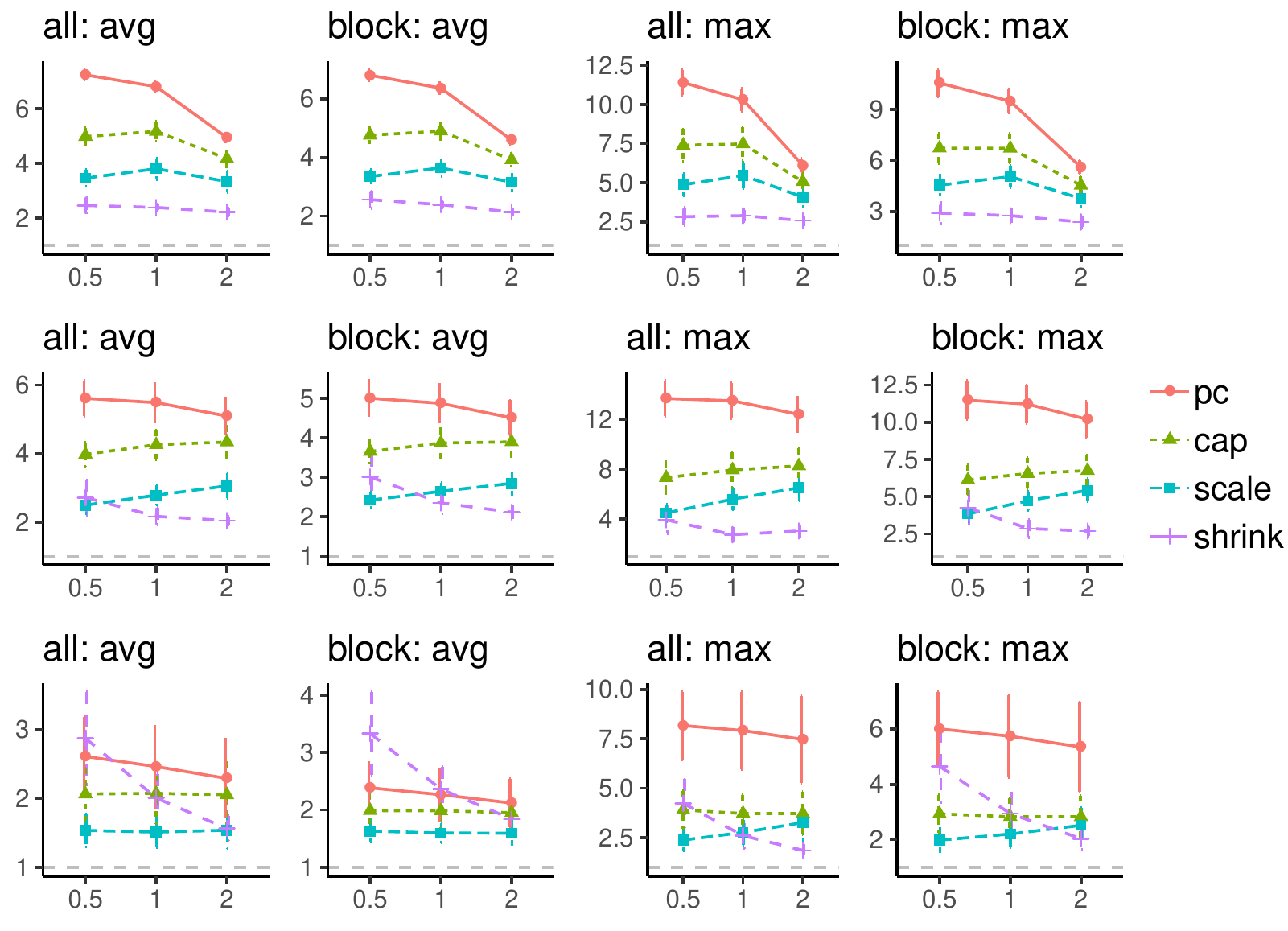}
\caption{$\text{err}_{\avg}(\wh\chi^{\circ}_{it})$
and $\text{err}_{\max}(\wh\chi^{\circ}_{it})$
averaged over $1000$ realisations generated under Model 2 with $\varrho = 0.5$, $T = 500$,
$n \in \{200, 500, 1000\}$ (top to bottom)
and $\phi \in \{0.5, 1, 2\}$ (left to right within each plot).}
\label{fig:sim:err:T500s5}
\end{figure}

\begin{figure}[htb]
\centering
\includegraphics[width=1\textwidth]{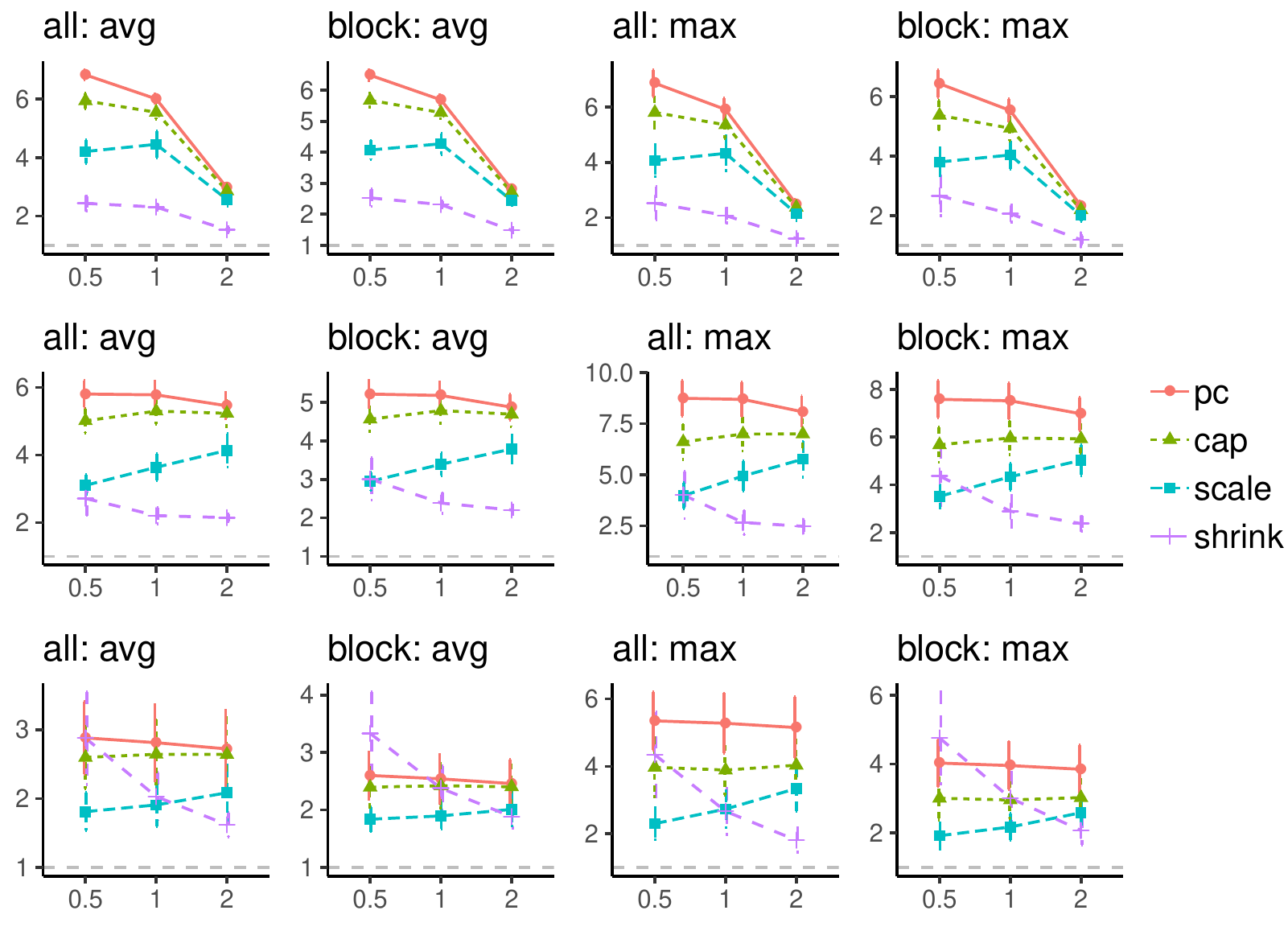}
\caption{$\text{err}_{\avg}(\wh\chi^{\circ}_{it})$
and $\text{err}_{\max}(\wh\chi^{\circ}_{it})$
averaged over $1000$ realisations generated under Model 2 with $\varrho = 0.9$, $T = 500$,
$n \in \{200, 500, 1000\}$ (top to bottom)
and $\phi \in \{0.5, 1, 2\}$ (left to right within each plot).}
\label{fig:sim:err:T500s9}
\end{figure}

\begin{figure}[htb]
\centering
\includegraphics[width=1\textwidth]{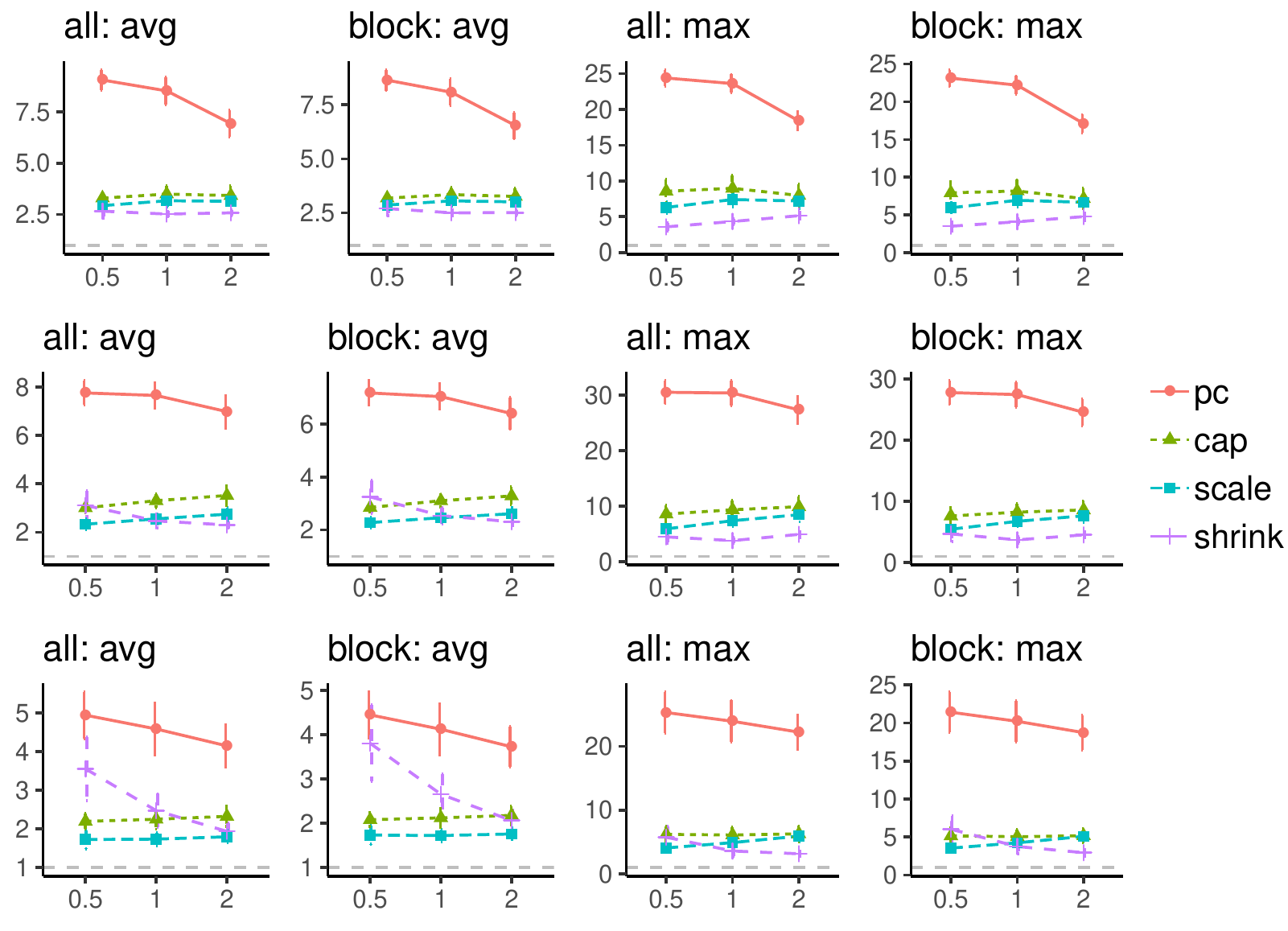}
\caption{$\text{err}_{\avg}(\wh\chi^{\circ}_{it})$
and $\text{err}_{\max}(\wh\chi^{\circ}_{it})$
averaged over $1000$ realisations generated under Model 2 with $\varrho = 0.2$, $T = 1000$,
$n \in \{200, 500, 1000\}$ (top to bottom)
and $\phi \in \{0.5, 1, 2\}$ (left to right within each plot).}
\label{fig:sim:err:T1000s2}
\end{figure}

\begin{figure}[htb]
\centering
\includegraphics[width=1\textwidth]{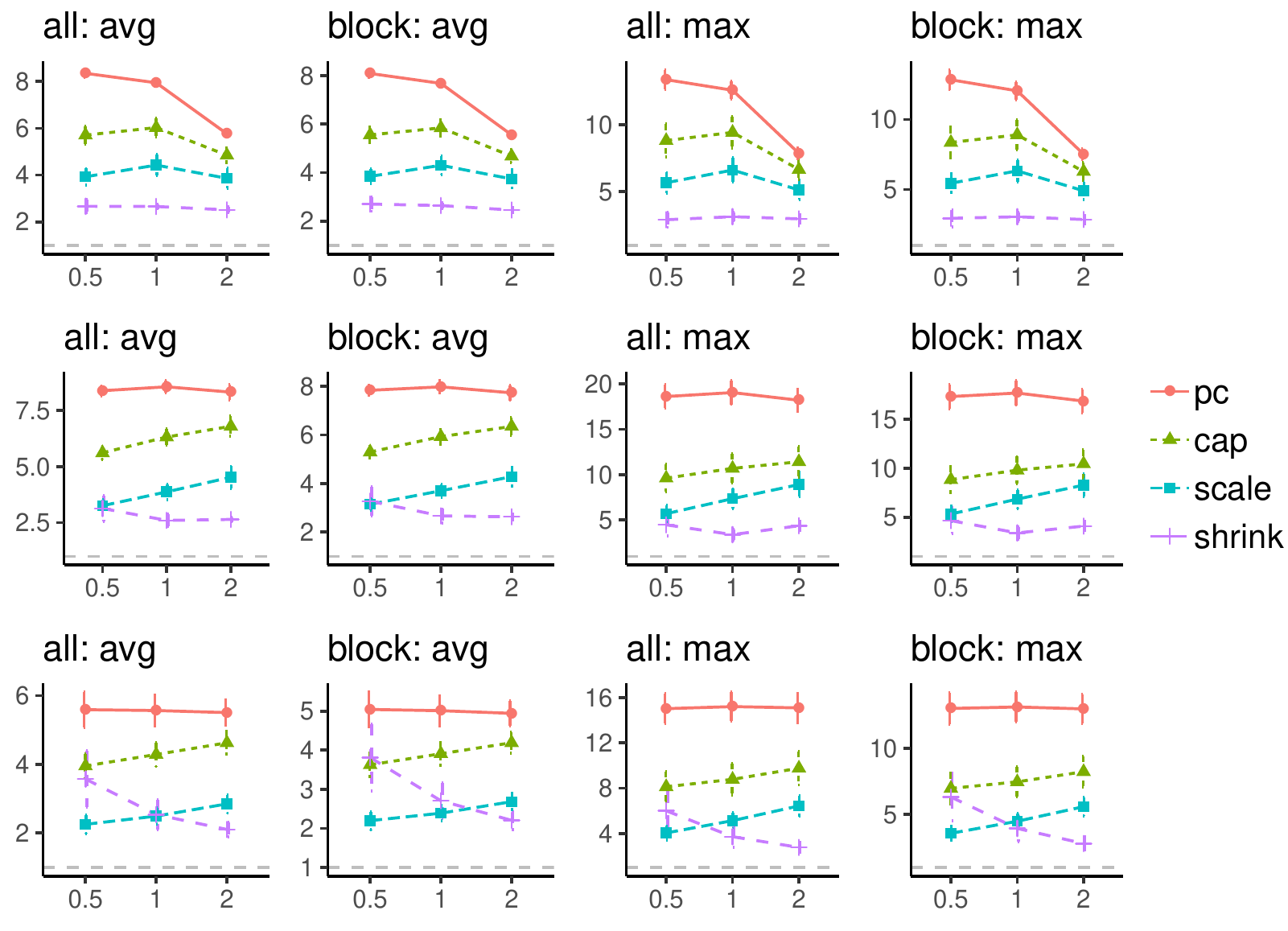}
\caption{$\text{err}_{\avg}(\wh\chi^{\circ}_{it})$
and $\text{err}_{\max}(\wh\chi^{\circ}_{it})$
averaged over $1000$ realisations generated under Model 2 with $\varrho = 0.5$, $T = 1000$,
$n \in \{200, 500, 1000\}$ (top to bottom)
and $\phi \in \{0.5, 1, 2\}$ (left to right within each plot).}
\label{fig:sim:err:T1000s5}
\end{figure}

\begin{figure}[htb]
\centering
\includegraphics[width=1\textwidth]{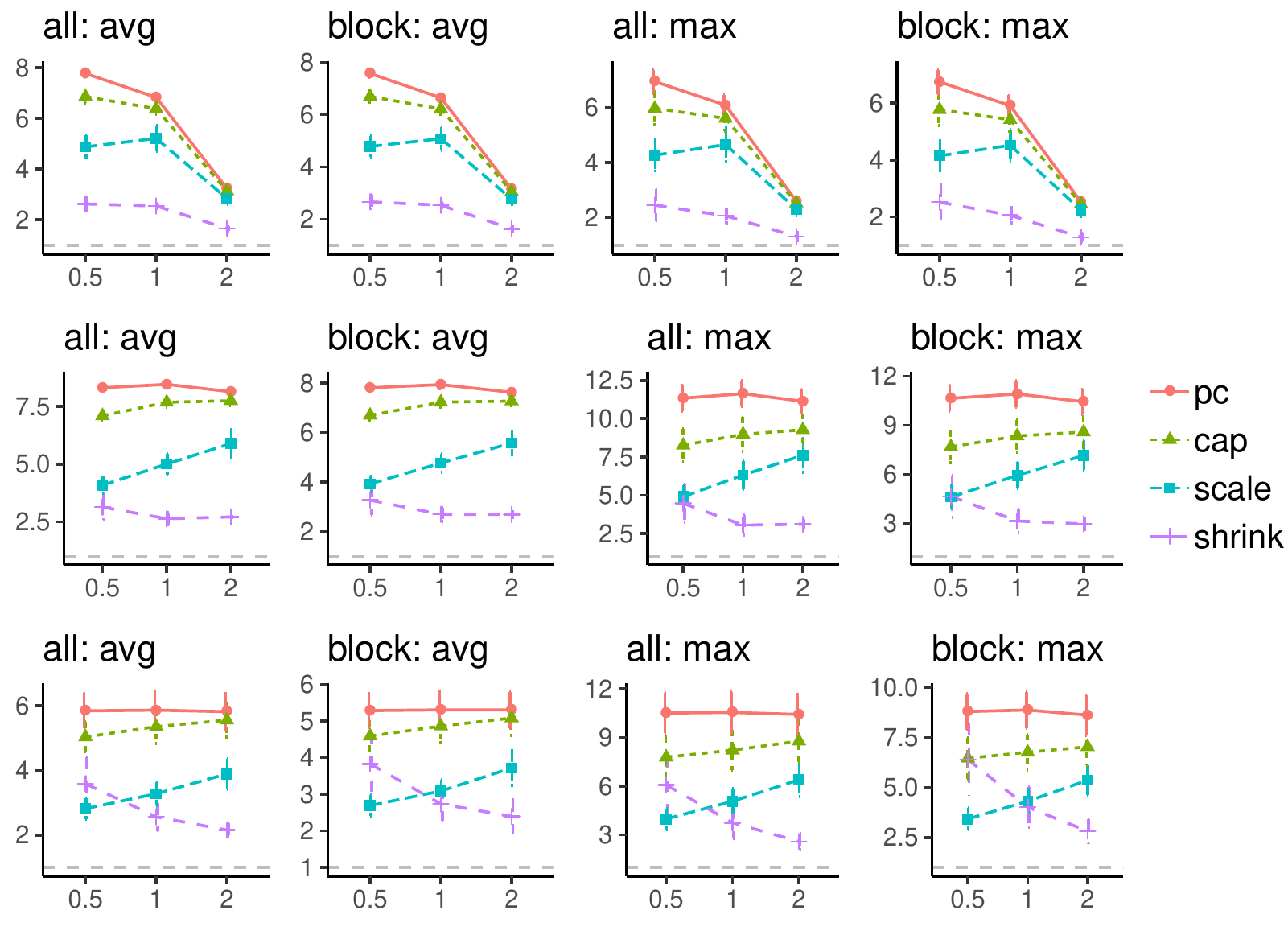}
\caption{$\text{err}_{\avg}(\wh\chi^{\circ}_{it})$
and $\text{err}_{\max}(\wh\chi^{\circ}_{it})$
averaged over $1000$ realisations generated under Model 2 with $\varrho = 0.9$, $T = 1000$,
$n \in \{200, 500, 1000\}$ (top to bottom)
and $\phi \in \{0.5, 1, 2\}$ (left to right within each plot).}
\label{fig:sim:err:T1000s9}
\end{figure}

\begin{figure}[htb]
\centering
\includegraphics[width=1\textwidth]{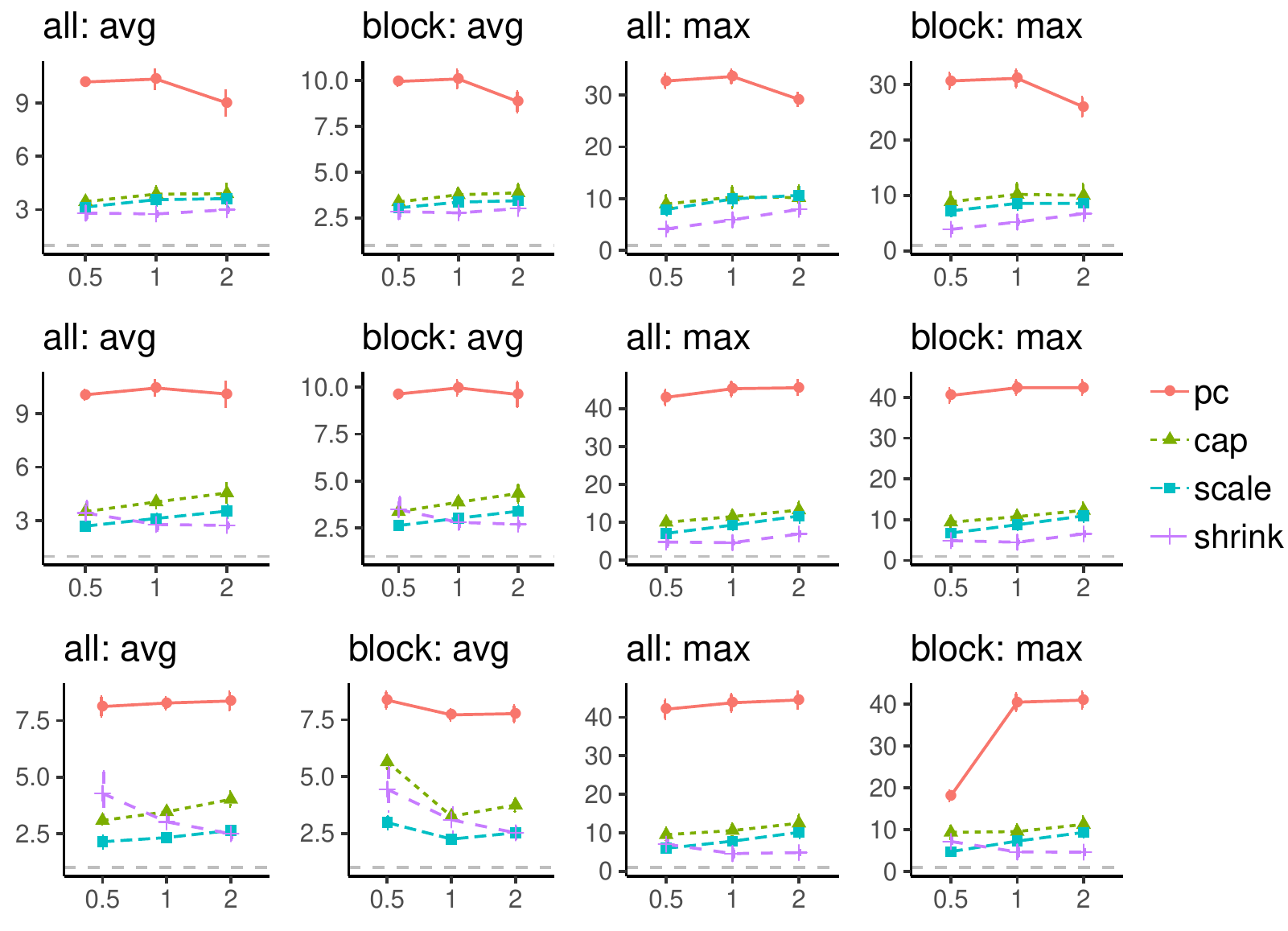}
\caption{$\text{err}_{\avg}(\wh\chi^{\circ}_{it})$
and $\text{err}_{\max}(\wh\chi^{\circ}_{it})$
on $1000$ realisations generated under Model 2 with $\varrho = 0.2$, $T = 2000$,
$n \in \{200, 500, 1000\}$ (top to bottom)
and $\phi \in \{0.5, 1, 2\}$ (left to right within each plot).}
\label{fig:sim:err:T2000s2}
\end{figure}

\begin{figure}[htb]
\centering
\includegraphics[width=1\textwidth]{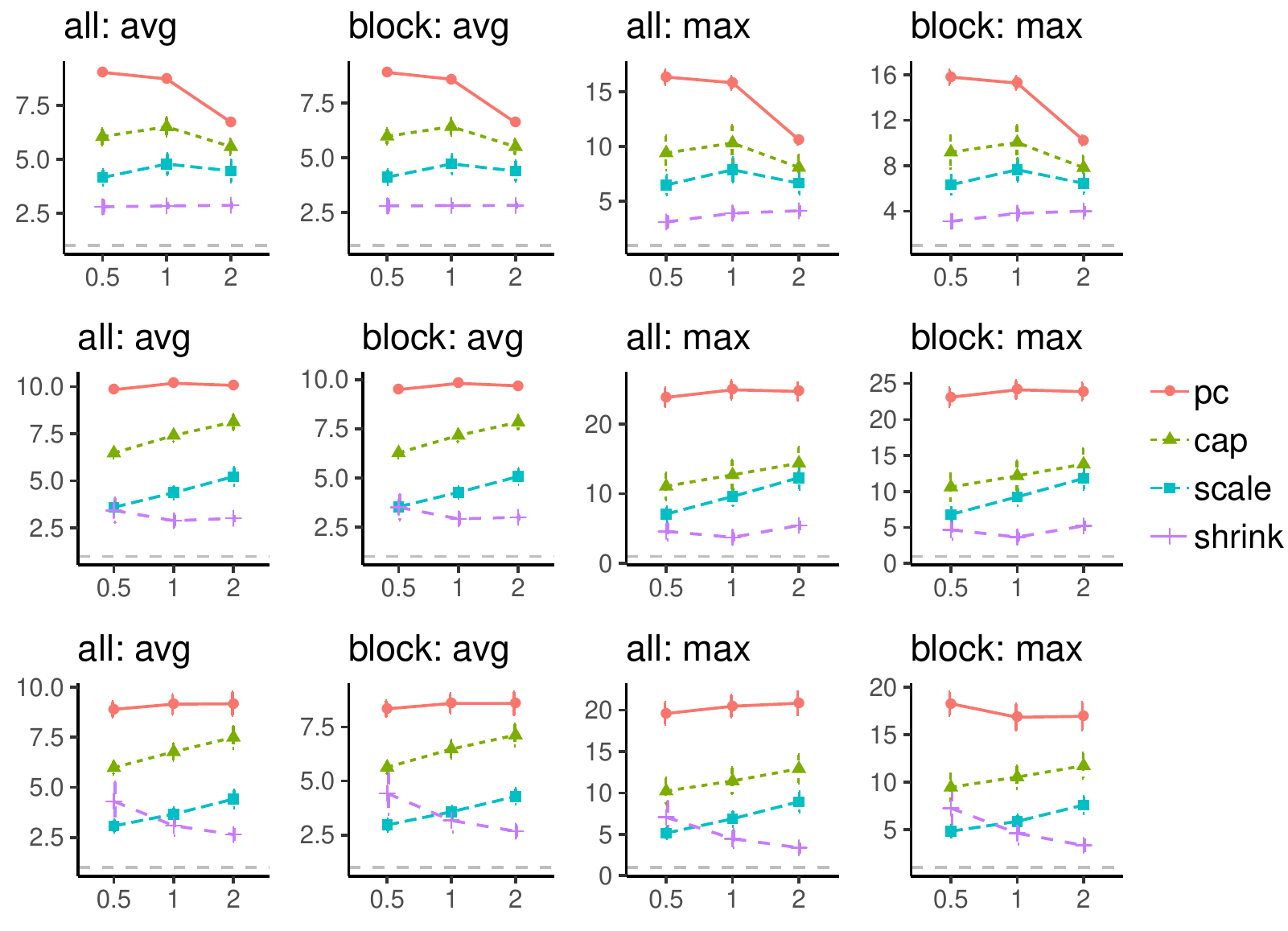}
\caption{$\text{err}_{\avg}(\wh\chi^{\circ}_{it})$
and $\text{err}_{\max}(\wh\chi^{\circ}_{it})$
averaged over $1000$ realisations generated under Model 2 with $\varrho = 0.5$, $T = 2000$,
$n \in \{200, 500, 1000\}$ (top to bottom)
and $\phi \in \{0.5, 1, 2\}$ (left to right within each plot).}
\label{fig:sim:err:T2000s5}
\end{figure}

\begin{figure}[htb]
\centering
\includegraphics[width=1\textwidth]{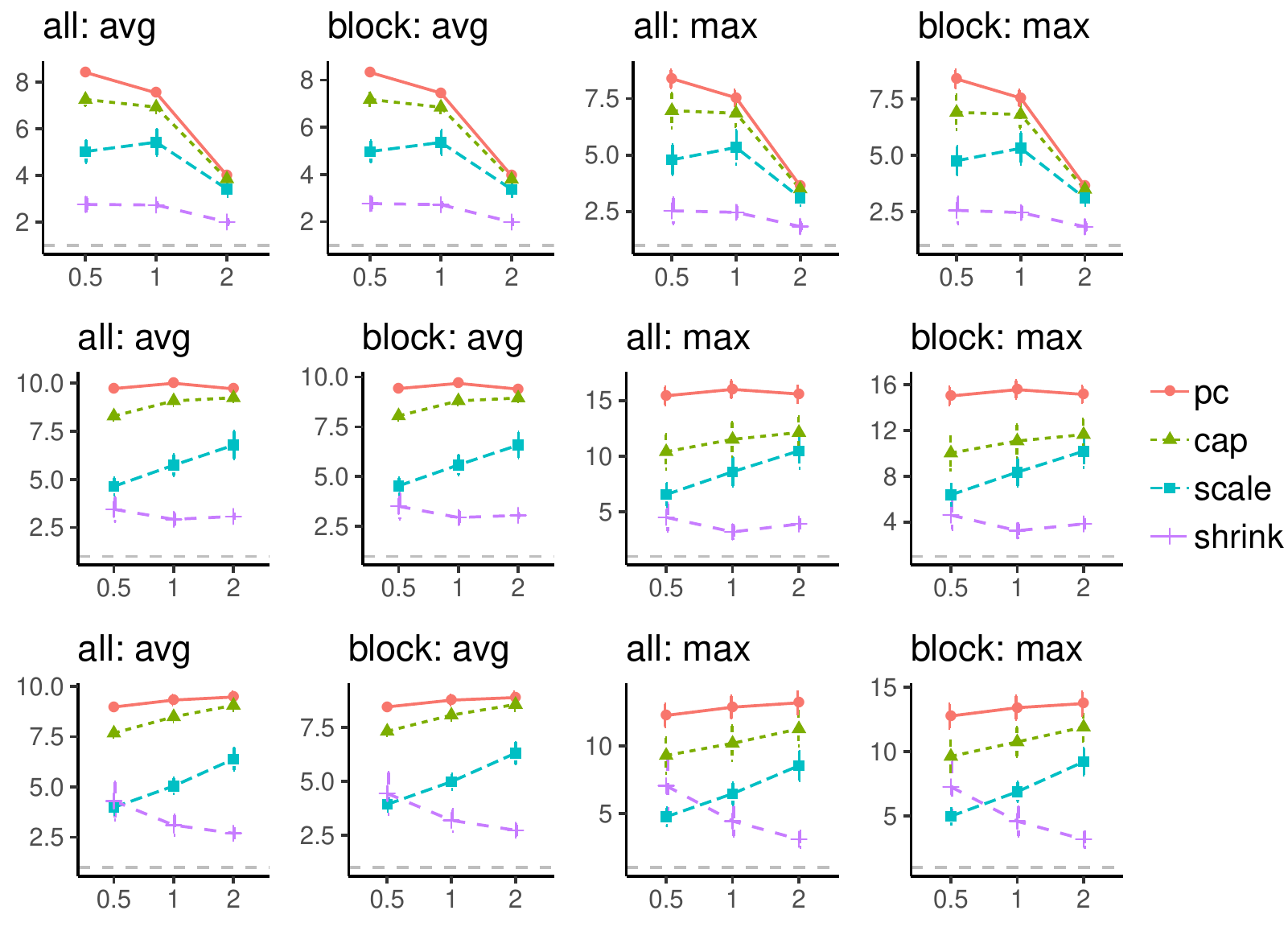}
\caption{$\text{err}_{\avg}(\wh\chi^{\circ}_{it})$
and $\text{err}_{\max}(\wh\chi^{\circ}_{it})$
averaged over $1000$ realisations generated under Model 2 with $\varrho = 0.9$, $T = 2000$,
$n \in \{200, 500, 1000\}$ (top to bottom)
and $\phi \in \{0.5, 1, 2\}$ (left to right within each plot).}
\label{fig:sim:err:T2000s9}
\end{figure}

\clearpage

\begin{table}[hp]
\setlength{\tabcolsep}{3pt}
\centering
\caption{Summary of relative estimation error over $1000$ realisations 
from various estimators for 
$T = 500$,
$n \in \{200, 500, 1000\}$,
$\varrho \in \{0.2, 0.5, 0.9, 1\}$ (with $\varrho = 1$ representing Model 1)
and $\phi \in \{0.5, 1, 2\}$; 
\eqref{eq:r:bn} is used for the estimation of $r = 5$.}
\label{table:sim:err:500}
{\footnotesize
\begin{tabular}{ccc|cccc|cccc|cccc|cccc}
\hline\hline
&	&	&	\multicolumn{8}{c}{all} &								\multicolumn{8}{c}{block} 								\\	
&	&	&	\multicolumn{4}{c}{$\text{err}_{\avg}$} &				\multicolumn{4}{c}{$\text{err}_{\max}$} &				\multicolumn{4}{c}{$\text{err}_{\avg}$} &				\multicolumn{4}{c}{$\text{err}_{\max}$}				\\	
$n$ &	$s$ &	$\phi$ &	$\wh\chi^{\pca}_{it}$ &	$\wh\chi^{\capp}_{it}$ &	$\wh\chi^{\sca}_{it}$ &	$\wh\chi^{\cs}_{it}$ &	$\wh\chi^{\pca}_{it}$ &	$\wh\chi^{\capp}_{it}$ &	$\wh\chi^{\sca}_{it}$ &	$\wh\chi^{\cs}_{it}$ &	$\wh\chi^{\pca}_{it}$ &	$\wh\chi^{\capp}_{it}$ &	$\wh\chi^{\sca}_{it}$ &	$\wh\chi^{\cs}_{it}$ &	$\wh\chi^{\pca}_{it}$ &	$\wh\chi^{\capp}_{it}$ &	$\wh\chi^{\sca}_{it}$ &	$\wh\chi^{\cs}_{it}$ 	\\	\hline
200 &	0.2 &	0.5 &	7.13 &	2.72 &	2.43 &	2.4 &	24.02 &	6.1 &	5.87 &	3.26 &	6.57 &	2.72 &	2.47 &	2.53 &	16.3 &	5.25 &	4.42 &	2.83	\\	
&	&	1 &	6.21 &	2.71 &	2.44 &	2.13 &	20.97 &	5.78 &	6.2 &	3.88 &	5.92 &	2.78 &	2.53 &	2.22 &	13.94 &	4.9 &	4.59 &	2.86	\\	
&	&	2 &	4.42 &	2.38 &	2.17 &	1.93 &	13.93 &	4.39 &	5.23 &	4.13 &	4.58 &	2.58 &	2.36 &	2.05 &	9.59 &	3.91 &	3.95 &	2.93	\\	
&	0.5 &	0.5 &	7.24 &	4.98 &	3.46 &	2.46 &	11.4 &	7.4 &	4.87 &	2.83 &	6.8 &	4.76 &	3.35 &	2.56 &	10.58 &	6.73 &	4.54 &	2.90	\\	
&	&	1 &	6.81 &	5.17 &	3.81 &	2.38 &	10.3 &	7.48 &	5.47 &	2.89 &	6.36 &	4.9 &	3.64 &	2.39 &	9.5 &	6.72 &	5.05 &	2.75	\\	
&	&	2 &	4.95 &	4.17 &	3.33 &	2.22 &	6.11 &	5.06 &	4.08 &	2.58 &	4.59 &	3.91 &	3.15 &	2.13 &	5.61 &	4.54 &	3.76 &	2.37	\\	
&	0.9 &	0.5 &	6.83 &	5.94 &	4.21 &	2.44 &	6.86 &	5.8 &	4.06 &	2.52 &	6.49 &	5.68 &	4.07 &	2.53 &	6.44 &	5.37 &	3.8 &	2.66	\\	
&	&	1 &	6.01 &	5.55 &	4.45 &	2.31 &	5.91 &	5.35 &	4.32 &	2.08 &	5.69 &	5.28 &	4.27 &	2.31 &	5.54 &	4.94 &	4.04 &	2.07	\\	
&	&	2 &	2.97 &	2.86 &	2.57 &	1.53 &	2.48 &	2.38 &	2.15 &	1.25 &	2.8 &	2.71 &	2.45 &	1.49 &	2.32 &	2.2 &	2.01 &	1.18	\\	
&	1 &	0.5 &	6.29 &	3.51 &	2.72 &	2.51 &	12.58 &	5.36 &	4.42 &	3.41 &	5.86 &	3.42 &	2.68 &	2.64 &	11.76 &	5.06 &	4.23 &	3.47	\\	
&	&	1 &	5.5 &	3.44 &	2.82 &	2.11 &	9.38 &	4.62 &	4.25 &	2.93 &	5.11 &	3.33 &	2.75 &	2.15 &	8.73 &	4.35 &	4.05 &	2.82	\\	
&	&	2 &	3.4 &	2.52 &	2.18 &	1.59 &	3.77 &	2.47 &	2.34 &	1.76 &	3.16 &	2.42 &	2.11 &	1.55 &	3.49 &	2.33 &	2.23 &	1.65	\\	\hline
500 &	0.2 &	0.5 &	4.51 &	2.04 &	1.74 &	2.64 &	18.19 &	4.57 &	3.64 &	3.28 &	3.99 &	1.98 &	1.77 &	2.95 &	14.63 &	3.7 &	3.1 &	3.53	\\	
&	&	1 &	3.88 &	1.99 &	1.7 &	1.97 &	16.06 &	4.31 &	4.03 &	2.55 &	3.42 &	1.92 &	1.72 &	2.19 &	12.54 &	3.4 &	3.33 &	2.47	\\	
&	&	2 &	3.03 &	1.84 &	1.6 &	1.61 &	12.48 &	3.77 &	4.04 &	2.73 &	2.7 &	1.78 &	1.61 &	1.77 &	9.33 &	2.97 &	3.21 &	2.30	\\	
&	0.5 &	0.5 &	5.61 &	3.98 &	2.49 &	2.72 &	13.67 &	7.32 &	4.49 &	3.93 &	5.01 &	3.66 &	2.42 &	3.01 &	11.48 &	6.14 &	3.85 &	4.26	\\	
&	&	1 &	5.49 &	4.26 &	2.78 &	2.17 &	13.48 &	7.93 &	5.59 &	2.74 &	4.87 &	3.87 &	2.65 &	2.35 &	11.21 &	6.56 &	4.73 &	2.87	\\	
&	&	2 &	5.09 &	4.34 &	3.06 &	2.04 &	12.39 &	8.26 &	6.51 &	3.03 &	4.51 &	3.9 &	2.85 &	2.12 &	10.19 &	6.75 &	5.43 &	2.69	\\	
&	0.9 &	0.5 &	5.8 &	5.01 &	3.1 &	2.71 &	8.74 &	6.61 &	3.97 &	4.01 &	5.21 &	4.56 &	2.95 &	3.01 &	7.59 &	5.68 &	3.52 &	4.37	\\	
&	&	1 &	5.78 &	5.3 &	3.64 &	2.21 &	8.69 &	7 &	4.94 &	2.66 &	5.18 &	4.78 &	3.39 &	2.38 &	7.53 &	5.97 &	4.34 &	2.90	\\	
&	&	2 &	5.46 &	5.23 &	4.14 &	2.15 &	8.08 &	7 &	5.76 &	2.47 &	4.87 &	4.69 &	3.79 &	2.2 &	6.98 &	5.93 &	5.03 &	2.37	\\	
&	1 &	0.5 &	6.6 &	2.58 &	2.13 &	2.87 &	25.35 &	5.09 &	4.83 &	4.16 &	5.81 &	2.53 &	2.15 &	3.16 &	21.83 &	4.32 &	4.35 &	4.24	\\	
&	&	1 &	6.23 &	2.68 &	2.3 &	2.31 &	23.24 &	5.01 &	5.73 &	4.15 &	5.45 &	2.61 &	2.29 &	2.48 &	19.69 &	4.3 &	5.11 &	3.84	\\	
&	&	2 &	5.53 &	2.73 &	2.47 &	2.15 &	18.24 &	4.75 &	5.98 &	4.8 &	4.81 &	2.62 &	2.4 &	2.21 &	15.2 &	4.09 &	5.26 &	4.15	\\	\hline
1000 &	0.2 &	0.5 &	1.85 &	1.3 &	1.27 &	2.86 &	9.79 &	2.67 &	2.11 &	3.88 &	1.81 &	1.41 &	1.43 &	3.33 &	6.99 &	2.04 &	1.81 &	4.24	\\	
&	&	1 &	1.45 &	1.18 &	1.14 &	1.94 &	5.79 &	1.86 &	1.79 &	2.37 &	1.51 &	1.32 &	1.31 &	2.32 &	4.29 &	1.59 &	1.6 &	2.66	\\	
&	&	2 &	1.12 &	1.06 &	1.04 &	1.42 &	2.23 &	1.26 &	1.28 &	1.59 &	1.26 &	1.22 &	1.22 &	1.75 &	1.92 &	1.23 &	1.27 &	1.80	\\	
&	0.5 &	0.5 &	2.61 &	2.06 &	1.54 &	2.87 &	8.16 &	3.9 &	2.38 &	4.24 &	2.39 &	1.99 &	1.63 &	3.33 &	6 &	2.93 &	1.98 &	4.65	\\	
&	&	1 &	2.46 &	2.07 &	1.51 &	2.01 &	7.92 &	3.73 &	2.77 &	2.6 &	2.27 &	1.99 &	1.6 &	2.36 &	5.73 &	2.83 &	2.2 &	2.94	\\	
&	&	2 &	2.29 &	2.05 &	1.54 &	1.56 &	7.47 &	3.72 &	3.26 &	1.86 &	2.12 &	1.95 &	1.6 &	1.84 &	5.34 &	2.83 &	2.52 &	2.04	\\	
&	0.9 &	0.5 &	2.88 &	2.6 &	1.81 &	2.88 &	5.35 &	3.97 &	2.3 &	4.35 &	2.6 &	2.4 &	1.84 &	3.33 &	4.03 &	3.01 &	1.92 &	4.76	\\	
&	&	1 &	2.81 &	2.64 &	1.91 &	2.03 &	5.28 &	3.89 &	2.73 &	2.67 &	2.54 &	2.42 &	1.9 &	2.38 &	3.96 &	2.96 &	2.17 &	3.01	\\	
&	&	2 &	2.72 &	2.64 &	2.08 &	1.62 &	5.15 &	4.03 &	3.34 &	1.81 &	2.46 &	2.4 &	2.01 &	1.88 &	3.85 &	3.03 &	2.59 &	2.07	\\	
&	1 &	0.5 &	4.42 &	1.58 &	1.53 &	2.94 &	28.19 &	3 &	3.52 &	3.52 &	3.77 &	1.63 &	1.64 &	3.37 &	21.55 &	2.31 &	2.98 &	3.60	\\	
&	&	1 &	3.8 &	1.51 &	1.47 &	2.11 &	24.93 &	2.74 &	4.08 &	3.28 &	3.27 &	1.57 &	1.58 &	2.42 &	18.58 &	2.18 &	3.35 &	2.95	\\	
&	&	2 &	3.21 &	1.46 &	1.46 &	1.69 &	20.41 &	2.71 &	4.47 &	3.84 &	2.8 &	1.52 &	1.55 &	1.93 &	14.94 &	2.19 &	3.6 &	3.13	\\	\hline\hline
\end{tabular}}
\end{table}

\begin{table}[hp]
\setlength{\tabcolsep}{3pt}
\centering
\caption{Summary of relative estimation error from various estimators for 
$T = 1000$,
$n \in \{200, 500, 1000\}$,
$\varrho \in \{0.2, 0.5, 0.9, 1\}$
and $\phi \in \{0.5, 1, 2\}$.}
\label{table:sim:err:1000}
{\footnotesize
\begin{tabular}{ccc|cccc|cccc|cccc|cccc}
\hline\hline
&	&	&	\multicolumn{8}{c}{all} &								\multicolumn{8}{c}{block} 								\\	
&	&	&	\multicolumn{4}{c}{$\text{err}_{\avg}$} &				\multicolumn{4}{c}{$\text{err}_{\max}$} &				\multicolumn{4}{c}{$\text{err}_{\avg}$} &				\multicolumn{4}{c}{$\text{err}_{\max}$}				\\	
$n$ &	$s$ &	$\phi$ &	$\wh\chi^{\pca}_{it}$ &	$\wh\chi^{\capp}_{it}$ &	$\wh\chi^{\sca}_{it}$ &	$\wh\chi^{\cs}_{it}$ &	$\wh\chi^{\pca}_{it}$ &	$\wh\chi^{\capp}_{it}$ &	$\wh\chi^{\sca}_{it}$ &	$\wh\chi^{\cs}_{it}$ &	$\wh\chi^{\pca}_{it}$ &	$\wh\chi^{\capp}_{it}$ &	$\wh\chi^{\sca}_{it}$ &	$\wh\chi^{\cs}_{it}$ &	$\wh\chi^{\pca}_{it}$ &	$\wh\chi^{\capp}_{it}$ &	$\wh\chi^{\sca}_{it}$ &	$\wh\chi^{\cs}_{it}$ 	\\	\hline
200 &	0.2 &	0.5 &	9.05 &	3.3 &	2.93 &	2.66 &	24.41 &	8.56 &	6.3 &	3.58 &	8.64 &	3.18 &	2.86 &	2.7 &	23.12 &	7.94 &	5.96 &	3.52	\\	
&	&	1 &	8.52 &	3.49 &	3.17 &	2.52 &	23.62 &	8.97 &	7.4 &	4.33 &	8.08 &	3.35 &	3.06 &	2.5 &	22.19 &	8.19 &	6.96 &	4.12	\\	
&	&	2 &	6.93 &	3.42 &	3.15 &	2.58 &	18.39 &	7.97 &	7.2 &	5.14 &	6.54 &	3.26 &	3.01 &	2.51 &	17.11 &	7.18 &	6.69 &	4.79	\\	
&	0.5 &	0.5 &	8.35 &	5.71 &	3.93 &	2.67 &	13.37 &	8.82 &	5.67 &	2.91 &	8.09 &	5.55 &	3.84 &	2.7 &	12.85 &	8.36 &	5.44 &	2.93	\\	
&	&	1 &	7.95 &	6.03 &	4.43 &	2.66 &	12.57 &	9.43 &	6.61 &	3.13 &	7.68 &	5.84 &	4.3 &	2.64 &	12.05 &	8.9 &	6.33 &	3.04	\\	
&	&	2 &	5.76 &	4.85 &	3.86 &	2.51 &	7.86 &	6.66 &	5.13 &	2.97 &	5.55 &	4.68 &	3.74 &	2.45 &	7.52 &	6.28 &	4.9 &	2.85	\\	
&	0.9 &	0.5 &	7.78 &	6.86 &	4.88 &	2.63 &	6.95 &	5.97 &	4.27 &	2.46 &	7.57 &	6.69 &	4.78 &	2.66 &	6.75 &	5.77 &	4.15 &	2.52	\\	
&	&	1 &	6.83 &	6.39 &	5.21 &	2.55 &	6.09 &	5.61 &	4.66 &	2.07 &	6.65 &	6.22 &	5.08 &	2.53 &	5.91 &	5.42 &	4.52 &	2.05	\\	
&	&	2 &	3.24 &	3.14 &	2.85 &	1.66 &	2.6 &	2.53 &	2.31 &	1.32 &	3.15 &	3.05 &	2.77 &	1.63 &	2.53 &	2.45 &	2.24 &	1.28	\\	
&	1 &	0.5 &	7.84 &	4.06 &	3.13 &	2.75 &	14.19 &	5.82 &	4.91 &	3.72 &	7.55 &	3.97 &	3.07 &	2.79 &	13.79 &	5.62 &	4.79 &	3.73	\\	
&	&	1 &	6.85 &	4.02 &	3.29 &	2.39 &	10.2 &	4.96 &	4.62 &	3.25 &	6.59 &	3.92 &	3.22 &	2.39 &	9.9 &	4.79 &	4.5 &	3.18	\\	
&	&	2 &	3.75 &	2.64 &	2.28 &	1.64 &	3.38 &	2.26 &	2.12 &	1.64 &	3.59 &	2.57 &	2.22 &	1.61 &	3.27 &	2.17 &	2.06 &	1.59	\\	\hline
500 &	0.2 &	0.5 &	7.76 &	3.01 &	2.33 &	3.11 &	30.56 &	8.62 &	5.95 &	4.51 &	7.18 &	2.85 &	2.28 &	3.25 &	27.75 &	7.59 &	5.43 &	4.66	\\	
&	&	1 &	7.66 &	3.3 &	2.55 &	2.47 &	30.42 &	9.39 &	7.42 &	3.84 &	7.04 &	3.1 &	2.46 &	2.54 &	27.43 &	8.22 &	6.72 &	3.69	\\	
&	&	2 &	6.99 &	3.52 &	2.74 &	2.3 &	27.38 &	9.94 &	8.5 &	4.96 &	6.4 &	3.28 &	2.61 &	2.31 &	24.53 &	8.58 &	7.62 &	4.51	\\	
&	0.5 &	0.5 &	8.39 &	5.62 &	3.26 &	3.14 &	18.6 &	9.61 &	5.7 &	4.47 &	7.84 &	5.3 &	3.15 &	3.27 &	17.33 &	8.84 &	5.33 &	4.65	\\	
&	&	1 &	8.56 &	6.32 &	3.88 &	2.6 &	19.03 &	10.69 &	7.35 &	3.39 &	7.98 &	5.93 &	3.7 &	2.67 &	17.7 &	9.81 &	6.85 &	3.39	\\	
&	&	2 &	8.32 &	6.8 &	4.53 &	2.65 &	18.19 &	11.42 &	8.91 &	4.37 &	7.73 &	6.35 &	4.28 &	2.63 &	16.87 &	10.45 &	8.28 &	4.09	\\	
&	0.9 &	0.5 &	8.31 &	7.1 &	4.09 &	3.14 &	11.34 &	8.27 &	4.9 &	4.46 &	7.81 &	6.7 &	3.93 &	3.27 &	10.65 &	7.69 &	4.62 &	4.66	\\	
&	&	1 &	8.46 &	7.68 &	5 &	2.63 &	11.63 &	8.98 &	6.31 &	3.04 &	7.94 &	7.23 &	4.77 &	2.7 &	10.91 &	8.35 &	5.94 &	3.16	\\	
&	&	2 &	8.14 &	7.75 &	5.89 &	2.71 &	11.14 &	9.27 &	7.61 &	3.1 &	7.62 &	7.27 &	5.58 &	2.69 &	10.44 &	8.59 &	7.16 &	2.98	\\	
&	1 &	0.5 &	11.75 &	4.1 &	2.95 &	3.4 &	38.32 &	8 &	6.71 &	5.43 &	10.79 &	3.91 &	2.87 &	3.52 &	35.86 &	7.15 &	6.32 &	5.47	\\	
&	&	1 &	11.61 &	4.49 &	3.46 &	2.91 &	37.26 &	8.52 &	8.54 &	5.91 &	10.61 &	4.27 &	3.33 &	2.94 &	34.64 &	7.59 &	8.01 &	5.64	\\	
&	&	2 &	10.85 &	4.89 &	4.03 &	3 &	31.28 &	8.57 &	9.56 &	7.5 &	9.87 &	4.61 &	3.84 &	2.93 &	28.88 &	7.59 &	8.92 &	6.97	\\	\hline
1000 &	0.2 &	0.5 &	4.95 &	2.19 &	1.72 &	3.56 &	25.29 &	6.18 &	4.06 &	5.74 &	4.45 &	2.08 &	1.73 &	3.8 &	21.42 &	5.15 &	3.54 &	6.02	\\	
&	&	1 &	4.59 &	2.25 &	1.73 &	2.47 &	23.92 &	6.09 &	4.89 &	3.57 &	4.12 &	2.12 &	1.72 &	2.65 &	20.22 &	5.04 &	4.24 &	3.75	\\	
&	&	2 &	4.15 &	2.32 &	1.79 &	1.93 &	22.21 &	6.27 &	5.91 &	3.14 &	3.72 &	2.18 &	1.76 &	2.06 &	18.73 &	5.17 &	5.09 &	2.94	\\	
&	0.5 &	0.5 &	5.59 &	3.95 &	2.25 &	3.57 &	15.03 &	8.13 &	4.05 &	6.01 &	5.04 &	3.63 &	2.2 &	3.81 &	13.02 &	6.97 &	3.58 &	6.32	\\	
&	&	1 &	5.56 &	4.28 &	2.49 &	2.54 &	15.21 &	8.79 &	5.13 &	3.71 &	5.01 &	3.91 &	2.39 &	2.71 &	13.12 &	7.47 &	4.49 &	3.95	\\	
&	&	2 &	5.5 &	4.62 &	2.85 &	2.1 &	15.1 &	9.77 &	6.43 &	2.77 &	4.94 &	4.19 &	2.68 &	2.2 &	12.98 &	8.23 &	5.59 &	2.80	\\	
&	0.9 &	0.5 &	5.84 &	5.03 &	2.82 &	3.59 &	10.5 &	7.79 &	3.99 &	6.08 &	5.28 &	4.59 &	2.69 &	3.83 &	8.81 &	6.47 &	3.43 &	6.41	\\	
&	&	1 &	5.86 &	5.35 &	3.28 &	2.56 &	10.54 &	8.22 &	5.07 &	3.75 &	5.3 &	4.86 &	3.09 &	2.73 &	8.88 &	6.77 &	4.31 &	4.01	\\	
&	&	2 &	5.81 &	5.55 &	3.89 &	2.15 &	10.42 &	8.77 &	6.39 &	2.59 &	5.3 &	5.08 &	3.72 &	2.39 &	8.63 &	7.04 &	5.37 &	2.82	\\	
&	1 &	0.5 &	10.25 &	2.77 &	2.17 &	3.72 &	57.54 &	6.7 &	6.12 &	5.51 &	8.91 &	2.57 &	2.12 &	3.93 &	50.79 &	5.4 &	5.47 &	5.55	\\	
&	&	1 &	9.44 &	2.79 &	2.32 &	2.78 &	54.97 &	6.58 &	7.69 &	5.51 &	8.16 &	2.57 &	2.23 &	2.89 &	47.79 &	5.29 &	6.8 &	5.08	\\	
&	&	2 &	8.46 &	2.84 &	2.53 &	2.43 &	49.91 &	6.83 &	9.4 &	7.4 &	7.28 &	2.6 &	2.38 &	2.46 &	42.66 &	5.44 &	8.15 &	6.45	\\	\hline\hline
\end{tabular}}
\end{table}

\begin{table}[hp]
\setlength{\tabcolsep}{3pt}
\centering
\caption{Summary of relative estimation error from various estimators for 
$T = 2000$,
$n \in \{200, 500, 1000\}$,
$\varrho \in \{0.2, 0.5, 0.9, 1\}$
and $\phi \in \{0.5, 1, 2\}$.}
\label{table:sim:err:2000}
{\footnotesize
\begin{tabular}{ccc|cccc|cccc|cccc|cccc}
\hline\hline
&	&	&	\multicolumn{8}{c}{all} &								\multicolumn{8}{c}{block} 								\\	
&	&	&	\multicolumn{4}{c}{$\text{err}_{\avg}$} &				\multicolumn{4}{c}{$\text{err}_{\max}$} &				\multicolumn{4}{c}{$\text{err}_{\avg}$} &				\multicolumn{4}{c}{$\text{err}_{\max}$}				\\	
$n$ &	$s$ &	$\phi$ &	$\wh\chi^{\pca}_{it}$ &	$\wh\chi^{\capp}_{it}$ &	$\wh\chi^{\sca}_{it}$ &	$\wh\chi^{\cs}_{it}$ &	$\wh\chi^{\pca}_{it}$ &	$\wh\chi^{\capp}_{it}$ &	$\wh\chi^{\sca}_{it}$ &	$\wh\chi^{\cs}_{it}$ &	$\wh\chi^{\pca}_{it}$ &	$\wh\chi^{\capp}_{it}$ &	$\wh\chi^{\sca}_{it}$ &	$\wh\chi^{\cs}_{it}$ &	$\wh\chi^{\pca}_{it}$ &	$\wh\chi^{\capp}_{it}$ &	$\wh\chi^{\sca}_{it}$ &	$\wh\chi^{\cs}_{it}$ 	\\	\hline
200 &	0.2 &	0.5 &	10.17 &	3.46 &	3.15 &	2.81 &	32.69 &	8.97 &	7.91 &	4.15 &	9.96 &	3.37 &	3.05 &	2.84 &	30.62 &	8.85 &	7.2 &	3.88	\\	
&	&	1 &	10.34 &	3.87 &	3.56 &	2.76 &	33.55 &	10.3 &	9.94 &	5.9 &	10.1 &	3.76 &	3.36 &	2.77 &	31.08 &	10.18 &	8.56 &	5.20	\\	
&	&	2 &	9 &	3.9 &	3.62 &	2.99 &	29.14 &	10.21 &	10.64 &	7.92 &	8.85 &	3.87 &	3.44 &	3.01 &	25.93 &	10 &	8.56 &	6.69	\\	
&	0.5 &	0.5 &	9.05 &	6.06 &	4.15 &	2.79 &	16.35 &	9.42 &	6.5 &	3.12 &	8.9 &	5.99 &	4.11 &	2.8 &	15.81 &	9.2 &	6.33 &	3.10	\\	
&	&	1 &	8.74 &	6.52 &	4.78 &	2.83 &	15.83 &	10.32 &	7.89 &	3.93 &	8.59 &	6.42 &	4.72 &	2.82 &	15.28 &	10.05 &	7.65 &	3.81	\\	
&	&	2 &	6.73 &	5.59 &	4.47 &	2.86 &	10.61 &	8.09 &	6.66 &	4.14 &	6.61 &	5.5 &	4.4 &	2.82 &	10.23 &	7.83 &	6.46 &	4.01	\\	
&	0.9 &	0.5 &	8.42 &	7.26 &	5.03 &	2.76 &	8.39 &	6.97 &	4.79 &	2.53 &	8.32 &	7.18 &	4.97 &	2.77 &	8.38 &	6.9 &	4.75 &	2.55	\\	
&	&	1 &	7.54 &	6.93 &	5.43 &	2.73 &	7.54 &	6.85 &	5.34 &	2.46 &	7.45 &	6.84 &	5.36 &	2.72 &	7.54 &	6.8 &	5.31 &	2.45	\\	
&	&	2 &	4 &	3.85 &	3.41 &	2 &	3.63 &	3.49 &	3.1 &	1.83 &	3.94 &	3.79 &	3.37 &	1.98 &	3.63 &	3.48 &	3.09 &	1.82	\\	
&	1 &	0.5 &	8.76 &	4.3 &	3.31 &	2.9 &	14.85 &	5.97 &	4.97 &	3.78 &	8.71 &	4.3 &	3.32 &	2.95 &	14.81 &	5.88 &	5.01 &	3.87	\\	
&	&	1 &	7.62 &	4.21 &	3.45 &	2.51 &	10.7 &	5 &	4.65 &	3.3 &	7.57 &	4.22 &	3.46 &	2.53 &	10.66 &	4.96 &	4.67 &	3.33	\\	
&	&	2 &	4.41 &	2.93 &	2.52 &	1.83 &	3.94 &	2.47 &	2.35 &	1.85 &	4.37 &	2.93 &	2.54 &	1.84 &	3.91 &	2.46 &	2.34 &	1.85	\\	\hline
500 &	0.2 &	0.5 &	10.05 &	3.51 &	2.7 &	3.43 &	42.97 &	10 &	7.06 &	4.78 &	9.63 &	3.37 &	2.64 &	3.49 &	40.5 &	9.39 &	6.73 &	4.85	\\	
&	&	1 &	10.43 &	4.05 &	3.12 &	2.79 &	45.25 &	11.5 &	9.25 &	4.63 &	9.96 &	3.88 &	3.03 &	2.81 &	42.41 &	10.75 &	8.76 &	4.48	\\	
&	&	2 &	10.09 &	4.55 &	3.53 &	2.73 &	45.49 &	13.21 &	11.64 &	6.89 &	9.6 &	4.34 &	3.4 &	2.71 &	42.41 &	12.28 &	10.93 &	6.51	\\	
&	0.5 &	0.5 &	9.83 &	6.47 &	3.6 &	3.44 &	23.84 &	11.08 &	7.07 &	4.59 &	9.49 &	6.27 &	3.52 &	3.5 &	23.1 &	10.63 &	6.82 &	4.68	\\	
&	&	1 &	10.17 &	7.41 &	4.38 &	2.89 &	24.93 &	12.71 &	9.61 &	3.73 &	9.81 &	7.16 &	4.26 &	2.91 &	24.13 &	12.18 &	9.26 &	3.68	\\	
&	&	2 &	10.06 &	8.12 &	5.23 &	3.01 &	24.71 &	14.39 &	12.29 &	5.45 &	9.68 &	7.83 &	5.06 &	2.99 &	23.88 &	13.8 &	11.83 &	5.23	\\	
&	0.9 &	0.5 &	9.73 &	8.3 &	4.64 &	3.45 &	15.42 &	10.4 &	6.55 &	4.5 &	9.42 &	8.05 &	4.53 &	3.51 &	15.02 &	10.03 &	6.37 &	4.61	\\	
&	&	1 &	10 &	9.09 &	5.74 &	2.93 &	16 &	11.51 &	8.59 &	3.2 &	9.68 &	8.8 &	5.59 &	2.95 &	15.57 &	11.09 &	8.36 &	3.24	\\	
&	&	2 &	9.71 &	9.25 &	6.78 &	3.08 &	15.58 &	12.12 &	10.47 &	3.9 &	9.39 &	8.94 &	6.58 &	3.05 &	15.16 &	11.65 &	10.19 &	3.84	\\	
&	1 &	0.5 &	16.36 &	5.08 &	3.53 &	3.83 &	44.19 &	8.93 &	7.32 &	5.84 &	15.64 &	4.86 &	3.45 &	3.88 &	43.77 &	8.2 &	7.18 &	5.98	\\	
&	&	1 &	16.58 &	5.75 &	4.34 &	3.39 &	44.9 &	10.01 &	9.8 &	6.78 &	15.81 &	5.46 &	4.18 &	3.39 &	44.36 &	9.06 &	9.5 &	6.76	\\	
&	&	2 &	15.92 &	6.44 &	5.25 &	3.7 &	40.86 &	10.65 &	11.76 &	9.47 &	15.16 &	6.12 &	5.04 &	3.64 &	40.27 &	9.79 &	11.45 &	9.34	\\	\hline
1000 &	0.2 &	0.5 &	8.12 &	3.08 &	2.14 &	4.28 &	42.1 &	9.51 &	5.97 &	7.04 &	8.37 &	5.65 &	2.98 &	4.44 &	18.12 &	9.38 &	4.78 &	7.20	\\	
&	&	1 &	8.28 &	3.46 &	2.32 &	3.01 &	43.77 &	10.58 &	7.82 &	4.59 &	7.72 &	3.25 &	2.25 &	3.1 &	40.48 &	9.57 &	7.29 &	4.69	\\	
&	&	2 &	8.37 &	4.01 &	2.64 &	2.49 &	44.5 &	12.48 &	10.11 &	4.85 &	7.77 &	3.75 &	2.52 &	2.52 &	40.97 &	11.28 &	9.36 &	4.65	\\	
&	0.5 &	0.5 &	8.89 &	5.98 &	3.07 &	4.3 &	19.6 &	10.23 &	5.15 &	7.07 &	8.34 &	5.64 &	2.97 &	4.43 &	18.28 &	9.46 &	4.82 &	7.26	\\	
&	&	1 &	9.15 &	6.77 &	3.66 &	3.08 &	20.48 &	11.46 &	6.87 &	4.47 &	8.6 &	6.48 &	3.58 &	3.17 &	16.85 &	10.54 &	5.86 &	4.61	\\	
&	&	2 &	9.17 &	7.5 &	4.41 &	2.64 &	20.84 &	12.92 &	8.94 &	3.38 &	8.59 &	7.13 &	4.29 &	2.68 &	16.95 &	11.73 &	7.59 &	3.33	\\	
&	0.9 &	0.5 &	8.98 &	7.68 &	3.99 &	4.3 &	12.28 &	9.31 &	4.75 &	7.06 &	8.45 &	7.31 &	3.92 &	4.43 &	12.78 &	9.63 &	4.98 &	7.26	\\	
&	&	1 &	9.32 &	8.49 &	5.05 &	3.1 &	12.89 &	10.21 &	6.47 &	4.44 &	8.76 &	8.07 &	4.98 &	3.18 &	13.42 &	10.76 &	6.88 &	4.59	\\	
&	&	2 &	9.48 &	9.06 &	6.39 &	2.7 &	13.22 &	11.28 &	8.55 &	3.08 &	8.88 &	8.56 &	6.3 &	2.72 &	13.74 &	11.91 &	9.21 &	3.18	\\	
&	1 &	0.5 &	20.27 &	4.83 &	3.14 &	4.62 &	91.41 &	11.42 &	8.9 &	7.8 &	18.13 &	4.4 &	2.95 &	4.69 &	87.37 &	9.85 &	8.51 &	8.02	\\	
&	&	1 &	20.07 &	5.29 &	3.73 &	3.64 &	91.6 &	12.44 &	11.87 &	8.51 &	17.85 &	4.79 &	3.47 &	3.62 &	86.96 &	10.68 &	11.34 &	8.36	\\	
&	&	2 &	19.31 &	5.92 &	4.53 &	3.53 &	86.72 &	13.92 &	15.22 &	12.13 &	17.05 &	5.3 &	4.15 &	3.4 &	81.76 &	11.84 &	14.41 &	11.63	\\	\hline\hline
\end{tabular}}
\end{table}

\clearpage

\subsection{Sensitivity to the factor number estimators}

The simulation results in Section~\ref{sec:sim} provide a benchmark
for what we could expect in practice where the popularly adopted estimator of $r$ is adopted.
We supplement the results by additional simulation studies under Model~1 in Section~\ref{sec:sim:model}
to verify the observations above, where we investigate $\text{err}_{\max}(\wh\chi^{\circ}_{it})$
of $\wh{\chi}^{\pca}_{it}$, $\wh{\chi}^{\capp}_{it}$, $\wh{\chi}^{\sca}_{it}$ and $\wh{\chi}^{\cs}_{it}$ 
obtained with $\wh r \in \{r, r + 1, r + 5\}$.
Figure~\ref{fig:m1:err} shows that even when $\wh r$ is only slightly larger than $r$,
the PC estimator incurs much larger errors than those returned by the proposed modified PC estimators,
and the gap tends to increase with $n$.
On the other hand, when $\wh r = r$, all 
$\wh{\chi}^{\capp}_{it}$, $\wh{\chi}^{\sca}_{it}$ and $\wh{\chi}^{\cs}_{it}$
are shown to perform as well as the oracle PC estimator.

\begin{figure}[htbp]
\centering
\includegraphics[width=1\textwidth]{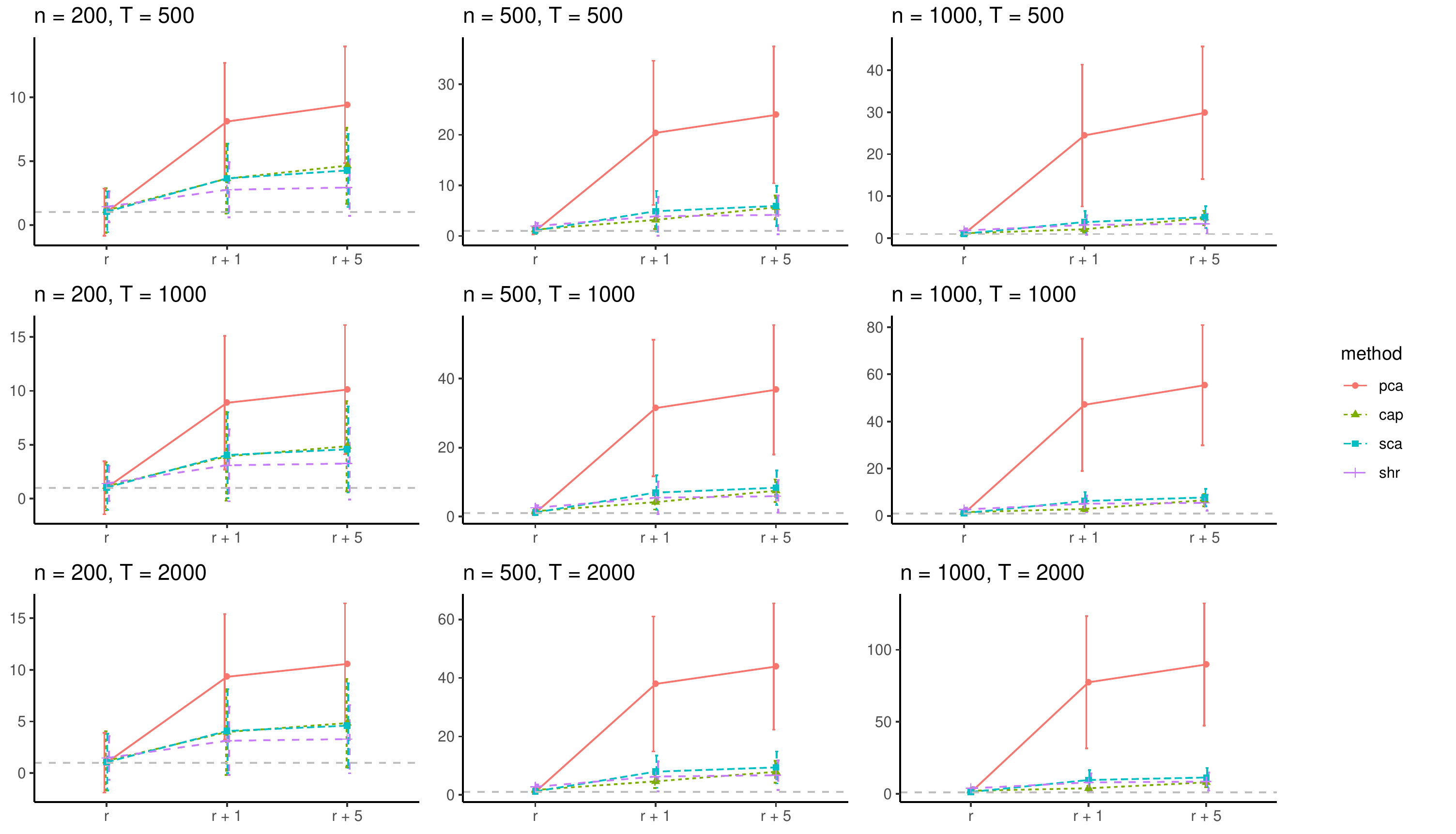}
\caption{\small $\text{err}_{\max}(\wh\chi^{\circ}_{it})$
of $\wh{\chi}^{\pca}_{it}$, $\wh{\chi}^{\capp}_{it}$, $\wh{\chi}^{\sca}_{it}$
and $\wh{\chi}^{\cs}_{it}$ 
averaged over $1000$ realisations generated
with $\phi = 1$, $n \in \{200, 500, 1000\}$ (left to right), $T \in \{500, 1000, 2000\}$ (top to bottom)
and $\wh r \in \{r, r + 1, r + 5\}$ (left to right within each plot) under Model~1;
the vertical errors bars represent the standard deviations
and the broken horizontal lines indicate where $\text{err}_{\max}(\wh\chi^{\circ}_{it}) = 1$.}
\label{fig:m1:err}
\end{figure}

\subsection{Sensitivity to the choice of $c_w$}
\label{sec:cw}

We have performed simulation studies 
in order to systematically assess the influence of the choice of constant $c_w$
on the estimation errors of the resultant estimators.
Under Model~1 described in Section~\ref{sec:sim:model},
we examine the blockwise scaled PC estimator $\wh{\chi}^{\bsca}_{it}(c_w)$ 
over a grid of $c_w \in \mc C = \{0.8, 0.9, \ldots, 1.5\} \times \sqrt n \; \max_{1 \le i \le n} |\wh w_{x, i1}|$,
where we obtain $\wh r$ as in~\eqref{eq:r:bn} and set the number of blocks $L_T = 5$.
Noticing that the blockwise estimation provides a natural means for constructing the cross-validation measure,
we have also recorded
$c^*= \arg\min_{c_w \in \mc C} \wt{\text{err}}_{\max}(c_w)$
for each realisation,
where $\wt{\text{err}}_{\max}(c_w) =
\max_{1 \le i \le n} T^{-1} \sum_{t = 1}^T (\wh\chi^{\bsca}_{it}(c_w) - x_{it})^2$.
We choose this error measure as 
the average of $(\wh\chi^{\bsca}_{it}(c_w) - x_{it})^2$ over $i$ and $t$
is invariably minimised by the largest $c_w$ in consideration,
which tends to over-fit the data.
Varying $n \in \{200, 500, 1000\}$ and $T \in \{500, 1000, 2000\}$ while fixing
the noise-to-signal ratio at $\phi = 1$,
we report the estimation error
$\overline{\text{err}}_{\max}(c_w) =
\max_{1 \le i \le n} T^{-1} \sum_{t = 1}^T (\wh\chi^{\bsca}_{it}(c_w) - \chi_{it})^2$
for $c_w \in \mc C \cup \{c^*\}$
over $1000$ realisations for each setting, see Figure~\ref{fig:cv:err} in Section~\ref{sec:scale}.
It is clear that the variability of $\overline{\text{err}}_{\max}(c_w)$ increases with $c_w$,
while its average is minimised when the multiplicative constant is set at $1$ for most $n$ and $T$;
the cross-validated choice $c^*$ does not outperform this choice
in terms of either variability or the average of the estimation error.
The default choice we recommend in the paper, $c_w = 1.1 \times \sqrt n \; \max_{1 \le i \le n} |\wh w_{x, i1}|$,
yields performance not far from that with the multiplicative constant $1$ and overall,  we conclude that
the performance of the scaled PC estimator is robust to a range of values for $c_w$ within reason.

\end{document}